\documentclass[a4paper,onecolumn,11pt]{quantumarticle}
\pdfoutput=1
\usepackage[utf8]{inputenc}
\usepackage[english]{babel}
\usepackage[T1]{fontenc}
\usepackage{amsmath}
\usepackage{hyperref}

\usepackage{tikz}
\usepackage{lipsum}

\usepackage[T1]{fontenc}
\usepackage{geometry}

\usepackage{amsmath,amssymb,amsthm,mathtools,bm}

\usepackage{physics}

\usepackage{graphicx}   % \includegraphics
\usepackage{dcolumn}    % decimal-aligned columns (optional)

\usepackage{enumitem}
\setlist[itemize]{leftmargin=2em}
\setlist[enumerate]{leftmargin=2em}

\usepackage{float}      % for [H]

\usepackage{pgfplots}
\pgfplotsset{compat=1.18}

\usepackage{tikz-3dplot}

\usepackage{tikz} 

\usepackage{hyperref}
\hypersetup{colorlinks=true,linkcolor=black,citecolor=blue,urlcolor=blue}

\numberwithin{equation}{section} % Equation numbering: 5.1, 5.2, ...
\theoremstyle{plain}
\newtheorem{theorem}{Theorem}[section]
\newtheorem{lemma}[theorem]{Lemma}
\newtheorem{corollary}[theorem]{Corollary}

\theoremstyle{definition}

\theoremstyle{remark}

\usepackage{hyperref}

\makeatletter
\renewcommand{\thetheorem}{\arabic{section}.\arabic{theorem}}

\makeatother

\usepackage{algpseudocode}
\makeatletter
\newcounter{algorithm}                   % counter for Algorithm
\renewcommand{\thealgorithm}{\arabic{algorithm}}

\def\fps@algorithm{tbp}
\def\ftype@algorithm{4}                  % a distinct float type (figure=1, table=2)
\def\ext@algorithm{loa}                  % list-of-algorithms file extension
\def\fnum@algorithm{Algorithm~\thealgorithm}

\newenvironment{algorithm}{\@float{algorithm}}{\end@float}
\newenvironment{algorithm*}{\@dblfloat{algorithm}}{\end@dblfloat}
\makeatother

\usepackage{float}
\floatstyle{ruled}
\restylefloat{algorithm}

\algrenewcommand\algorithmicrequire{\textbf{Input:}}
\algrenewcommand\algorithmicensure{\textbf{Output:}}

\usepackage{float}
\floatstyle{ruled}                 % 顶部细横线,类似你截图的样式
\floatname{algorithm}{Algorithm}   % 浮动名

\usepackage{graphicx}   % for \resizebox
\usepackage[caption=false]{subfig}

\let\REQUIRE\Require
\let\ENSURE\Ensure
\let\STATE\State
\let\FORALL\ForAll
\let\RETURN\Return
\let\COMMENT\Comment
\let\ENDFOR\EndFor

\algrenewcommand\algorithmicrequire{\textbf{Input:}}
\algrenewcommand\algorithmicensure{\textbf{Output:}}

\usepackage{amsmath}              
\numberwithin{equation}{section}

\makeatletter
\renewcommand{\theequation}{\arabic{section}.\arabic{equation}}
\makeatother

\begin{document}

\title{A Unified Complexity-Algorithm Account of Constant-Round QAOA Expectation Computation}

\author{Jingheng Wang}
\author{Shengminjie Chen}%
  \email{csmj@ict.ac.cn}
\author{Xiaoming Sun}
\author{Jialin Zhang}
 \affiliation{%
State Key Lab of Processors, Institute of Computing Technology,
Chinese Academy of Sciences, 100190, Beijing, China.
}
 \altaffiliation{School of Computer Science and Technology, University of Chinese Academy of Sciences, Beijing, 100049, China.}

% \author{Lídia del Rio}
% \affiliation{Institute for Theoretical Physics, ETH Zurich, 8093 Zurich, Switzerland}
% \orcid{0000-0002-2445-2701}
% \author{Christian Gogolin}
% \email{latex@quantum-journal.org}
% \homepage{http://quantum-journal.org}
% \orcid{0000-0003-0290-4698}
% \thanks{You can use the \texttt{\textbackslash{}email}, \texttt{\textbackslash{}homepage}, and \texttt{\textbackslash{}thanks} commands to add additional information for the preceding \texttt{\textbackslash{}author}. If applicable, this can also be used to indicate that a work has previously been published in conference proceedings.}
% \affiliation{Covestro Deutschland AG, Kaiser-Wilhelm-Allee 60, 51373 Leverkusen, Germany}
% \author{Marcus Huber}
% \affiliation{Institute for Quantum Optics \& Quantum Information (IQOQI), Austrian Academy of Sciences, Boltzmanngasse 3, Vienna A-1090, Austria}
% \orcid{0000-0003-1985-4623}
% \author{Cassandra Granade}
% \affiliation{Microsoft Research, Quantum Architectures and Computation Group, Redmond, WA 98052, USA}
% \author{Johannes Jakob Meyer}
% \affiliation{Dahlem Center for Complex Quantum Systems, Freie Universität Berlin, 14195 Berlin, Germany}
% \orcid{0000-0003-1533-8015}
% \author{Victor V. Albert}
% \affiliation{Institute for Quantum Information and Matter \& Walter Burke Institute for Theoretical Physics, Caltech, Pasadena, CA 91125, USA}
% \orcid{0000-0002-0335-9508}
\maketitle

% \begin{abstract}
% The Quantum Approximate Optimization Algorithm (QAOA) is widely studied for combinatorial optimization and has achieved significant advances both in theoretical guarantees and practicality, yet for general combinatorial optimization problems, the expected performance and classical simulability of fixed-round QAOA remain unclear. First, we show that exactly evaluating the expectation of fixed-round QAOA is NP-hard, and that even approximating it within inverse-exponential additive error is also NP-hard. To evaluate the expected performance of QAOA, we propose a dynamic programming algorithm leveraging tree decomposition. As a byproduct, when the \(p\)-local treewidth grows at most logarithmically with the number of vertices, this yields a polynomial-time \emph{exact} evaluation algorithm in the graph size $n$. Beyond Max-Cut, we extend the framework to general Binary Unconstrained Combinatorial Optimization (BUCO). Finally, we provide reproducible evaluations for rounds up to \(p=3\) on representative structured families, including generalized Petersen graph $GP(15,2)$, double-layer triangular 2-lifts, and truncated icosahedron graph $C_{60}$, and report cut ratios while benchmarking against locality-matched classical baselines.
% \end{abstract}
\begin{abstract}
The Quantum Approximate Optimization Algorithm (QAOA) is widely studied for combinatorial
optimization and has achieved significant advances both in theoretical guarantees and
practical performance, yet for general combinatorial optimization problems the expected
performance and classical simulability of fixed-round QAOA remain unclear. Focusing on
Max-Cut, we first show that for general graphs and any fixed round $p\ge2$, exactly
evaluating the expectation of fixed-round QAOA at prescribed angles is $\mathrm{NP}$-hard,
and that approximating this expectation within additive error $2^{-O(n)}$ in the number
$n$ of vertices is already $\mathrm{NP}$-hard. To evaluate the expected performance of QAOA, we propose a dynamic
programming algorithm leveraging tree decomposition. As a byproduct, when the $p$-local
treewidth grows at most logarithmically with the number of vertices, this yields a
polynomial-time \emph{exact} evaluation algorithm in the graph size $n$. Beyond Max-Cut,
we extend the framework to general Binary Unconstrained Combinatorial Optimization (BUCO).
Finally, we provide reproducible evaluations for rounds up to $p=3$ on representative
structured families, including the generalized Petersen graph $GP(15,2)$, double-layer
triangular 2-lifts, and the truncated icosahedron graph $C_{60}$, and report cut ratios
while benchmarking against locality-matched classical baselines.
\end{abstract}

\section{Introduction}
% The Quantum Approximate Optimization Algorithm (QAOA)~\cite{Farhi2014QAOA}, as a leading variational approach for combinatorial optimization, has attracted broad attention and has been applied to a range of problems, including Max-Cut and related constraint satisfaction tasks~\cite{Farhi2014QAOA,Hadfield2019QAOAOperator}. Building on these applications, many efforts aim to improve QAOA’s empirical efficiency and solution quality via problem-tailored mixers and the Quantum Alternating Operator Ansatz~\cite{Hadfield2019QAOAOperator}, warm-start strategies (including relaxations- or problem-informed initialization)~\cite{Egger2021WarmStart,Tate2023Warmest}, annealing-based initialization~\cite{Sack2021QAInit}, adaptive and counterdiabatic variants~\cite{Zhu2022ADAPTQAOA,Chandarana2022DCQAOA}, parameter-setting schemes for weighted instances~\cite{Sureshbabu2024ParamWeighted}, and hardware-scale studies that evaluate training overheads and noise sensitivity~\cite{Weidenfeller2022Scaling}. However, the assessment of these enhancements largely relies on numerical benchmarks; rigorous, instance-wide performance guarantees on general inputs remain limited, so QAOA is still commonly regarded as a heuristic quantum algorithm.

The Quantum Approximate Optimization Algorithm (QAOA)~\cite{Farhi2014QAOA}, as a leading variational approach for combinatorial optimization, has attracted broad attention and has been applied to a range of problems, including Max-Cut and related constraint satisfaction tasks~\cite{Farhi2014QAOA,Hadfield2019QAOAOperator,Zhou2020QAOAPRX,Blekos2024QAOAReview}. Building on these applications, many efforts aim to improve QAOA’s empirical efficiency and solution quality via problem-tailored mixers and the Quantum Alternating Operator Ansatz~\cite{Hadfield2019QAOAOperator}, warm-start strategies~\cite{Egger2021WarmStart,Tate2023Warmest}, annealing-based initialization~\cite{Sack2021QAInit}, adaptive and counterdiabatic variants~\cite{Zhu2022ADAPTQAOA,Chandarana2022DCQAOA}, parameter-setting schemes for weighted instances~\cite{Sureshbabu2024ParamWeighted}, and hardware-scale studies that quantify execution-time and noise constraints for QAOA~\cite{Weidenfeller2022Scaling}. However, the assessment of these enhancements still relies predominantly on numerical benchmarks; beyond a few settings where guarantees can be inherited from classical relaxations or are proved for restricted graph families, rigorous, instance-wide performance bounds on general inputs remain scarce, so QAOA is still most often regarded as a heuristic quantum algorithm.

To obtain theoretical guarantees for QAOA, a basic objective is to understand the
\emph{expected} performance of \emph{constant-round} QAOA, treating the number of
alternating phase--mixer rounds $p$ as a fixed constant. A central role is played by the QAOA
expectation
\[
  \langle C\rangle=\langle\psi_p(\bm{\gamma},\bm{\beta})|C|\psi_p(\bm{\gamma},\bm{\beta})\rangle,
\]
which is both the objective optimized by classical parameter-tuning routines and the quantity whose value underpins approximation guarantees.
In this work, we primarily focus on the Max-Cut objective, where $C$ is the
standard cut Hamiltonian on an input graph.

Farhi \emph{et al.}~\cite{Farhi2014QAOA} showed that one-round QAOA attains an
approximation ratio of at least $0.6924$ for Max-Cut on $3$-regular graphs. Building
on this, Wurtz and Love~\cite{WurtzLove2021PRA} proved worst-case approximation ratios
of at least $0.7559$ for $p=2$ and $0.7924$ for $p=3$ on $3$-regular graphs, the
latter under a conjecture that graphs with no ``visible'' cycles are worst case. On
sparse, locally tree-like inputs, constant-round QAOA has been analyzed from several
angles : Basso \emph{et al.}~\cite{Basso2022FOCS,Basso2022TQC} studied its performance
on large sparse hypergraphs, mixed spin-glass models, and large-girth regular graphs,
identifying both average-case limitations at fixed round and favorable high-round
behavior on large-girth regular graphs and the Sherrington--Kirkpatrick model;
Marwaha~\cite{Marwaha2021LocalVsQAOA} and Hastings~\cite{Hastings2019BoundedDepth}
showed that suitable local classical algorithms can match or even outperform low-round
QAOA on high-girth regular graphs and on problems such as MAX-3-LIN-2 and triangle-free
Max-Cut. More recently, Li \emph{et al.}~\cite{LiSuYangZhang2024} moved beyond the
high-girth regime by deriving an exact iterative formula for the expected cut
fraction of (multi-angle) QAOA on certain low-girth expander graphs,
and provided numerical evidence that constant-round QAOA and its variants can
outperform the best-known classical local algorithms on these families.

Viewed through the lens of the expectation $\langle C\rangle$, many of
these works share a common structure: they either derive explicit formulas or iterative
schemes for $\langle C\rangle$ on particular graph families, or else bound it directly, and then
translate these evaluations into approximation-ratio guarantees or comparisons with
classical algorithms. From this perspective, understanding the performance of constant-round
QAOA naturally leads to the classical problem of evaluating or approximating $\langle C\rangle$ as
a function of the input instance and its structure.

Despite important progress in the theoretical performance analysis of QAOA, two
fundamental questions remain open:

\begin{itemize}
    %  \item \textbf{Structural solvability landscape.} Positive results for expectation evaluation are scattered across specific graph families and have not yet crystallized into a parameter-driven picture that explains evaluability.
     
     \item \textbf{Structural solvability landscape.}
On the algorithmic side, positive results for expectation evaluation are mostly known for
graphs with additional structure, typically treated on a case-by-case basis using ad hoc
combinatorial arguments or dynamic-programming techniques. These tractability results are
scattered across particular graph families and are not organized by a unifying structural
parameter such as treewidth, local treewidth (cf.~\cite{Eppstein2000LocalTW,DemaineHajiaghayi2004LocalTW,DemaineFominHajiaghayiThilikos2005HMinorFree}), or clique-/rank-width (cf.~\cite{OumSeymour2006RankWidth}), so they do not yet provide a
parameter-driven picture that systematically explains when QAOA expectation evaluation
should be feasible and when it should be intractable.

    %  via explicit structural measures (e.g., treewidth or local treewidth) together with the number of rounds~\cite{Eppstein2000LocalTW,DemaineHajiaghayi2004LocalTW}. Ongoing advances in tree-decomposition algorithms further motivate a structure-focused approach~\cite{Bodlaender1996LinearTW,Korhonen2021TwoApproxTW}.
    % \item \textbf{Expectation vs.\ sampling.} Much of the existing hardness evidence concerns \emph{sampling} tasks or estimating specific output probabilities. While such results clarify complexity limits of near-term circuits, they do not directly speak to the quantity actually optimized in variational practice—the \emph{expected} objective value. Consequently, the worst-case complexity of computing the expectation at prescribed angles on general graphs, especially for $p\ge 2$, remains largely unresolved.
    
    \item \textbf{Expectation vs.\ sampling.}
Within the QAOA literature, most rigorous complexity-theoretic results concern
\emph{sampling} tasks or individual output probabilities rather than expectation
values. For example, even one-round QAOA can generate output distributions that are
classically intractable to sample from under standard assumptions~\cite{FarhiHarrow2016QAOASupremacy},
and there are average-case hardness results for approximating output probabilities of
random one-round QAOA circuits within additive error $2^{-O(n)}$~\cite{Krovi2022}.
By contrast, despite the extensive work analyzing the approximation guarantees of
constant-round QAOA on various graph and spin-glass ensembles (see the performance
results discussed above), the classical complexity of evaluating the QAOA expectation at
prescribed angles on general graphs for round $p\ge 2$ remains largely unresolved.

\end{itemize}

In this work, we try to answer the above two questions through a unified \emph{complexity--algorithm} account of \emph{constant-round} QAOA \emph{expectation} computation. The main contributions are summarized as follows:
\begin{enumerate}
% \item \textbf{Hardness on general graphs for round $p\ge2$.} For Max-Cut of graphs and any round $p\ge2$, exactly evaluating the QAOA expectation at prescribed angles is $\mathrm{NP}$-hard; moreover, achieving inverse-exponential additive accuracy remains $\mathrm{NP}$-hard. This directly targets the quantity optimized in practice and complements earlier results, which establish hardness of estimating output probabilities for single-round QAOA but do not address expectation values.
\item \textbf{Hardness on general graphs for round $p\ge2$.}
For Max-Cut on graphs and any round $p\ge2$, exactly evaluating the QAOA expectation
at prescribed angles is $\mathrm{NP}$-hard, and even achieving additive
error at most $c\,2^{-2n}$ for some constant $c>0$ is already $\mathrm{NP}$-hard under this reduction.
This directly targets the quantity optimized in practice and complements earlier
results, which establish hardness of estimating output probabilities for single-round
QAOA but do not address expectation values.

\item \textbf{Algorithmic–structural tractability and generalization.} We design an \emph{iterative expectation evaluator} for constant $p$ with total running time $\exp\!\bigl(O(p\cdot \mathrm{ltw}_p(G))\bigr)\cdot \mathrm{poly}(n)$, where $\mathrm{ltw}_p(G)$ measures the treewidth of $p$-local graphs of G. On graph families whose local treewidth grows at most logarithmically in $n$, this yields \emph{polynomial-time exact} expectations. The framework further \emph{generalizes beyond Max-Cut} to BUCO via pseudo-Boolean expansions, yielding parallel evaluability guarantees in these settings.

\item \textbf{Reproducible experiments and locality-matched baselines.} We implement the $p$-local evaluator and report normalized expected cut ratios $\langle C\rangle/|E|$ for rounds up to $p=3$ on three structured families (e.g., generalized petersen graph $GP(15,2)$, double-layer triangular $2$-lifts, and truncated icosahedron graph $C_{60}$). To ensure fair comparison, we include classical $k$-local baselines matched to QAOA’s lightcone radius (e.g., multi-step threshold-flip rules and Barak–Marwaha–type Gaussian-sum rules), enabling reproducible, like-for-like benchmarking of low round QAOA versus classical locality.

\end{enumerate}

Organization. Sec.~\ref{sec:preliminaries} fixes notation and structural parameters (including local treewidth). Sec.~\ref{QAOAhard} proves $\mathrm{NP}$-hardness of
exact evaluation and shows that, for general graphs and any fixed round $p\ge2$,
approximating the expectation within additive error $2^{-O(n)}$ in the number of
vertices $n$ is already $\mathrm{NP}$-hard. Sec.~\ref{QAOADP} presents the evaluator and its complexity guarantees and treats BUCO. Sec.~\ref{sec:experiment} reports experiments on three structured families (Generalized Petersen graph $GP(15,2)$, double-layer triangular $2$-lifts, and truncated icosahedron graph $C_{60}$), followed by conclusions in Sec.~\ref{sec:conclusion}.

\section{PRELIMINARIES}\label{sec:preliminaries}

This section introduces the basic notions of QAOA, together with auxiliary definitions needed in our proofs and algorithms. Throughout, we write $[n]=\{1,\dots,n\}$.

\subsection{Definitions in Graph Theory}\label{sec:2.1}

%\textcolor{red}{
Let $G=(V,E)$ be a simple undirected graph with $n=|V|$ and $m=|E|$. For each node $v\in V$, let $N(v)$ denote the neighborhood set and $\deg(v)=|N(v)|$. For an edge $e=(u,v)\in E$ and integer $p\ge 0$, the $p$-hop vertex neighborhood represents the node set whose distance from $u$ and $v$ is not greater than $p$, i.e.,
%}
\begin{equation}\label{eq:2.1}
N_p(e)=\Bigl\{\, w\in V \ \Bigm|\ \min\!\bigl(\operatorname{dist}_G(u,w),\,\operatorname{dist}_G(v,w)\bigr)\le p \Bigr\}.
\end{equation}
%\textcolor{red}{
where $\operatorname{dist}_G(u,w)$ is distance between the node $u$ and the node $w$. Normally, this distance is measured by the shortest path length in $G$ from $u$ to $w$. Leveraging the $p$-hop neighborhood set, the edge set of $p$-local subgraph $G_p(e)$ around the edge $e$ is as follows:
%}
\begin{equation}\label{eq:2.2}
\begin{aligned}
E\bigl(G_p(e)\bigr)
= \left\{\, (x,y)\in E \;\middle|\;
\begin{array}{l}
x,y\in N_p(e),\\[2pt]
\min\{\operatorname{dist}_G(u,x),\operatorname{dist}_G(v,x),
      \operatorname{dist}_G(u,y),\operatorname{dist}_G(v,y)\}\le p-1
\end{array}
\right\}.
\end{aligned}
\end{equation}
For a vertex $v$, define $N_p(v)$ and $G_p(v)$ analogously.

As a natural extension of graphs, a hypergraph is $H=(V,\mathcal{E})$ with $\mathcal{E}\subseteq 2^V$. Its primal (Gaifman) graph $PG(H)$ connects $x,y\in V$ whenever there exists $f\in\mathcal{E}$ with $\{x,y\}\subseteq f$. Neighborhoods and distances on $H$ are taken via $PG(H)$.
For a hyperedge $S\in\mathcal{E}$ and an integer $p\ge 0$, define the $p$-hop
vertex neighborhood of $S$ by
\[
  N_p^H(S) = \bigl\{ v\in V : dist_{PG(H)}(v,S)\le p \bigr\},
  \qquad
  dist_{PG(H)}(v,S) := \min_{u\in S} dist_{PG(H)}(u,v).
\]
The corresponding $p$-local induced sub-hypergraph around $S$ is
\[
  H_p(S) := H\bigl[ N_p^H(S) \bigr],
\]
that is, the sub-hypergraph of $H$ with vertex set $N_p^H(S)$ and with hyperedge
set $\{ T\in\mathcal{E} : T\subseteq N_p^H(S) \}$.

% A graph $H$ is a \emph{minor} of $G$ if $H$ can be obtained from a subgraph of $G$ by contracting edges. A class $\mathcal{C}$ is \emph{$H$-minor-free} if no member contains $H$ as a minor. A graph is \emph{planar} if it admits a crossing-free embedding in the plane. A graph is \emph{apex} if the deletion of one vertex makes it planar; a class is \emph{apex-minor-free} if it excludes some fixed apex graph as a minor.

\subsection{QAOA and Its Expectation Performance}
We briefly recall the QAOA ansatz for a combinatorial optimization problem. Consider an objective function $f(x)$ defined over $n$-bit strings, our goal is to find a bit string $x$ such that $f(x)$ is as large as possible. Normally, the above function can be mapped into a Hamiltonian $C$ that encodes function values for all possible input as a diagonal matrix in the computational basis, i.e.,
\[
C |\bm x\rangle = f(\bm x) |\bm x\rangle \quad \forall \bm x \in \{0,1\}^n
\]

The QAOA ansatz at round $p \in \mathbb{N}$ is a general quantum algorithm for obtaining an approximate solution to a combinatorial optimization problem. Qubits are indexed by the bit set $[n]$. Starting from the uniform superposition over all computational basis states, i.e.,  $|+\rangle^{\otimes n} = \frac{1}{\sqrt{2^n}}\sum_{x \in \{0,1\}^n } |x\rangle $, QAOA alternates between using the phase operator with the parameter $\bm\gamma$ and the mixing operator with the parameter $\bm\beta$, i.e.,
\[
|\psi_p(\bm{\gamma,\beta}) \rangle= U(B,\beta_p)U(C,\gamma_p)\cdots U(B,\beta_1)U(C,\gamma_1) |+\rangle^{\otimes n} = \Bigl(\prod_{\ell=1}^p e^{-i\beta_\ell B}\,e^{-i\gamma_\ell C}\Bigr)\,|+\rangle^{\otimes n}
\]
where $U(B,\beta_p) = e^{-i\beta_p B} = \prod_{j \in [n]} e^{-i \beta_p X_j}$ is the mix operator, $U(C,\gamma_p) = e^{-i \gamma_p C}$ is the phase operator, and angles $(\bm{\gamma,\beta})\in\mathbb{R}^p\times\mathbb{R}^p$. Measuring the QAOA state $|\psi_p(\bm{\gamma,\beta}) \rangle$ on the computational basis states, a good solution for combinatorial optimization problem can be recovered and its expected performance is as follows:
\[
F_p(G;\bm{\gamma,\beta})= \langle\psi_p(\bm{\gamma,\beta})|\,C\,|\psi_p(\bm{\gamma,\beta})\rangle
\]
Normally, obtaining the best angles is not easy. As an alternative, some suitable angles can be obtained by combining with the classical optimizer.  

In special cases for Max-Cut, qubits are indexed by the node set $V$ and the problem Hamiltonian (cost operator) aggregates edgewise parity penalties, while the mixer flips single qubit:
\begin{equation}
C=\sum_{(u,v)\in E}\frac{1-Z_uZ_v}{2},
\qquad B=\sum_{j\in V}X_j
\end{equation}
It is convenient to decompose the cost into edge terms
\begin{equation}
C_e=\frac{1-Z_uZ_v}{2},\qquad C=\sum_{e\in E}C_e,
\end{equation}
Leveraging the above expression, the expected performance of QAOA for Max-Cut problem can be expressed as the sum of edgewise contributions, i.e.,
\begin{equation}\label{eq: contribution}
F_p(G;\bm{\gamma,\beta})
  = \sum_{e\in E}\langle\psi_p(\bm{\gamma,\beta})|\,C_e\,|\psi_p(\bm{\gamma,\beta})\rangle .
\end{equation}
% Evaluating each edge contribution through the Heisenberg-evolved observable and exploiting the fact that $C_e$ is supported only on the p-local graph $G_p(e)$ of $e$, our convention is to compute the edge contribution using only the restricted subcircuit on the subgraph $G_p(e)$:

Evaluating each edge contribution through the Heisenberg-evolved observable
\(U_p^{\dagger} C_e U_p\) and exploiting the fact that this operator is supported only on
the \(p\)-local graph \(G_p(e)\) of \(e\), we compute the edge contribution using only the
restricted subcircuit on the subgraph \(G_p(e)\):

\begin{equation}
\mathrm{contrib}_p(e)
= \bigl\langle + \bigr|^{\otimes |V(G_p(e))|}
  \!U_p^\dagger\, C_e\, U_p
  \bigl|+\bigr\rangle^{\otimes |V(G_p(e))|}
\end{equation}
where $U_p = \prod_{\ell=1}^p e^{-i\beta_\ell B}\,e^{-i\gamma_\ell C},$ and $V(G_p(e))$ is the vertex set of $G_p(e)$. Summing $\mathrm{contrib}_p(e)$ over $e\in E$ reproduces \eqref{eq: contribution}. The restricted subcircuit on the subgraph is normally referred as 'local-lightcone' evaluation that is the basis for our dynamic-programming algorithms and for the structural parameters introduced later.

In addition, the above formalism can be easily extended to weighted $k$-ary constraints.  For a hypergraph $H=(V,\mathcal{E})$
with a hyperedge $f\subseteq V$ of arity $|f| \le k$ and weight $w_f\ge 0$, define the parity penalty
\begin{equation}
C_f=\frac{w_f}{2}\Bigl(1-\prod_{i\in f} Z_i\Bigr).
\end{equation}
Taking $C=\sum_{f\in\mathcal{E}} C_f$ yields a QAOA instance for hypergraph Max-Cut–type objectives.

\subsection{(Hyper)Tree Decompositions and Local Width}
Tree decompositions capture how a graph can be assembled from small, overlapping pieces arranged along a tree; they provide the standard vehicle for dynamic programming on sparse structure. Let $H=(V_H,E_H)$ be a graph.  A tree decomposition is a pair $(T,\{X_a\}_{a\in V(T)})$ whose nodes $a$ are labelled by subsets (``bags'') $X_a\subseteq V_H$ and that satisfies:
\begin{enumerate}
     \item \emph{Covering}: $\bigcup_{a} X_a=V_H$;
    \item \emph{Edge coverage}: for every $(x,y)\in E_H$ there exists $a \in V(T)$ with $\{x,y\}\subseteq X_a$;
    \item \emph{Running intersection}: for each $v\in V_H$, the set $\{a\in V(T):\, v\in X_a\}$ induces a connected subtree of $T$.
\end{enumerate}
The width of the decomposition is $\max_a |X_a|-1$, and the treewidth $\operatorname{tw}(H)$ is the minimum width over all decompositions~\cite{Bodlaender1996LinearTW}.  Intuitively, small treewidth means that all interactions can be localized within small bags along a tree-shaped backbone.

For a hypergraph $H=(V,\mathcal{E})$ we adopt the standard convention that
$\operatorname{tw}(H)=\operatorname{tw}\!\bigl(PG(H)\bigr)$,
where $PG(H)$ is the primal (Gaifman) graph on $V$ with edges between any two vertices co-occurring in a hyperedge.  Thus any tree decomposition of $PG(H)$ can be used verbatim for $H$: bags contain vertices, and every hyperedge is fully contained in some bags.  This choice aligns the hypergraph case with our operator definitions based on Pauli $Z$-parities.

In light of the $p$-hop lightcone convention, we quantify the local structural complexity seen by an edge at depth $p$.  For $p\ge 0$, define the $p$-local treewidth of $G$ by
\begin{equation}
\operatorname{ltw}_p(G)=\max_{e\in E}\ \operatorname{tw}\bigl(G_p(e)\bigr).
\end{equation}
  For hypergraphs, $\operatorname{ltw}_p$ is computed on the primal graph $PG(H)$.  By definition, $\operatorname{ltw}_p(G)$ is monotone in $p$ and bounded by the global treewidth ($\operatorname{ltw}_p(G)\le \operatorname{tw}(G)$). This parameter measures the “hardest” $p$-local neighborhood and will govern the state size of our dynamic programs.

\section{Inapproximability of the Expectation of Fixed-Round QAOA}{\label{QAOAhard}}
In this section, we establish the computational complexity of evaluating the expected cut value of QAOA on Max-Cut, i.e., for any round \(p \ge 2\), the exact evaluation problem given a graph \(G\) and angles \((\bm{\gamma,\beta})\) is NP-hard. In addition, when assuming \(\mathbf{P}\neq\mathbf{NP}\), no polynomial-time classical algorithm can approximate this expectation within additive error \(2^{-O(n)}\) . Our first formal statement is as follows:

\begin{theorem}[Main Theorem]\label{thm:3.1}
For any round \(p \ge 2\), if \(\mathbf{P} \ne \mathbf{NP}\), then no classical polynomial-time algorithm can \emph{exactly} compute the expected performance of QAOA for the Max-Cut problem at any arbitrarily given parameters \((\bm{\gamma}, \bm{\beta})\).
\end{theorem}

% The key idea of the proof for the above theorem is to reduce from a general graph $G$ to a new graph $G'$ such that $\mathrm{MAXCUT}(G)$ can be easily recovered from $\mathrm{MAXCUT}(G')$;
The key idea of the proof is to reduce an arbitrary graph $G$ to a new graph $G'$. The construction ensures that, at QAOA round $p$, the local lightcone of the relevant observable spans the entire graph, so that the expectation of the cut objective encodes $\mathrm{MAXCUT}(G')$; moreover, $\mathrm{MAXCUT}(G)$ can be recovered from $\mathrm{MAXCUT}(G')$ in polynomial time; then we rewrite \eqref{eq:3.1} as a Laurent polynomial over the unit circle and show that its \emph{largest nonzero exponent} encodes information about $\mathrm{MAXCUT}(G')$. Leveraging the sampleability of $h(e^{i\phi})$, we can recover this exponent via a discrete inverse Fourier transform in polynomial time and thus obtain $\mathrm{MAXCUT}(G)$. If the expected performance of QAOA for Max-Cut could be exactly computed in time polynomial in the graph size, this would contradict $P \neq NP$.

\begin{proof}
To begin the proof, we first present the definition of the expected value of QAOA. For a fixed constant $p\ge 2$, the QAOA expectation can be written as
\begin{equation}\label{eq:3.1}
\mathbb{E}[C] \,=\, \langle\bm{\gamma,\beta}\mid C\mid \bm{\gamma,\beta}\rangle. 
\end{equation}
In our proof, we choose a special parameter sequence with the assumption $\cos\psi\gg\sin\psi$ and $\cos\psi\ge1-2^{-n^2}$ as below:
\[
\bm \beta=\Bigl\{\,\psi,\ldots,\psi,\,\tfrac{\pi}{4},\,\tfrac{\pi}{4}\Bigr\},\qquad
\bm \gamma=\{\phi,\phi,\ldots,\phi\},
\]
where \(n=|V(G')|\). Pick the exact dyadic angle
\(
 \psi:=2^{-\lceil n^2/2\rceil}\in(0,\tfrac{\pi}{2}).
\)
 Using \(\cos x\ge 1-\tfrac{x^2}{2}\) we obtain
 \[
 \cos\psi\ \ge\ 1-\tfrac{\psi^2}{2}\ \ge\ 1-2^{-n^2}.
 \]
 Here \(\psi=1/2^{\lceil n^2/2\rceil}\) is dyadic, so its binary encoding length is \(\Theta(n^2)\) bits. The constant \(\pi\) is fixed and independent of \(n\), so no numerical-precision or encodability issues arise.

% where $n$ is the number of vertices of the constructed graph $G'$. To ensure $\cos\psi\ge1-2^{-n^2}$, choose the exact dyadic rational $\psi=2^{-\lceil n^2/2\rceil}$.
% Since $\cos \psi\ge 1-\tfrac{\psi^{2}}{2}$, we obtain
% $\cos \psi\ge 1-\tfrac{\psi^{2}}{2}\ge 1-2^{-n^2}$.
% % On the other hand, for $0\le \psi\le \tfrac{\pi}{2}$ one has the upper bound $\cos \psi\le 1-\tfrac{2}{\pi^{2}}\psi^{2}$, which implies that achieving $1-2^{-n^2}$ requires angular proximity of order
% % $\Theta(2^{-n^{2}/2})$ (radians).
% Note that $\psi$ has denominator $2^{\lceil n^2/2\rceil}$, so the parameter’s encoding length is only $O(n^{2})$ bits; $\pi$ is a fixed constant independent of $n$, hence this formulation introduces no numerical-precision or encodability issues.
% (If one insists that all parameters be rational, one may take $q=2^{\,n^{2}+3}$ and set $\psi=\lceil (\tfrac{\pi}{2})q\!\rceil/q$, which likewise gives $\sin x\ge 1-2^{-n^{2}}$ with $O(n^{2})$-bit length.)

% \subsection{Graph Construction and Reduction ($G\to G'$)}\label{sec:5.2}
\subsection{Graph Construction and Reduction (\texorpdfstring{$G \to G'$}{G to G'})}\label{sec:5.2}

First, we introduce the construction procedure for $G'$ with $|V(G')|=n$ and $|E(G')|=m$. Let $G$ be the original graph with $|V(G)|=n_0$ and $|E(G)|=m_0$ (assume $G$ is connected). Construct $G'$ as follows:

% \paragraph{Vertex blow-up.}
% For every vertex $u \in V$, create a bipartite complete graph
% \(
%   B(u) \cong K_{n_0+1,\,10n_0},
% \)
% with bipartition $(L_u, R_u)$ satisfying $|L_u| = n_0+1$ and $|R_u| = 10n_0$. Add all edges between $L_u$ and $R_u$ to $E'$. All vertices of $B(u)$ are regarded as the replacement gadget for $u$.

\paragraph{Vertex blow-up.}
For every vertex $u \in V$, create a bipartite complete graph
\(
  B(u) \cong K_{n_0+1,\,10n_0},
\)
with bipartition $(L_u,R_u)$ satisfying $|L_u| = n_0+1$ and $|R_u| = 10n_0$, and add all edges between $L_u$ and $R_u$ to $E'$. We fix an arbitrary but consistent labelling inside each gadget: we enumerate the left and right side as
\[
  L_u = \{u_{1,0},u_{1,1},\dots,u_{1,n_0}\},\;R_u=\{u_{2,0},u_{2,1},\dots,u_{2,10n_0-1}\}.
\]
Here the first index ($1$ or $2$) indicates whether the vertex lies in $L_u$ or $R_u$, and the second index $i$ records its position. All vertices of $B(u)$ are regarded as the replacement gadget for $u$.

\paragraph{Synchronous edges for original edges.}
For every original edge $(u,v) \in E$ and for every index $i \in \{0,1,\dots,n_0\}$, add an edge
\(
  \{\,u_{1,i},\,v_{1,i}\,\}
\)
between the $i$-th vertex of $L_u$ and the $i$-th vertex of $L_v$. Thus each original edge generates $(n_0+1)$ parallel position-aligned edges. Insert all such edges into $E'$.

\paragraph{Global bipartite frame.}
Add an extra complete bipartite graph
\(
  K_{100n_0^2,\,100n_0^2},
\)
whose two parts are
\[
  X = \{x_0, x_1, \dots, x_{100n_0^2-1}\}, \qquad
  Y = \{y_0, y_1, \dots, y_{100n_0^2-1}\}.
\]
Insert all edges of this bipartite graph into $E'$. Moreover, for every $u \in V$, connect the vertices $u_{1,i} \in L_u$ for $i\in[0,n_0]$ and $u_{2,j} \in R_u$ for $j\in[0,10n_0-1]$ to both $x_0$ and $y_0$.
\paragraph{Global controller.}
Add a new vertex $w$ and connect it to
  $x_0, y_0$, and every $u_{1,0}$ for $u \in V$.
Insert all these edges into $E'$.

\medskip
After the above four steps we obtain
\(
  V' \;=\; \bigcup_{u \in V} V(B(u)) \;\cup\; X \;\cup\; Y \;\cup\; \{w\},
\)
and the number of vertices is
\[
  |V'| \;=\; n_0 \bigl( (n_0+1) + 10n_0 \bigr) + 200n_0^2 + 1
         \;=\; 211n_0^2 + n_0 + 1.
\]

We will present in Appendix~\ref{app:C} a polynomial-time procedure that recovers \(\mathrm{MAXCUT}(G)\) from \(\mathrm{MAXCUT}(G')\).

\begin{figure}[htbp]
  \centering
  \begin{tikzpicture}[
    scale=1,
    dot/.style={circle,fill=black,inner sep=1pt},
    odot/.style={circle,draw,inner sep=1pt},
    core/.style={line width=0.9pt},
    aux/.style={thin},
    >=latex
  ]

  %================= expanded G' box =================
  \draw[very thick,rounded corners] (0,0) rectangle (12,4);
  \node[anchor=west] at (0.25,0.35) {$\text{expanded }G$};

  % top vertex w
  \node[dot,label={[yshift=2pt]90:$w$}] (w) at (6,5) {};

  %================= left gadget (K_{large,small}) =================
  % large side (outer) -- assume >= n0+1 = 3 shown
  \foreach \i/\y in {1/3.6, 2/3.2, 3/2.8, 4/2.4, 5/2.0}{
    \node[odot] (Llarge\i) at (1.1,\y) {};
  }
  \node at (1.1,1.6) {$\vdots$};

  % small side (inner)
  \foreach \i/\y in {1/3.6, 2/3.2, 3/2.8}{
    \node[odot] (Lsmall\i) at (3.1,\y) {};
  }
  \node at (3.1,2.4) {$\vdots$};

  % complete bipartite (left)
  \foreach \i in {1,2,3,4,5}{
    \foreach \j in {1,2,3}{
      \draw (Llarge\i) -- (Lsmall\j);
    }
  }

  %================= right gadget (K_{large,small}) =================
  % small side (inner)
  \foreach \i/\y in {1/3.6, 2/3.2, 3/2.8}{
    \node[odot] (Rsmall\i) at (8.9,\y) {};
  }
  \node at (8.9,2.4) {$\vdots$};

  % large side (outer)
  \foreach \i/\y in {1/3.6, 2/3.2, 3/2.8, 4/2.4, 5/2.0}{
    \node[odot] (Rlarge\i) at (10.9,\y) {};
  }
  \node at (10.9,1.6) {$\vdots$};

  % complete bipartite (right)
  \foreach \i in {1,2,3,4,5}{
    \foreach \j in {1,2,3}{
      \draw (Rlarge\i) -- (Rsmall\j);
    }
  }

  %================= one-to-one between small sides =================
  \draw[core] (Lsmall1) -- (Rsmall1);
  \draw[core] (Lsmall2) -- (Rsmall2);
  \draw[core] (Lsmall3) -- (Rsmall3);

  % w to top small nodes
  \draw[core] (w) -- (Lsmall1);
  \draw[core] (w) -- (Rsmall1);

  %================= bottom bipartite (X, Y) =================
  % X
  \node[dot,label=left:$x_0$] (x0) at (4.0,-0.8) {};
  \node[dot] (x2) at (4.0,-1.4) {};
  \node[dot] (x3) at (4.0,-2.0) {};
  \node at (4.0,-2.5) {$\vdots$};
  \node at (3.5,-1.8) {$X$};

  % Y
  \node[dot,label=right:$y_0$] (y0) at (8.0,-0.8) {};
  \node[dot] (y2) at (8.0,-1.4) {};
  \node[dot] (y3) at (8.0,-2.0) {};
  \node at (8.0,-2.5) {$\vdots$};
  \node at (8.5,-1.8) {$Y$};

  % bottom complete bipartite (first few shown)
  \foreach \x in {x0,x2,x3}{
    \foreach \y in {y0,y2,y3}{
      \draw[core] (\x) -- (\y);
    }
  }

  %================= connections from bottom to gadgets =================
  % w to x0,y0
  \draw[core] (w) -- (x0);
  \draw[core] (w) -- (y0);

  % x0,y0 to ALL small-side vertices in expanded G (three shown on each side)
  \foreach \j in {1,2,3}{
    \draw[aux] (x0) -- (Lsmall\j);
    \draw[aux] (x0) -- (Rsmall\j);
    \draw[aux] (y0) -- (Lsmall\j);
    \draw[aux] (y0) -- (Rsmall\j);
  }

  % x0,y0 to ALL large-side vertices in expanded G (five shown on each side)
  \foreach \j in {1,2,3,4,5}{
    \draw[aux] (x0) -- (Llarge\j);
    \draw[aux] (x0) -- (Rlarge\j);
    \draw[aux] (y0) -- (Llarge\j);
    \draw[aux] (y0) -- (Rlarge\j);
  }

  \end{tikzpicture}
  \caption{Construction of $G'$ in the 2-vertex base case.}
  \label{fig:sec3-construction}
\end{figure}
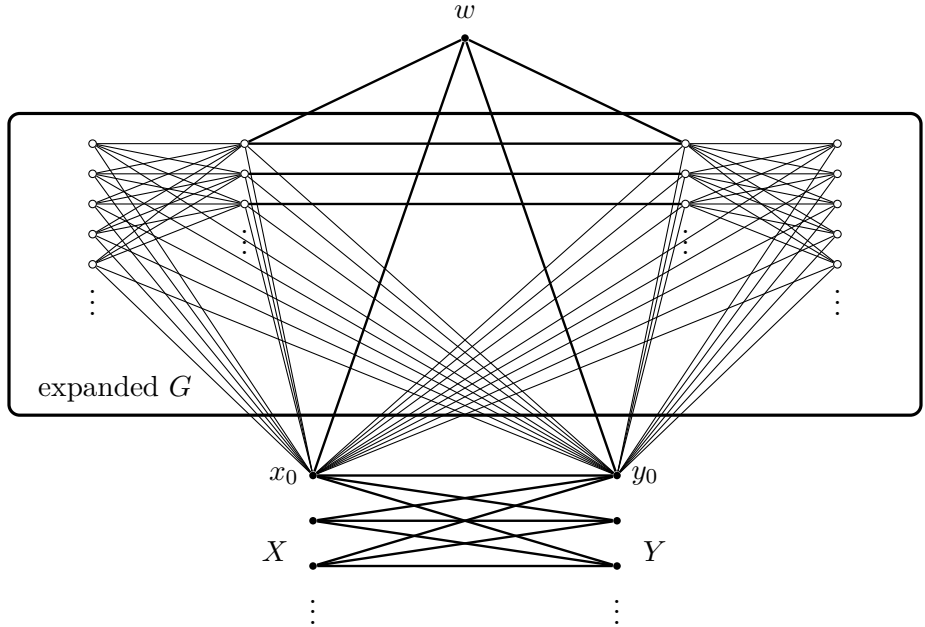

\subsection{Laurent Polynomial Representation and Decomposition}\label{sec:5.3}
Next, we introduce the expanding procedure for the expected performance of QAOA as a Laurent polynomial function.
Firstly, we have
\begin{equation}
\langle \bm{\gamma,\beta}\mid C\mid\bm{\gamma,\beta}\rangle
=\sum_{(u,v)\in E} \frac{1}{2}\,\langle\bm{\gamma,\beta}\mid (1-Z_uZ_v)\mid\bm{\gamma,\beta}\rangle, 
\end{equation}
Using layer labels $j\in\{-p,\ldots,-1,0,1,\ldots,p\}$, a single term can be written as \cite{Basso2022TQC}:
\begin{equation}\label{3.2}
\begin{aligned}
\langle\bm{\gamma,\beta}\mid Z_uZ_v\mid\bm{\gamma,\beta}\rangle
&= \sum_{\{\bm z\}} \, z_u^{[0]}z_v^{[0]}\,
\exp\!\Bigl(i\phi\,\sum_{j=1}^{p} \bigl(C_{j}(\bm z)-C_{-j}(\bm z) \bigr)\Bigr)
\prod_{v=1}^{n} f(\bm z_v),\\
\end{aligned}
\end{equation}
where, for each vertex $u\in V$, we write
\[
  \bm z_u = \bigl(z_u^{[j]}\bigr)_{j=-p}^p \in \{\pm 1\}^{2p+1}
\]
for its spin values on the $2p+1$ layers, and $\bm z^{[j]} = (z_u^{[j]})_{u\in V}$ denotes the spin configuration on layer $j$. The layer-$j$ cut value is
\[
  C_j(\bm z)
  = \sum_{(p,q)\in E}\frac{1}{2}\bigl(1 - z_p^{[j]} z_q^{[j]}\bigr),
\]
and the one-qubit mixer kernel is
\[
  f(\bm z_v)
  = \frac{1}{2}\,
    \langle z_v^{[1]}|e^{i\beta_1 X}|z_v^{[2]}\rangle\cdots
    \langle z_v^{[p]}|e^{i\beta_p X}|z_v^{[0]}\rangle\,
    \langle z_v^{[0]}|e^{-i\beta_p X}|z_v^{[-p]}\rangle\cdots
    \langle z_v^{[-2]}|e^{-i\beta_1 X}|z_v^{[-1]}\rangle .
\]
The single-qubit matrix elements are
\[
  \langle u|e^{\pm i\beta_i X}|v\rangle =
  \begin{cases}
    \cos\beta_i, & u = v,\\[2pt]
    \pm i\sin\beta_i, & u \neq v.
  \end{cases}
\]
Detailed derivations of~\eqref{3.2} can be found in Appendix~\ref{app:B}.
Formally, the above equation can be reformulated into the Laurent polynomial function as below:
\begin{equation}
\langle Z_uZ_v\rangle\;=\; \langle\bm{\gamma,\beta}\mid Z_uZ_v\mid\bm{\gamma,\beta}\rangle\;=\; h_{(u,v)}(e^{i\phi}),
\end{equation}
\begin{equation}
h_{(u,v)}(x)\;=\;\sum_{k=-m}^{m} w(k)\,x^{k}, 
\end{equation}
with $m$ the maximal possible degree and $w(k)$ the coefficients.  Hence
\begin{equation}
\langle\bm{\gamma,\beta}\mid C\mid\bm{\gamma,\beta}\rangle\;=\;-\frac{1}{2}h(e^{i\phi})+c_0\;=\;-\frac{1}{2}\sum_{(u,v)\in E} h_{(u,v)}(e^{i\phi})+c_0. 
\end{equation}
where $c_0=\sum_{(u,v)\in E}\frac{1}{2}=\frac{m}{2}$ is a constant.

\subsection{The Polynomial $h_{(x_0,y_0)}$ Induced by the Observable $Z_{x_0 }Z_{y_0}$}\label{sec:5.4}
% Because vertex $x,y$ are in different vertex sets for the optimal solution of MAXCUT$(G')$, i.e., the edge $(x,y)$ is always the cut edge for MAXCUT$(G')$.
Since the Laurent polynomial for the global observable \(C(\bm z)\) decomposes into a sum of Laurent polynomials associated with single-edge observables, we therefore first focus on the edge polynomial \(h_{(x_0,y_0)}(x)\) induced by \(Z_{x_0}Z_{y_0}\), with particular attention to its \emph{largest nonzero exponent}.
 For convenience, we denote a possible assignment for $\bm z=\{\bm z_i\}$ where $\bm z_i \in \{\pm 1\}^{2p+1}$ $\forall i \in V$ as a configuration. The detailed proof of lemmas in this subsection can be referred in Appendix~\ref{app:A}.

Before describing the largest nonzero exponent, we first introduce the Endpoint consistency phenomenon for each observable $Z_uZ_v$, which can make us ignore many items.
\begin{lemma}[Endpoint consistency]\label{lem:5.2}
For any vertex $q$, the dependence of the contribution to $\langle Z_uZ_v\rangle$ on $z_q^{[0]}$ falls into the following cases:
\begin{enumerate}
  \item If $q\notin\{u,v\}$ and $z_q^{[p]}\neq z_q^{[-p]}$, then replacing $z_q^{[0]}$ by $-z_q^{[0]}$ flips the sign of the contribution. The two configurations related by this flip therefore cancel in the sum.
  \item If $q\in\{u,v\}$ and $z_q^{[p]}=z_q^{[-p]}$, then flipping $z_q^{[0]}$ again flips the sign of the contribution, so the corresponding pair of configurations cancels.
  \item In all remaining cases (that is, when $q\notin\{u,v\}$ with $z_q^{[p]}=z_q^{[-p]}$, or $q\in\{u,v\}$ with $z_q^{[p]}\neq z_q^{[-p]}$), the contribution is independent of $z_q^{[0]}$.
\end{enumerate}
% For any vertex $q$:
% \begin{enumerate}
%   \item If $q\notin\{u,v\}$ and $z_q^{[p]}\neq z_q^{[-p]}$, then flipping $z_q^{[0]}$ changes the sign of the contribution to $\langle Z_uZ_v\rangle$; such paired configurations cancel.
%   \item If $q\in\{u,v\}$ and $z_q^{[p]}=z_q^{[-p]}$, then flipping $z_q^{[0]}$ likewise flips the sign of the contribution to $\langle Z_uZ_v\rangle$; such pairs cancel.
%   \item In all other cases (i.e., $q\notin\{u,v\}$ with $z_q^{[p]}=z_q^{[-p]}$, or $q\in\{u,v\}$ with $z_q^{[p]}\neq z_q^{[-p]}$), the contribution is independent of $z_q^{[0]}$.
% \end{enumerate}
% Let $u,v$ be the endpoints of the observable $Z_uZ_v$. For any vertex $q$ at layer $p$ and $-p$, we have:
% \begin{enumerate}
%   \item If $q\notin\{u,v\}$ and $z_q^{[p]}\neq z_q^{[-p]}$, then flipping $z_q^{[0]}$ changes the sign of the contribution to $\langle Z_uZ_v\rangle$; such paired configurations cancel.
%   \item If $q\in\{u,v\}$ and $z_q^{[p]}=z_q^{[-p]}$, then flipping $z_q^{[0]}$ likewise flips the sign; such pairs cancel.
%   \item In all other cases (i.e., $q\notin\{u,v\}$ with $z_q^{[p]}=z_q^{[-p]}$, or $q\in\{u,v\}$ with $z_q^{[p]}\neq z_q^{[-p]}$), the contribution is independent of $z_q^{[0]}$.
% \end{enumerate}
\end{lemma}
% \begin{lemma}[Endpoint consistency]\label{lem:5.2}
% If a spin $q\notin\{u,v\}$ (where $u,v$ are the endpoints of the observable $Z_uZ_v$) takes opposite values at layers $p$ and $-p$, then the two configurations differing only by flipping $q$ at layer $0$ contribute opposite amounts to $\langle Z_uZ_v\rangle$ and thus cancel; if $q\in\{u,v\}$ and $z_q^{[p]}=z_q^{[-p]}$, then the two configurations differing only at layer $0$ also cancel. For configurations satisfying $z_q^{[-p]}= z_q^{[p]}$ when $q\notin\{u,v\}$ and $z_q^{[-p]}\ne z_q^{[p]}$ when $q\in\{u,v\}$, the choice at layer $0$ does not affect the contribution.
% \end{lemma}
% Intuitively, because the boundary mixing angles are chosen as $\beta_{ p}= \frac{\pi}{4}$, flipping the single spin at layer $0$ while keeping all other spins fixed multiplies the corresponding amplitude by the opposite phase (via the factor $\pm i\,\sin(\frac{\pi}{4})\cos(\frac{\pi}{4})$), so the two configurations cancel in pairs; under the endpoint-consistency conditions, the contribution is therefore insensitive to the choice at layer $0$. XXXXX
Intuitively, once endpoint-consistency is imposed and since all cost layers are
diagonal in the $Z$ basis, the only dependence on the layer-$0$ spin of a given
vertex $q$ comes from the two end-layer mixers on $q$ (and, when $q\in\{u,v\}$,
from the observable factor $z_u^{[0]}z_v^{[0]}$). At the boundary
$\beta_p=\tfrac{\pi}{4}$ the relevant one-qubit matrix elements of
$e^{\pm i\beta_p X}$ all have the same magnitude and differ only by a sign,
so flipping $z_q^{[0]}$ either flips an overall sign (cases~(1)–(2), giving
pairwise cancellation) or leaves the contribution unchanged (case~(3)),
exactly as stated in the lemma.

\begin{lemma}[No high-degree contribution under asynchronous flips]\label{lem:5.3}
% Considering two vertex $u,v$ in the construction procedure and its the observable $Z_uZ_v$, if $z_u^{[p]}\ne z_v^{[p]}$, then there is no contribution to any exponent that $z_u^{[p]}= z_v^{[p]}$ can reach.
For the observable \(Z_{x_0}Z_{y_0}\), if \(z_{x_0}^{[p]}\neq z_{y_0}^{[p]}\), then there is no contribution to coefficients with exponents
\[
k\ge \Bigl((p-1)\,\mathrm{MC}+200n_0^2+2n_0(11n_0+1)\Bigr),
\]
\end{lemma}
where $\mathrm{MC}:=\mathrm{MAXCUT}(G')$.
Using the symmetry and cancellation above, configurations that can attain the \emph{largest nonzero exponent} and are non-negligible can be restricted to: taking a same maximum-cut assignment on layers $[1,p-1]$, a same minimum-cut assignment on layers $[-p+1,-1]$, and then using the boundary layers $(p,0,-p)$ to achieve the largest exponent. Concretely, let $C_j(\bm z)$ and $C_{-j}(\bm z)$ be a common maximum cut and a minimum cut for $j \in \{1,2,...,p-1\}$, respectively. Adjusting the spins at layers $p,0,-p$ yields a nonzero largest term. At layer $p$ we set $x_0,y_0$ to be opposite to $w$ and the rest to agree with $w$, while at layer $-p$ we flip both $x_0,y_0$ and keep all the other vertices identical to those at layer $p$. Under this operation, the largest exponent becomes $\Bigl((p-1)\,\mathrm{MC}+200n_0^2+2n_0(11n_0+1)\Bigr)$. In addition, we also have an observation that edges other than $(x_0, y_0)$ will not contribute to the largest nonzero exponent. The formal statement is as follows:
\begin{lemma}[largest nonzero exponent]\label{lem:5.4}
For the observable $Z_{x_0}Z_{y_0}$, the largest nonzero exponent in $h_{(x_0,y_0)}$ attainable by the induced polynomial is
\begin{equation}\label{eq:3.8}
\Bigl((p-1)\,\mathrm{MC}+200n_0^2+2n_0(11n_0+1)\Bigr).
\end{equation}
\end{lemma}
\begin{lemma}[Lower bound for other edges]\label{lem:5.5}
Let the observable be $Z_iZ_j$ with $(i,j)\notin\{(x_0,y_0),$ $(y_0,x_0)\}$. Then all coefficients of the terms in $h_{(i,j)}$ with exponents \(\ge \Bigl((p-1)\,\mathrm{MC}+200n_0^2+2n_0(11n_0+1)\Bigr)\) vanish.
\end{lemma}
Combining the above lemmas, the main remaining obstacle has become the calculation for the largest nonzero exponent of \(h\). It reduces to consider the largest nonzero exponent of \(h_{(x_0,y_0)}\), i.e., $\Bigl((p-1)\,\mathrm{MC}+200n_0^2+2n_0(11n_0+1)\Bigr)$.

\subsection{Sampling on the Unit Circle and FFT Interpolation: Recovering Max-Cut from the Exponent}\label{sec:5.5}
Leveraging the aforementioned statements, we illustrate that the largest nonzero exponent of $h$ is $D_{max} =\Bigl((p-1)\,\mathrm{MC}+200n_0^2+2n_0(11n_0+1)\Bigr)$. Next, we introduce the recovering procedure for MC from the exponent.
% By Lemma~\ref{lem:5.4}, let
% \begin{equation}
% D\;=\;(p-1)\,\mathrm{MC}+\delta\,(2n-6)+2.
% \end{equation}

 On the $(2m+1)$-st roots of unity $\omega=e^{2\pi i/(2m+1)}$, choose the $2m+1$ equally spaced points $x_j=\omega^j$ for $j=-m,\ldots,m$. Then
\begin{equation}
-\frac{1}{2}h(x_j)+c_0=-\frac{1}{2}\sum_{k=-m}^{m} w(k)\,\omega^{jk}+c_0=\sum_{k=-m}^{m} w_1(k)\,\omega^{jk}. 
\end{equation}
where $c_0$ is a known constant.
Let $\mathcal{F}(x)=x^{m}(-\frac{1}{2}h(x)+c_0)$. Given the sequence $\mathcal{F}(x_j)$, the inverse FFT yields
\begin{equation}\label{eq:DFT}
w_1(k-m)=\frac{1}{2m+1}\sum_{j=0}^{2m} \mathcal{F}(x_j)\,\omega^{-jk},\qquad k=0,1,\ldots,2m, 
\end{equation}
so that all coefficients $\{w_1(k)\}$ are recovered in time $O(m\log m)$. By Lemmas~\ref{lem:5.4} and \ref{lem:5.5}, the \emph{largest nonzero exponent} is uniquely determined by the $Z_{x_0}Z_{y_0}$-induced term; therefore, we can obtain $D_{max}$ by finding the term with the highest degree that has a non-zero coefficient, obtain $\mathrm{MC}=\mathrm{MAXCUT}(G')$, and then recover $\mathrm{MAXCUT}(G)$. If a classical polynomial-time algorithm could compute \eqref{eq:3.1} for arbitrary $(\bm \gamma,\bm \beta)$, then the above would yield a polynomial-time algorithm for Max-Cut on arbitrary graphs, contradicting $\mathbf{P}\ne\mathbf{NP}$. This proves Theorem~\ref{thm:3.1}.
\quad
\end{proof}

\subsection{Inapproximability under Additive Error}\label{sec:5.6}

We now derive hardness-of-approximation consequences from the coefficient-recovery argument based on the inverse transform \eqref{eq:DFT}.

We will use the following standing lemma.
\begin{lemma}[Gap at the top degree]\label{lem:3.7}
\emph{There exists a constant $c>0$ (independent of $n$ when $n$ is large enough) such that the coefficient at the largest nonzero exponent $D_{\max}$ satisfies $|w(D_{\max})|=|-2w_1(D_{\max})|\ge c\cdot 2^{-2n}$.}
\end{lemma}
The detailed proof of this lemma is deferred to Appendix~\ref{app:A}. We now begin the proof of the Corollary~\ref{cor:5.6}.

% \begin{corollary}[Additive-error bound]\label{cor:5.6}
% For any fixed round $p\ge 2$, if $\mathbf{P}\ne\mathbf{NP}$ then no classical polynomial-time algorithm can, for arbitrarily prescribed parameters $(\bm\gamma,\bm\beta)$, approximate \emph{(3.1)} within additive error $o(2^{-2n})$. In particular, there exists a constant $c'>0$ (independent of $n$) such that no algorithm can achieve additive error $\le c'\cdot 2^{-2n}$.
% \end{corollary}

\begin{corollary}[Additive-error bound]\label{cor:5.6}
For any fixed round $p\ge 2$, if $\mathbf{P}\neq\mathbf{NP}$ then there exists a constant $c'>0$ (independent of $n$) such that no classical polynomial-time algorithm can, for arbitrarily prescribed parameters $(\bm\gamma,\bm\beta)$, approximate \eqref{eq:3.1} on $n$-vertex instances within additive error $c'2^{-2n}$. 
\end{corollary}

\begin{proof}
Write $h(x)$ for the Laurent polynomial in \eqref{eq:3.1} so that $\langle C\rangle = -\frac{1}{2}h(e^{i\phi})+c_0$. As established in Sec.~\ref{QAOAhard}, the support of $h$ lies in $\{-m,\ldots,m\}$. Let $\mathcal{F}(x):=x^{m}(-\frac{1}{2}h(x)+c_0)=\sum_{k=0}^{2m}w_1(k-m)\,x^{k}$ be the associated ordinary polynomial, and let $\omega:=e^{2\pi i/(2m+1)}$.

\medskip
Consider any putative polynomial-time approximator $\mathcal{A}$ for \eqref{eq:3.1} at $(\bm\gamma,\bm\beta)$. On input the $2m{+}1$ angles $\phi_j:=\tfrac{2\pi j}{2m+1}$, it returns values $\langle C\rangle^*(\omega^j)$ obeying
\[
\bigl|\langle C\rangle^*(\omega^j)-\langle C\rangle(\omega^j)\bigr|\le \varepsilon\quad\text{for all }j\in\{0,\ldots,2m\}.
\]
Since $|\,\omega^j\,|=1$, setting $\mathcal{F}^*(\omega^j):=(\omega^j)^{m}(\langle C \rangle^*(\omega^j))$ gives
\[
\mathcal{F}^*(\omega^j)=\mathcal{F}(\omega^j)+\eta_j,\qquad |\eta_j|\le \varepsilon.
\]
Recover the coefficients by the inverse DFT \eqref{eq:DFT}:
\[
w_1^*(k{-}m)=\frac{1}{2m{+}1}\sum_{j=0}^{2m}\mathcal{F}^*(\omega^j)\,\omega^{-jk}.
\]
Averaging preserves the uniform error, hence for every $k$,
\begin{equation}\label{eq:coeff-error}
\Bigl|\,w_1^*(k{-}m)-w_1(k{-}m)\Bigr|
\le \frac{1}{2m{+}1}\sum_{j=0}^{2m}|\eta_j|
\le \varepsilon.
\end{equation}

Choose $\varepsilon\le (c/3)\,2^{-2n}$. Since $D_{\max}$ is the \emph{largest} index with a nonzero coefficient, all degrees $>D_{\max}$ have coefficient $0$, while $|w_1(D_{\max})|\ge c\,2^{-2n}$ by Lemma~\ref{lem:3.7} ($c$ is a constant). By \eqref{eq:coeff-error} we then have
\[
\begin{cases}
|w_1^*(s)|\le \varepsilon \le (c/3)\,2^{-2n}, & \text{for all } s>D_{\max},\\[2pt]
|w_1^*(D_{\max})|\ge |w_1(D_{\max})|-\varepsilon \ge (2c/3)\,2^{-2n}.
\end{cases}
\]
Therefore, scanning indices from $m$ downward and thresholding at $\tau:=(c/2)\,2^{-2n}$, the first index whose estimated magnitude exceeds $\tau$ is \emph{exactly} $D_{\max}$.

Finally, by Lemma~\ref{lem:5.4},  $D_{\max}$ is an explicit affine function of the $\mathrm{MAXCUT}$ value (with fixed coefficients depending only on $p$ and the normalization), so this procedure recovers $\mathrm{MAXCUT}$ in classical polynomial time, contradicting $\mathbf{P}\ne\mathbf{NP}$.
\end{proof}

This corollary can be naturally extended to its quantum version.
Assuming $\mathbf{NP}\nsubseteq\mathbf{BQP}$, for any fixed $p\ge 2$ there is also no $\mathrm{BQP}$ algorithm that, for arbitrarily prescribed $(\boldsymbol{\bm \gamma},\boldsymbol{\bm \beta})$, approximates \eqref{eq:3.1} within additive error $c'2^{-2n}$. Since we did not require the machines to be quantum or classical in our previous arguments, therefore achieving accuracy $\le c'\,2^{-2n}$  would let us solve $\mathrm{Max}\text{-}\mathrm{Cut}$ which would conflict with $\mathbf{NP}\nsubseteq\mathbf{BQP}$.

\section{QAOA Expectation with Tree Decomposition}{\label{QAOADP}}
In this section, we first present a \emph{bag-level dynamic programming (DP)} framework leveraging tree decomposition to exactly calculate the expectation of QAOA, which also suffers from the hardness result in the above section. We also extend our algorithm to evaluate the expectation of QAOA for general combinatorial optimization problems.

\subsection{Bag-Level Dynamic Programming}
Recalling the definition of the expectation for QAOA with the constant round $p$ on a general graph, it can be reformulated as a sum of strictly $p$-local contributions for each edge. Considering the edge $e=(l,r)$ with weight $w_{lr}$ and its radius-$p$ induced subgraph $G_p(e)$, the contribution for the edge $e$ is as below:
\begin{equation}\label{eq:master}
\mathrm{contrib}_p(e)
=\sum_{\{\bm z_v:\,v\in V(G_p(e))\}}
\Bigl[\tfrac{w_{lr}}{2}\bigl(1-z^{[0]}_l z^{[0]}_r\bigr)\Bigr]\,
\exp\!\Bigl(i\sum_{j=-p}^{p}\Gamma_j\,C_{G_p(e)}\bigl(z^{[j]}\bigr)\Bigr)\,
\prod_{v\in V(G_p(e))} f\!\bigl(\bm z_v\bigr)
\end{equation}
where $\bm\Gamma=\{\gamma_1,\dots,\gamma_p,0,-\gamma_p,\dots,-\gamma_1\}$ collects the angle coefficients, 
$C_{G_p(e)}(\cdot)$ is the cut on $G_p(e)$, and $f(\cdot)$ is the one-site mixer kernel.

\noindent\textbf{High-level idea.}
Fix a round $p$ and angles $(\bm\gamma,\bm\beta)$. For each edge $e=(l,r)$ let $G_p(e)$ be the graph induced radius-$p$ neighborhood of $(l,r)$. By finite lightcones, the Heisenberg-evolved edge observable $F_p(G;\bm\gamma,\bm\beta)$ has support contained in $G_p(e)$, so the objective splits into strictly $p$-local edge contributions. A \emph{volume-based} baseline evaluates each $F_p(G;\bm\gamma,\bm\beta)$ by enumerating all spins on $V(G_p(e))$, paying $2^{|V(G_p(e))|}$. Our approach is \emph{separator-based} by equipping $G_p(e)$ with a tree decomposition and viewing bag intersections as \emph{separators}. Conditioning on a separator fixes the interface to descendants so that interior spins factor; we cache the resulting partial sums as \emph{messages} indexed by separator assignments by using dynamic programming on the decomposition. This replaces the exponential dependence on the \emph{volume} $|V(G_p(e))|$ by an exponential in the \emph{local treewidth} $\mathrm{tw}(G_p(e))$ : the per-edge cost is polynomial in $|G|$ and singly exponential in $(p,\mathrm{tw}(G_p(e)))$ (e.g., $2^{O((\mathrm{tw}(G_p(e))+1)\cdot p)}$ up to polynomial factors), rather than $2^{|V(G_p(e))|}$. Hence, whenever $G_p(e)$ admits small separators, DP reuses sub-computations across the exponentially many global assignments and yields exponential savings over brute force; when $\mathrm{tw}(G_p(e))$ is large the algorithm degrades gracefully toward the natural exponential baseline while remaining exact. The pseudocode of the bag-level dynamic program is given in Algorithm~\ref{alg:1}.

\begin{algorithm}
  \caption{ComputeQAOAExpectation($G,\bm\gamma,\bm\beta,p$)}\label{alg:1}
  \label{alg:qaoa-expectation}
  \begin{algorithmic}[1]
    \REQUIRE a graph $G=(V,E)$, angles $\bm\gamma,\bm\beta$, and round $p$
    \ENSURE unnormalized expectation $Ex$ of QAOA on $G$
    \STATE $Ex \leftarrow 0$
    \FORALL{$e=(u,v)\in E$}
      \STATE $H \leftarrow \textsc{BFSSubgraph}(G,e,p)$ \COMMENT{$p$-local subgraph $G_p(e)$ as defined in \eqref{eq:2.1}–\eqref{eq:2.2}}
      \STATE $(T,\{X_a\}_{a\in V(T)}) \leftarrow \textsc{Korhonen2ApproxTD}(H)$
      \STATE $\Delta \leftarrow \textsc{DPOnDecomposition}(H,T,\{X_a\},\bm\gamma,\bm\beta,e)$
      \STATE $Ex \leftarrow Ex + \Delta$
    \ENDFOR \\
    \RETURN $Ex$
  \end{algorithmic}
\end{algorithm}

The subsequent subsection details the bag messages, separator structure, and the root assembly used by \textsc{DPOnDecomposition}, which implements \eqref{eq:master} exactly on the chosen tree decomposition.

% \subsection{Algorithmic Framework (Pseudocode)}

% \paragraph{Algorithm 3.1: \textsc{ComputeQAOAExpectation}.}
% \emph{Input:} \(G=(V,E)\), \((\boldsymbol\gamma,\boldsymbol\beta)\), depth \(p\).
% \emph{Output:} \(F_p(G;\boldsymbol\gamma,\boldsymbol\beta)\).
% \begin{enumerate}
% \item Set \(Ex\leftarrow 0\).
% \item For each edge \(e=\{u,v\}\in E\):
%   \begin{enumerate}
%   \item Let \(H\leftarrow G_p(e)\) (BFS induced subgraph of radius \(p\)).
%   \item Use a single-exponential-time 2-approximation for treewidth (Korhonen): given \(k\), within
%   \begin{equation}\label{eq:TTD}
%   T_{\mathrm{TD}}(H,k)=n_H\cdot 2^{O(k)}
%   \end{equation}
%   either return a decomposition of width \(\le 2k{+}1\) or certify \(\mathrm{tw}(H)>k\).
%   Increase \(k=1,2,\ldots\) and keep the smallest returned width; when \(k\approx \mathrm{tw}(H)\) the returned width satisfies \(t'\le 2\,\mathrm{tw}(H)+1\) (see \eqref{eq:twprime-bound}).
%   \item Run \textsc{DPOnDecomposition} (Algorithm~3.2) on \((H,T,\{X_a\})\) to obtain \(\Delta=\mathrm{contrib}_p(e)\).
%   \item Update \(Ex\leftarrow Ex+\Delta\).
%   \end{enumerate}
% \item Return \(Ex\).
% \end{enumerate}

\subsection{DP on a Tree Decomposition and Root Assembly}
For convenience, we say that edge $(x,y)$ belongs to bag $X$ if and only if $X$ contains both vertices $x$ and $y$.
When considering the contribution of the edge $(l,r)$, the root of tree decomposition is a bag $root$ that contains the edge $(l,r)$. For a directed edge \((a\!\to\!\mathrm{pa}(a))\) in the tree decomposition $(T,\{X_a\}_{a\in V(T)})$, we define the separator \(S_a=X_a\cap X_{\mathrm{pa}(a)}\) (with \(S_{root}=\varnothing\)) and introduce the notation of vertices and \emph{top edges} of each bag:
\[
V_a=\{\,u\in X_a:\ \mathrm{top}(u)=a\,\},\qquad
E_a=\bigl\{\,(u,v)\in X_a:\ \mathrm{top}((u,v))=a\,\bigr\},
\]
where \(\mathrm{top}(\cdot)\) is the bag closest to the root within the connected subtree of bags containing the item.

Given a separator assignment \(\mathbf{z}_{S_a}=\{\,z_u^{[j]}\in\{\pm1\}:~u\in S_a,~j=-p,\ldots,p\,\}\), define the \emph{bag message} (bottom-up):
\begin{equation}\label{eq:bag-message}
% \boxed{
\begin{aligned}
g(a,\mathbf{z}_{S_a})
=\sum_{\mathbf{z}_{V_a}}
\Biggl[
\underbrace{\Bigl(\prod_{u\in V_a} f\bigl(z_u\bigr)\Bigr)}_{\text{vertex kernels settled at }V_a}
\cdot
\underbrace{\exp\!\Bigl(i\!\sum_{j=-p}^{p}\Gamma_j\,C_{E_a}\bigl(\mathbf{z}^{[j]}\bigr)\Bigr)}_{\text{edge phases settled at }E_a}
\cdot
\prod_{b\in\mathrm{child}(a)} g\bigl(b,\mathbf{z}_{S_b}\bigr)
\Biggr]
\end{aligned}
% }
\end{equation}
where \(C_{E_a}(\mathbf{z}^{[j]})=\sum_{(x,y)\in E_a}\tfrac{w_{xy}}{2}\bigl(1-z_x^{[j]}z_y^{[j]}\bigr)\). If $a$ is leaf-bag, we set \(\prod_{b\in\mathrm{child}(a)} g\bigl(b,\mathbf{z}_{S_b}\bigr) = 1\). For the \(root\) bag, we can obtain the contribution of the edge $(l,r)$ from the message at its child.
\begin{equation}\label{eq:root-assembly}
% \boxed{
\begin{aligned}
\mathrm{contrib}_p(e)
&= \sum_{\mathbf{z}_{X_{root}}}
\underbrace{\frac{w_{lr}}{2} \bigl(1-z^{[0]}_{l}\,z^{[0]}_{r}\bigr)}_{\text{measurement factor}}
\Bigl(\prod_{u\in V_{root}} f\bigl(z_u\bigr)\Bigr)
\\&\quad\times \exp\!\Bigl(i\sum_{j=-p}^{p}\Gamma_j\,C_{E_{root}}\bigl(\mathbf{z}^{[j]}\bigr)\Bigr)
\prod_{b\in \mathrm{child}(root)} g\bigl(b,\mathbf{z}_{S_b}\bigr)
\end{aligned}
% }
\end{equation}

\paragraph*{Correctness.}
By the running–intersection property of the tree decomposition, for every vertex \(u\) (resp.\ edge \((l,r)\)), bags containing it form a connected subtree of \(T\). Writing \(\mathrm{top}(u)\) and \(\mathrm{top}((l,r))\) for the unique bags in those subtrees closest to the root. The families \(\{V_a\}_{a \in V(T)}\) and \(\{E_a\}_{a \in V(T)}\) are pairwise disjoint and satisfy
\(\cup_a V_a=V(G_p(e))\) and \(\cup_a E_a=E(G_p(e))\), respectively.

Fix any bag \(a\) and separator assignment \(\mathbf{z}_{S_a}\).
The message \(g(a,\mathbf{z}_{S_a})\) defined by the bag–message recurrence \eqref{eq:bag-message}
equals the sum, over all copy variables in \(\mathrm{subtree}(a)\) that extend \(\mathbf{z}_{S_a}\), of the integrand obtained by multiplying:
(i) for each vertex \(u\) in the subtree, the one–site kernel \(f(\bm z_{u})\) exactly once, counted at \(\mathrm{top}(u)\); and
(ii) for each edge \((l,r)\) in the subtree, the phase \(\exp\!\bigl(i\sum_j \Gamma_j\,C_{(l,r)}(\mathbf{z}^{[j]})\bigr)\) exactly once, counted at \(\mathrm{top}((l,r))\).
This is proved by a bottom–up induction on \(T\): the leaf case is immediate from \eqref{eq:bag-message}; for the inductive step, child subtrees are glued along separators \(X_b\cap X_a\), and the running–intersection property guarantees cross–bag consistency so that each vertex kernel and edge phase is accounted for exactly once.
At the root, the assembly step \eqref{eq:root-assembly} inserts the observable for the chosen edge \(e\) and sums over \(\mathbf{z}_{X_{root}}\), thereby reproducing the localized objective exactly.

%\paragraph{Algorithm 3.2: \textsc{DPOnDecomposition}.}

\begin{algorithm}[H]
  \caption{DPOnDecomposition($H,T,\{X_a\},\bm\gamma,\bm\beta,e=(l,r)$)}
  \label{alg:dp-on-decomposition}
  \begin{algorithmic}[1]
    \REQUIRE induced subgraph $H$, tree decomposition $(T,\{X_a\})$, angles $(\bm\gamma,\bm\beta)$, edge $e=(l,r)$
    \ENSURE $\mathrm{contrib}_p(e)$
    \STATE \textbf{postorder-traverse} $T$ from leaves to the root $root$ (with $\{l,r\}\subseteq X_{root}$)
    \FORALL{bags $a$ in postorder}
      \STATE compute the message $g(a,\mathbf{z}_{S_a})$ via the bag recurrence
      \STATE \quad $g(a,\mathbf{z}_{S_a})=\sum_{\mathbf{z}_{V_a}}
        \Bigl[\prod_{u\in V_a} f\!\bigl(z_u\bigr)\Bigr]\,
        \exp\!\Bigl(i\sum_{j=-p}^{p}\Gamma_j\,C_{E_a}\bigl(\mathbf{z}^{[j]}\bigr)\Bigr)\,
        \prod_{b\in\mathrm{child}(a)} g\!\bigl(b,\mathbf{z}_{S_b}\bigr)$
    \ENDFOR
    \STATE assemble at the root $root$ obtaining $contrib_p(e)$ by \eqref{eq:root-assembly}
    % \STATE \quad $\mathrm{contrib}_p(e)
    %   = \sum_{\mathbf{z}_{X_{root}}}
    %   \frac{w_{lr}}{2}\bigl(1-z^{[0]}_{l}z^{[0]}_{r}\bigr)\,
    %   \Bigl[\prod_{u\in V_{root}} f\!\bigl(z_u\bigr)\Bigr]\,
    %   \exp\!\Bigl(i\sum_{j=-p}^{p}\Gamma_j\,C_{E_{root}}\bigl(\mathbf{z}^{[j]}\bigr)\Bigr)\,
    %   \prod_{b\in \mathrm{child}(root)} g\!\bigl(b,\mathbf{z}_{S_b}\bigr)$
    \STATE \RETURN $\mathrm{contrib}_p(e)$
  \end{algorithmic}
\end{algorithm}

Next, we discuss the complexity of the algorithm. For convenience, denote $L:=2p+1$.

\paragraph{Per-subgraph parameters.} For a $p$-local subgraph $G_p(e)$, let $\widetilde{\mathrm{tw}}(G_p(e))$ be the width of the tree decomposition actually used by the DP. Using the $2$-approximation for treewidth~\cite{Korhonen2021TwoApproxTW}, we obtain
$\widetilde{\mathrm{tw}}(G_p(e))\le 2\,\mathrm{tw}(G_p(e))+1$.

\paragraph{Building neighborhoods and TDs.}
Extracting the radius-$p$ neighborhood $G_p(e)$ by BFS takes
$\widetilde O(|V(G_p(e))|+|E(G_p(e))|)$ time.
Computing a $2$-approximate tree decomposition (TD) for $G_p(e)$ takes $\widetilde O\!\big(|V(G_p(e))|\cdot 2^{O(\mathrm{tw}(G_p(e)))}\big)$
time and returns width $\le 2\mathrm{tw}(G_p(e)){+}1$.

\paragraph{Per-subgraph DP.}
Process the TD in postorder. A bag with at most $\widetilde{\mathrm{tw}}(G_p(e)){+}1$ vertices has
$2^{L(\widetilde{\mathrm{tw}}(G_p(e))+1)}$ assignments of replicated spins.
Since each local kernel/phase and child message is aggregated only once, from an amortized perspective, each table entry aggregates $O(1)$ local kernels/phases and $O(1)$ child messages, and this process must be carried out after pretabulation, hence
\[
T_{\mathrm{DP}}(G_p(e))\;=\;2^{\,L\,(\widetilde{\mathrm{tw}}(G_p(e))+1)}\cdot \widetilde O\!\big(|V(G_p(e))|+|E(G_p(e))|\big)
\]
\[
S_{\mathrm{DP}}(G_p(e))\;=\;2^{\,L\,(\widetilde{\mathrm{tw}}(G_p(e))+1)}\cdot \widetilde O(|V(G_p(e))|)
\label{eq:perH-DP}
\]

\paragraph{Total over all edges.}
For $G_p(e)$, write $n_e:=|V(G_p(e))|$, $m_e:=|E(G_p(e))|$,
$t_e:=\mathrm{tw}(G_p(e))$, $w_e:=\widetilde{\mathrm{tw}}(G_p(e))$, and define
\[
N:=\sum_{e} n_e,\qquad M:=\sum_{e} m_e,\qquad
t:=\max_e t_e,\qquad w:=\max_e w_e.
\]
Summing the three components (neighborhood extraction, TD, DP) gives
\begin{equation}\label{eq:total-time}
% \boxed{
T_{\text{total}}
\;\le\;
\underbrace{\widetilde O(N+M)}_{\text{build all }G_p(e)}
\;+\;
\underbrace{\widetilde O\!\big(N\cdot 2^{O(t)}\big)}_{\text{TD $2$-approx}}
\;+\;
\underbrace{2^{\,L\,(w+1)}\cdot \widetilde O(N+M)}_{\text{all bag-DP work}}
% }\,.
\end{equation}
In particular, for fixed $p$ the DP term dominates asymptotically, and
\[
T_{\text{total}}
\;=\;
2^{\,L\,(w+1)}\cdot \widetilde O(N+M)\;+\;\text{(lower-order terms)}.
\]
The peak working space is
\(
S_{\text{total}}\;=\;2^{\,L\,(w+1)}\cdot \widetilde O(N).
\)

\paragraph{Bounded local treewidth (fixed $p$).}
If the instance family has bounded local treewidth, i.e., $t=\max_e t_e=O(1)$,
then also $w=O(1)$ for the $2$-approx TD (and $w=t$ for exact TDs). Hence
\(
T_{\text{total}}\;=\;\mathrm{poly}(N),
\;S_{\text{total}}\;=\;\mathrm{poly}(N).
\)

\subsection{Generalization to Binary Unconstrained Combinatorial Optimization}
\label{sec:gen-hg-buco}

We now show that the $p$-local DP framework developed for graph Max-Cut extends, with
minimal changes, to general binary unconstrained combinatorial optimization (BUCO).  We
first express a BUCO instance in a standard pseudo-Boolean form and encode it as a
hypergraph cut; in this way, the relevant locality and width parameters are inherited by
the primal graph of the induced $p$-local subinstances. Throughout this subsection we
reuse the notation for replicated layers, angle coefficients $\Gamma$, and the single-qubit
mixer kernel $f(\cdot)$ from Secs.~\ref{sec:preliminaries} and~\ref{QAOADP}.

Any BUCO objective admits the pseudo-Boolean expansion
\begin{equation}\label{eq:buco-poly}
C_B(\bm z)
 = c_0+\sum_{j}c_j z_j+\sum_{j<k}c_{jk} z_j z_k
   +\sum_{j<k<\ell}c_{jk\ell} z_j z_k z_\ell+\cdots,\qquad z_i\in\{\pm1\},
\end{equation}
where the sum ranges over all nonempty subsets $S\subseteq V$ and $c_S\in\mathbb{R}$ are
(real) coefficients.  The corresponding target Hamiltonian is
\begin{equation}\label{eq:buco-ham}
C_B
 = c_0 I+\sum_{S\neq\emptyset} c_S \prod_{i\in S} Z_i.
\end{equation}

Grouping monomials by their support, let
\[
  \mathcal{E}
  :=\{\,S\subseteq V:\ c_S\neq 0\,\}
\]
and define a hypergraph $H=(V,\mathcal{E})$ whose hyperedges are the supports $S$ of
nonzero monomials.  For each $S\in\mathcal{E}$ set
\[
  w_S := -\,2c_S
\]
and view $w_S\in\mathbb{R}$ as a weight on $S$.
The associated hypergraph cut objective and Hamiltonian are defined by
\begin{equation}\label{eq:hg-cut}
C_H(\bm z)
 = \frac12\sum_{S\in\mathcal{E}} w_S\Bigl(1-\prod_{u\in S} z_u\Bigr),
 \qquad z_u\in\{\pm1\},
\end{equation}
\begin{equation}\label{eq:hg-ham}
C_H
 = \frac12\sum_{S\in\mathcal{E}} w_S\Bigl(1-\prod_{u\in S} Z_u\Bigr).
\end{equation}
A direct calculation shows that for every spin configuration $\bm z$,
\[
  C_H(\bm z)
   = \frac12\sum_{S\in\mathcal{E}} w_S
     - \frac12\sum_{S\in\mathcal{E}} w_S\prod_{u\in S} z_u
   = -\sum_{S\in\mathcal{E}} c_S + \sum_{S\in\mathcal{E}} c_S\prod_{u\in S} z_u.
\]
Thus
\[
  C_B(\bm z)
   = c_0 + \sum_{S\in\mathcal{E}} c_S\prod_{u\in S} z_u
   = c_* + C_H(\bm z),
   \qquad
   c_* := c_0 + \sum_{S\in\mathcal{E}} c_S.
\]
In particular, $C_B$ and $C_H$ differ only by an additive constant, so they have the same
maximizers:
\[
  \arg\max_{\bm z} C_B(\bm z)
   = \arg\max_{\bm z} C_H(\bm z),
  \qquad
  \max_{\bm z} C_B(\bm z)
   = c_* + \max_{\bm z} C_H(\bm z).
\]
At the Hamiltonian level we likewise have
\[
  C_B = c_* I + C_H,
\]
so for any QAOA state $\ket{\psi_p}$ the expectations are related by
\[
  \big\langle \psi_p\big|\,C_B\,\big|\psi_p\big\rangle
   = c_* + \big\langle \psi_p\big|\,C_H\,\big|\psi_p\big\rangle.
\]
That is, evaluating the BUCO objective under QAOA is equivalent (up to an additive
constant) to evaluating the weighted hypergraph cut objective $C_H$, namely the total
weight of cut hyperedges.

Hence, in the remainder of this subsection we focus on the hypergraph cut form
\eqref{eq:hg-ham}.  Fix a hyperedge $S_0\in\mathcal{E}$ and write the corresponding
observable as
\[
  \mathcal{O}_{S_0}
   = \frac{w_{S_0}}{2}\Bigl(1-\prod_{i\in S_0}Z_i\Bigr).
\]
We let $H_{S_0}=H_p(S_0)$ denote the induced $p$-local sub-hypergraph defined in
Sec.~\ref{sec:preliminaries}.  With the layer spins $\bm z^{[j]}$ and the single-site
kernel $f(\cdot)$ as in Appendix~\ref{app:B}, the localized formula is
\begin{equation}\label{eq:hg-local-main}
\big\langle \psi_p\big|\,\mathcal{O}_{S_0}\,\big|\psi_p\big\rangle
=\sum_{\{\bm z_v:v\in V(H_{S_0})\}}
\frac{w_{S_0}}{2}\Bigl(1-\prod_{i\in S_0} z_{i}^{[0]}\Bigr)\,
\exp\!\Bigl(i\sum_{j=-p}^{p}\Gamma_j\,C_{H_{S_0}}\bigl(\bm z^{[j]}\bigr)\Bigr)\,
\prod_{v\in V(H_{S_0})} f\bigl(\bm z_v\bigr),
\end{equation}
where
\[
  C_{H_{S_0}}(\bm z^{[j]})
   = \frac12\sum_{S\in\mathcal{E}(H_{S_0})} w_S
       \Bigl(1-\prod_{u\in S} z_u^{[j]}\Bigr).
\]

On the hypergraph $H_{S_0}$, take a graph tree decomposition $(T,\{X_a\})$ of the primal
graph $PG(H_{S_0})$ such that every hyperedge $S\in\mathcal{E}(H_{S_0})$ is contained in
at least one bag $X_a$ (here we say a hyperedge $S$ belongs to $X_a$ if and only if $X_a$
contains all vertices in $S$) and the running–intersection property holds. Choose a root
bag $root$ with $S_0\subseteq X_{root}$. For each arc $(a\to\mathrm{pa}(a))$, set
$S_a=X_a\cap X_{\mathrm{pa}(a)}$ and $V_a=X_a\setminus S_a$. Define the set of hyperedges
introduced at $a$ by
\begin{equation}\label{eq:introduced-edges}
\mathcal{E}_a=\{\,S\in\mathcal{E}(H_{S_0}):~S\subseteq X_a,\ S\nsubseteq X_{\mathrm{pa}(a)}\,\}.
\end{equation}
For a separator configuration $\mathbf{z}_{S_a}$, define the bag message
\begin{equation}\label{eq:hg-message}
g(a,\mathbf{z}_{S_a})
=\sum_{\mathbf{z}_{V_a}}
\Biggl[
\prod_{u\in V_a} f\bigl(\bm z_u\bigr)\,
\exp\!\Bigl(i\sum_{j=-p}^{p}\Gamma_j\,C_{\mathcal{E}_a}\bigl(\mathbf{z}^{[j]}\bigr)\Bigr)\,
\prod_{b\in\mathrm{child}(a)} g\bigl(b,\mathbf{z}_{S_b}\bigr)
\Biggr],
\end{equation}
where
\begin{equation}\label{eq:hg-Ea}
C_{\mathcal{E}_a}(\mathbf{z}^{[j]})
=\frac12\sum_{S\in \mathcal{E}_a} w_S\Bigl(1-\prod_{u\in S} z_u^{[j]}\Bigr).
\end{equation}
If $a$ is a leaf bag, we set
\[
\prod_{b\in\mathrm{child}(a)} g\bigl(b,\mathbf{z}_{S_b}\bigr) = 1.
\]

The root assembly is
\begin{equation}\label{eq:hg-root}
\begin{aligned}
\big\langle \psi_p\big|\,\mathcal{O}_{S_0}\,\big|\psi_p\big\rangle
&= \sum_{\mathbf{z}_{X_{root}}}
\frac{w_{S_0}}{2}\Bigl(1-\prod_{i\in S_0} z_i^{[0]}\Bigr)\,
\Bigl[\prod_{u\in V_{root}} f\bigl(\bm z_u\bigr)\Bigr]\,
\\&\quad\times \exp\!\Bigl(i\sum_{j=-p}^{p}\Gamma_j\,C_{\mathcal{E}_{root}}\bigl(\mathbf{z}^{[j]}\bigr)\Bigr)\,
\prod_{b\in \mathrm{child}(root)} g\bigl(b,\mathbf{z}_{S_b}\bigr).
\end{aligned}
\end{equation}

\begin{algorithm}[tbp]
  \caption{ComputeQAOAExpectation--BUCO
    \texorpdfstring{$(C_B,\boldsymbol{\gamma},\boldsymbol{\beta},p)$}{(C_B,gamma,beta,p)}}
  \label{alg:qaoa-expectation-buco}
  \begin{algorithmic}[1]
    \REQUIRE BUCO objective $C_B(\bm z)
      = c_0 + \sum_{S\neq\emptyset} c_S \prod_{i\in S} z_i$ on spins $z_i\in\{\pm1\}$,
      angles $(\boldsymbol{\gamma},\boldsymbol{\beta})$, round $p$
    \ENSURE $F_p(C_B;\boldsymbol{\gamma},\boldsymbol{\beta})
      = \big\langle \psi_p\big|\,C_B\,\big|\psi_p\big\rangle$
    \STATE Construct the hypergraph $H=(V,\mathcal{E})$ with one hyperedge $S\in\mathcal{E}$
      for each nonzero monomial $c_S\prod_{i\in S}z_i$, and set $w_S \gets -\,2c_S$
    \STATE $c_* \gets c_0 + \sum_{S\in\mathcal{E}} c_S$
    \STATE $Ex \gets c_*$ \hfill // constant shift: $C_B = c_* I + C_H$
    \FORALL{$S_0\in\mathcal{E}$}
      \STATE $H_{S_0} \gets H_p(S_0)$
      \STATE Compute a tree decomposition $(T,\{X_a\})$ of $PG(H_{S_0})$
        with a root bag $root$ such that $S_0\subseteq X_{root}$
      \STATE $\Delta \gets \langle \psi_p|\mathcal{O}_{S_0}|\psi_p\rangle$
        via bag messages and root assembly \eqref{eq:hg-message}--\eqref{eq:hg-root}
      \STATE $Ex \gets Ex+\Delta$
    \ENDFOR \\
    \RETURN $Ex$
  \end{algorithmic}
\end{algorithm}

We next analyse the time complexity of the hypergraph algorithm.

\paragraph{Setup (hypergraph, per-local instance).}
Let $H=(V,\mathcal{E})$ be a hypergraph and write
$r:=\max_{S\in\mathcal{E}}|S|$ for its (global) maximum hyperedge size.
For each hyperedge $S\in\mathcal{E}$, let $H_S:=H_p(S)$.  Write $n_S:=|V(H_S)|$,
\(
A_S:=\sum_{T\in\mathcal{E}(H_S)}|T|
\)
for the total arity within $H_S$,
and let $t_S:=\mathrm{tw}(PG(H_S))$ be the treewidth of the primal graph.
Let $w_S$ denote the TD width actually used by the DP on $H_S$, and set $L:=2p+1$.
Define the global aggregates
\[
N_H:=\sum_{S\in\mathcal{E}} n_S,\qquad
A_H:=\sum_{S\in\mathcal{E}} A_S,\qquad
t:=\max_{S} t_S,\qquad
\kappa:=\max_{S} w_S.
\]
We use the unit-cost RAM model; per-round trigonometric coefficients are pretabulated
(affecting only polynomial factors).

\paragraph{Tree decomposition on $PG(H_S)$.}
Using any single-exponential $2$-approximation for treewidth, we obtain a TD for
$PG(H_S)$ of width $w_S\le 2t_S{+}1$ in time
\[
T_{\mathrm{TD}}(H_S)\;=\;\widetilde O\!\big(n_S\cdot 2^{O(t_S)}\big).
\]
% where $\widetilde O(\cdot)$ hides factors polynomial in $t_S$ and in the encoding length
% of $H_S$ (that is, in $n_S$ and $A_S$).

\paragraph{Per-local DP on $H_S$.}
Each bag of $PG(H_S)$ contains at most $w_S{+}1$ vertices and thus has
$2^{L(w_S+1)}$ replicated-spin assignments. Aggregating local kernels/phases and child
messages per assignment gives
\[
T_{\mathrm{DP}}(H_S)\;=\;2^{\,L\,(w_S+1)}\cdot \widetilde{O}\!\big(n_S + A_S\big),
\qquad
S_{\mathrm{DP}}(H_S)\;=\;2^{\,L\,(w_S+1)}\cdot \widetilde{O}(n_S),
\]
where $\widetilde O(\cdot)$ hides factors polynomial in $p,w_S,r$ and the input size of
$H_S$.

\paragraph{Global cost (sum over local $p$-instances).}
Extracting each $H_S$ on $PG(H)$ costs $\widetilde O(n_S + A_S)$, hence
$\sum_S \widetilde O(n_S + A_S)=\widetilde O(N_H + A_H)$.
Summing the TD and DP costs over $S\in\mathcal{E}$ and using $t=\max_S t_S$ and
$\kappa=\max_S w_S$ yields
\begin{equation}\label{eq:total-hyper}
 T_{\text{total}}^{\mathrm{(hyper)}}
\ \le\
\underbrace{\widetilde O(N_H + A_H)}_{\text{extract all }H_S}
\ +\ 
\underbrace{\widetilde O\!\big(N_H\cdot 2^{O(t)}\big)}_{\text{TD $2$-approx on }PG(H_S)}
\ +\ 
\underbrace{2^{\,L\,(\kappa+1)}\cdot \widetilde O(N_H + A_H)}_{\text{all bag-level DP}}.
\end{equation}
In particular (for fixed $p$), the DP term dominates asymptotically:
\[
T_{\text{total}}^{\mathrm{(hyper)}}\;=\;2^{\,L\,(\kappa+1)}\cdot \widetilde O(N_H + A_H)\;+\;\text{(lower-order terms)},
\qquad
S_{\text{total}}^{\mathrm{(hyper)}}\;=\;2^{\,L\,(\kappa+1)}\cdot \widetilde O(N_H),
\]
where the space bound follows by processing the $H_S$ sequentially and freeing temporaries
after each subproblem.

\paragraph{Bounded local treewidth on $PG(H_S)$ (fixed $p$).}
If $\max_S t_S=O(1)$, then also $\kappa=O(1)$ for the $2$-approx TD (and $\kappa=t$ with
exact TDs), whence
\[
T_{\text{total}}^{\mathrm{(hyper)}}\;=\;\mathrm{poly}(N_H + A_H),
\qquad
S_{\text{total}}^{\mathrm{(hyper)}}\;=\;\mathrm{poly}(N_H).
\]
With exact TDs the DP factor tightens from $2^{L(\kappa+1)}$ to $2^{L(t+1)}$.
\medskip

\section{Experiments}\label{sec:experiment}
We conduct fixed-round QAOA experiments on three classes of structured graphs (see
Fig.~\ref{fig:three-structured-graphs}) using the exact \emph{\(p\)-local} expectation
evaluator proposed in Sec.~\ref{QAOADP}. Using the normalized cut ratio
\(\langle C\rangle/|E|\) as the metric, we compare constant-round performance across
these graphs, which differ in their short-cycle density and local geometry, and include
classical baselines that are \emph{locality-matched} to QAOA for a fair comparison\footnote{%
  Source code and instructions are available at
  \url{https://github.com/WABINSH/QAOA_expectation_experiment.git}.}.

\subsection{Graph Families and Rationale}

% \paragraph{Generalized Petersen graph $GP(15,2)$.}
% A canonical cubic benchmark on $30$ vertices and $45$ edges, built from an outer $15$-cycle, an inner $15$-cycle connected with step $2$, and $15$ radial spokes.
% Formally (indices modulo $15$), the edge sets are
% $\{v_i,v_{i+1}\}$, $\{v_i,w_i\}$, and $\{w_i,w_{i+2}\}$ for $i=0,\dots,14$.
% This preserves degree $3$ while injecting a controlled supply of short odd cycles so the graph is non-bipartite with girth $5$.
% The dihedral symmetries (rotations/reflections of the index $i$) keep local neighborhoods essentially uniform, enabling controlled studies of short-cycle effects at fixed degree.
\paragraph{Generalized Petersen graph $GP(15,2)$.}
The generalized Petersen graphs $GP(n,k)$ form a well-studied family~\cite{FruchtGraverWatkins1971GP};
here we use $GP(15,2)$, a canonical benchmark on $30$ vertices and $45$ edges (in particular $3$-regular),
built from an outer $15$-cycle, an inner $15$-cycle connected with step $2$, and $15$ radial spokes.
Formally (indices modulo $15$), the edge sets are
$\{v_i,v_{i+1}\}$, $\{v_i,w_i\}$, and $\{w_i,w_{i+2}\}$ for $i=0,\dots,14$.
This construction injects a controlled supply of short odd cycles so the graph is non-bipartite with girth $5$.
The dihedral symmetries (rotations/reflections of the index $i$) keep local neighborhoods essentially uniform,
enabling controlled studies of short-cycle effects at fixed degree.

% \paragraph{Double-layer triangular lattice via iterated 2-lifts.}
% Iterated 2–lifts provide a scalable construction with highly regular local neighborhoods and abundant short cycles. The uniform local structure enables like-for-like comparisons between QAOA and classical local rules under the same protocol.
\paragraph{Double-layer triangular lattice via iterated 2-lifts.}
We view this graph as a simple example of an iterated $2$-lift of a triangular base graph in the sense of Bilu and Linial~\cite{BiluLinial2006Lifts}.
Iterated $2$-lifts provide a scalable construction with highly regular local neighborhoods and abundant short cycles.
The uniform local structure enables like-for-like comparisons between QAOA and classical local rules under the same protocol.

% \paragraph{Double-layer triangular lattice via iterated 2-lifts.}
% Starting from a small triangular patch, iterated 2-lifts yield a scalable family of
% 3-regular graphs with highly regular local neighborhoods that inherit many short
% cycles from the underlying triangular lattice. The resulting uniform local structure
% makes the double-layer triangular lattice a convenient testbed for like-for-like
% comparisons between QAOA and classical local rules under the same protocol.

\paragraph{Truncated icosahedron graph $C_{60}$.}
$C_{60}$ is the truncated icosahedron graph of the truncated icosahedron, a 3-regular
fullerene formed by a pentagon–hexagon tiling with two edge types. Its girth is~5, and
the pentagonal faces introduce unavoidable local frustration for the Max-Cut objective.
The high symmetry distributes this frustration uniformly across the graph, making
$C_{60}$ a demanding test case for purely \(k\)-local classical rules; geometrically, it is the familiar truncated icosahedron / fullerene
polyhedron~\cite{Kostant1994TruncIco}. Meanwhile, the
bounded degree and homogeneous local neighborhoods resemble nearest-neighbor quantum
hardware layouts. For our evaluator, every edge-radius-3 neighborhood has treewidth~3,
so a \(p\)-local dynamic-programming scheme can exactly compute expectation values at
shallow \(p\). These features make $C_{60}$ a physically meaningful, symmetry-rich
stress test on which shallow QAOA can be probed against challenging classical baselines.

\begin{table}[h!]
\centering
\caption{$p$-local treewidth for the three instances.}
\label{tab:p-local-treewidth}
\begin{tabular}{lccc}
\hline
Instance & $p=1$ & $p=2$ & $p=3$ \\
\hline
$GP(15,2)$                 & 1 & 2 & 3 \\
Double-layer triangular 2-lift  & 2 & 2 & 2 \\
$C_{60}$                   & 1 & 2 & 3 \\
\hline
\end{tabular}
\end{table}

\begin{figure*}[t]
  \centering
  \begin{tabular}{@{}c@{\hspace{6pt}}c@{\hspace{6pt}}c@{}}

    % --- (1) GP(15,2) ---
    \resizebox{0.31\textwidth}{!}{

\resizebox{0.31\textwidth}{!}{
\begin{tikzpicture}[
  x=1cm,y=1cm,
  line cap=round,line join=round
]
  % ---- styles ----
  \tikzset{
    vertex/.style={circle,draw,fill=white,inner sep=0pt,minimum size=6.5pt},
    outer/.style={line width=0.9pt},
    inner/.style={line width=0.9pt},
    rung/.style={line width=0.7pt,dashed,opacity=0.5}
  }

  % ---- parameters ----
  \def\n{15}
  \def\R{4.3}
  \def\r{2.35}
  \def\start{90}

  % ---- vertices (no labels) ----
  \foreach \i in {0,...,14} {
    \pgfmathsetmacro{\ang}{\start - 360/\n*\i}
    \node[vertex] (v\i) at (\ang:\R) {};
    \node[vertex] (w\i) at (\ang:\r) {};
  }

  % ---- edges ----
  \foreach \i in {0,...,14} {
    \pgfmathtruncatemacro{\ip}{mod(\i+1,\n)}
    \draw[outer] (v\i) -- (v\ip);     % outer cycle
    \draw[rung]  (v\i) -- (w\i);      % spokes
    \pgfmathtruncatemacro{\ik}{mod(\i+2,\n)}
    \draw[inner] (w\i) -- (w\ik);     % inner step-2 cycle
  }
\end{tikzpicture}
}

    }
    &
    % --- (2) Double-layer triangular ---
    \resizebox{0.31\textwidth}{!}{\resizebox{0.31\textwidth}{!}{
\begin{tikzpicture}[scale=3.0, x=1cm,y=1cm, line cap=round, line join=round]

  \def\dotSz{5.0pt}  
  \def\edgeW{0.9pt} 
  \tikzset{
    vertex/.style={circle,draw,fill=white,inner sep=0pt,minimum size=\dotSz},
    edge/.style={line width=\edgeW}
  }

  \coordinate (v0)  at (-0.4318,  0.8984);
  \coordinate (v1)  at ( 0.1916,  0.2540);
  \coordinate (v2)  at ( 0.8468, -0.5606);
  \coordinate (v3)  at (-0.5854, -0.5616);
  \coordinate (v4)  at (-0.2724,  0.1560);
  \coordinate (v5)  at ( 0.0472,  1.0000);
  \coordinate (v6)  at (-0.5634, -0.8558);
  \coordinate (v7)  at ( 0.7583, -0.2811);
  \coordinate (v8)  at (-0.8899, -0.4723);
  \coordinate (v9)  at ( 1.0000, -0.0755);
  \coordinate (v10) at ( 0.0384, -0.1859);
  \coordinate (v11) at (-0.1463,  0.7238);
  \coordinate (v12) at ( 0.5874,  0.5625);
  \coordinate (v13) at (-0.8484,  0.5579);
  \coordinate (v14) at (-0.1893, -0.2540);
  \coordinate (v15) at ( 0.4287, -0.9001);
  \coordinate (v16) at (-0.0495, -1.0000);
  \coordinate (v17) at ( 0.2757, -0.1565);
  \coordinate (v18) at (-0.7558,  0.2785);
  \coordinate (v19) at ( 0.5622,  0.8554);
  \coordinate (v20) at (-0.0358,  0.1852);
  \coordinate (v21) at ( 0.1451, -0.7243);
  \coordinate (v22) at ( 0.8892,  0.4721);
  \coordinate (v23) at (-1.0000,  0.0737);

  \draw[edge] (v0) -- (v5);
  \draw[edge] (v0) -- (v11);
  \draw[edge] (v0) -- (v13);
  \draw[edge] (v1) -- (v4);
  \draw[edge] (v1) -- (v10);
  \draw[edge] (v1) -- (v12);
  \draw[edge] (v2) -- (v7);
  \draw[edge] (v2) -- (v9);
  \draw[edge] (v2) -- (v15);
  \draw[edge] (v3) -- (v6);
  \draw[edge] (v3) -- (v8);
  \draw[edge] (v3) -- (v14);
  \draw[edge] (v4) -- (v10);
  \draw[edge] (v4) -- (v18);
  \draw[edge] (v5) -- (v11);
  \draw[edge] (v5) -- (v19);
  \draw[edge] (v6) -- (v8);
  \draw[edge] (v6) -- (v16);
  \draw[edge] (v7) -- (v9);
  \draw[edge] (v7) -- (v17);
  \draw[edge] (v8) -- (v23);
  \draw[edge] (v9) -- (v22);
  \draw[edge] (v10) -- (v21);
  \draw[edge] (v11) -- (v20);
  \draw[edge] (v12) -- (v19);
  \draw[edge] (v12) -- (v22);
  \draw[edge] (v13) -- (v18);
  \draw[edge] (v13) -- (v23);
  \draw[edge] (v14) -- (v17);
  \draw[edge] (v14) -- (v20);
  \draw[edge] (v15) -- (v16);
  \draw[edge] (v15) -- (v21);
  \draw[edge] (v16) -- (v21);
  \draw[edge] (v17) -- (v20);
  \draw[edge] (v18) -- (v23);
  \draw[edge] (v19) -- (v22);

  \foreach \k in {0,...,23} {
    \node[vertex] at (v\k) {};
  }
\end{tikzpicture}
}}
    &
    % --- (3) C60 ---
    \resizebox{0.31\textwidth}{!}{\def\dotSz{4.0pt}   % 节点直径（可调 3.6–4.6pt）
\def\edgeW{0.7pt}   % 线宽（可调 0.6–0.8pt）

\tikzset{
  vertex/.style={circle,draw,fill=white,inner sep=0pt,minimum size=\dotSz},
  edge/.style={line width=\edgeW},
  rung/.style={edge, dashed, opacity=.55},
  cross/.style={edge, dotted, opacity=.55}
}
%tdplot_main_coords, 
\begin{tikzpicture}[scale=3.0, line cap=round, line join=round]
  \tdplotsetmaincoords{70}{120}

  % --- coordinates (UNCHANGED) ---
  \coordinate (v0) at (0.393360,0.341317,0.853681);
  \coordinate (v1) at (0.689077,0.410480,0.597227);
  \coordinate (v2) at (0.580992,0.721482,0.376712);
  \coordinate (v3) at (0.101487,0.604261,0.790297);
  \coordinate (v4) at (0.218480,0.840304,0.496140);
  \coordinate (v5) at (0.296823,-0.053746,0.953419);
  \coordinate (v6) at (0.501987,-0.354610,0.788835);
  \coordinate (v7) at (-0.116126,-0.527527,0.841564);
  \coordinate (v8) at (0.245401,-0.648978,0.720143);
  \coordinate (v9) at (-0.084973,-0.159084,0.983601);
  \coordinate (v10) at (0.902369,0.092418,0.420938);
  \coordinate (v11) at (0.811070,-0.279639,0.513778);
  \coordinate (v12) at (0.958922,-0.263122,-0.105999);
  \coordinate (v13) at (0.847066,-0.497930,0.185864);
  \coordinate (v14) at (0.994323,0.099304,0.038217);
  \coordinate (v15) at (-0.457174,-0.634236,0.623487);
  \coordinate (v16) at (-0.417262,-0.858731,0.297446);
  \coordinate (v17) at (-0.742664,-0.367566,0.559772);
  \coordinate (v18) at (-0.683075,-0.729592,0.033236);
  \coordinate (v19) at (-0.882855,-0.426648,0.196314);
  \coordinate (v20) at (0.811811,-0.322191,-0.486987);
  \coordinate (v21) at (0.561009,-0.612278,-0.557122);
  \coordinate (v22) at (0.707947,-0.016401,-0.706076);
  \coordinate (v23) at (0.299790,-0.486865,-0.820420);
  \coordinate (v24) at (0.390024,-0.116352,-0.913424);
  \coordinate (v25) at (0.287850,-0.882300,0.372409);
  \coordinate (v26) at (0.579733,-0.806870,0.113452);
  \coordinate (v27) at (-0.034721,-0.985364,0.166889);
  \coordinate (v28) at (0.440780,-0.862317,-0.249245);
  \coordinate (v29) at (0.062833,-0.973875,-0.218218);
  \coordinate (v30) at (0.684503,0.728396,-0.029898);
  \coordinate (v31) at (0.883999,0.426050,-0.192428);
  \coordinate (v32) at (0.744264,0.369278,-0.556511);
  \coordinate (v33) at (0.417999,0.859195,-0.295061);
  \coordinate (v34) at (0.458104,0.635595,-0.621417);
  \coordinate (v35) at (-0.299212,0.489361,0.819145);
  \coordinate (v36) at (-0.389203,0.119964,0.913307);
  \coordinate (v37) at (-0.706719,0.018798,0.707244);
  \coordinate (v38) at (-0.561215,0.613235,0.555861);
  \coordinate (v39) at (-0.811726,0.322826,0.486707);
  \coordinate (v40) at (-0.063294,0.974165,0.216785);
  \coordinate (v41) at (0.034734,0.985166,-0.168051);
  \coordinate (v42) at (-0.441915,0.862397,0.246947);
  \coordinate (v43) at (-0.289264,0.881167,-0.373994);
  \coordinate (v44) at (-0.581754,0.804960,-0.116625);
  \coordinate (v45) at (-0.580625,-0.722525,-0.375277);
  \coordinate (v46) at (-0.100710,-0.602514,-0.791729);
  \coordinate (v47) at (-0.218487,-0.839563,-0.497391);
  \coordinate (v48) at (-0.392808,-0.340133,-0.854407);
  \coordinate (v49) at (-0.688653,-0.411320,-0.597137);
  \coordinate (v50) at (0.084978,0.161757,-0.983165);
  \coordinate (v51) at (0.115937,0.529147,-0.840573);
  \coordinate (v52) at (-0.246837,0.647420,-0.721054);
  \coordinate (v53) at (-0.296626,0.055000,-0.953409);
  \coordinate (v54) at (-0.502866,0.353142,-0.788933);
  \coordinate (v55) at (-0.994181,-0.101305,-0.036638);
  \coordinate (v56) at (-0.959395,0.261347,0.106106);
  \coordinate (v57) at (-0.848785,0.494055,-0.188343);
  \coordinate (v58) at (-0.902126,-0.094891,-0.420909);
  \coordinate (v59) at (-0.811467,0.276143,-0.515041);

  % --- edges (UNCHANGED, draw first) ---
  \draw[edge] (v0) -- (v1);
  \draw[edge] (v0) -- (v3);
  \draw[edge] (v0) -- (v5);
  \draw[edge] (v1) -- (v2);
  \draw[edge] (v1) -- (v10);
  \draw[edge] (v2) -- (v4);
  \draw[edge] (v2) -- (v30);
  \draw[edge] (v3) -- (v4);
  \draw[edge] (v3) -- (v35);
  \draw[edge] (v4) -- (v40);
  \draw[edge] (v5) -- (v6);
  \draw[edge] (v5) -- (v9);
  \draw[edge] (v6) -- (v8);
  \draw[edge] (v6) -- (v11);
  \draw[edge] (v7) -- (v8);
  \draw[edge] (v7) -- (v9);
  \draw[edge] (v7) -- (v15);
  \draw[edge] (v8) -- (v25);
  \draw[edge] (v9) -- (v36);
  \draw[edge] (v10) -- (v11);
  \draw[edge] (v10) -- (v14);
  \draw[edge] (v11) -- (v13);
  \draw[edge] (v12) -- (v13);
  \draw[edge] (v12) -- (v14);
  \draw[edge] (v12) -- (v20);
  \draw[edge] (v13) -- (v26);
  \draw[edge] (v14) -- (v31);
  \draw[edge] (v15) -- (v16);
  \draw[edge] (v15) -- (v17);
  \draw[edge] (v16) -- (v18);
  \draw[edge] (v16) -- (v27);
  \draw[edge] (v17) -- (v19);
  \draw[edge] (v17) -- (v37);
  \draw[edge] (v18) -- (v19);
  \draw[edge] (v18) -- (v45);
  \draw[edge] (v19) -- (v55);
  \draw[edge] (v20) -- (v21);
  \draw[edge] (v20) -- (v22);
  \draw[edge] (v21) -- (v23);
  \draw[edge] (v21) -- (v28);
  \draw[edge] (v22) -- (v24);
  \draw[edge] (v22) -- (v32);
  \draw[edge] (v23) -- (v24);
  \draw[edge] (v23) -- (v46);
  \draw[edge] (v24) -- (v50);
  \draw[edge] (v25) -- (v26);
  \draw[edge] (v25) -- (v27);
  \draw[edge] (v26) -- (v28);
  \draw[edge] (v27) -- (v29);
  \draw[edge] (v28) -- (v29);
  \draw[edge] (v29) -- (v47);
  \draw[edge] (v30) -- (v31);
  \draw[edge] (v30) -- (v33);
  \draw[edge] (v31) -- (v32);
  \draw[edge] (v32) -- (v34);
  \draw[edge] (v33) -- (v34);
  \draw[edge] (v33) -- (v41);
  \draw[edge] (v34) -- (v51);
  \draw[edge] (v35) -- (v36);
  \draw[edge] (v35) -- (v38);
  \draw[edge] (v36) -- (v37);
  \draw[edge] (v37) -- (v39);
  \draw[edge] (v38) -- (v39);
  \draw[edge] (v38) -- (v42);
  \draw[edge] (v39) -- (v56);
  \draw[edge] (v40) -- (v41);
  \draw[edge] (v40) -- (v42);
  \draw[edge] (v41) -- (v43);
  \draw[edge] (v42) -- (v44);
  \draw[edge] (v43) -- (v44);
  \draw[edge] (v43) -- (v52);
  \draw[edge] (v44) -- (v57);
  \draw[edge] (v45) -- (v47);
  \draw[edge] (v45) -- (v49);
  \draw[edge] (v46) -- (v47);
  \draw[edge] (v46) -- (v48);
  \draw[edge] (v48) -- (v49);
  \draw[edge] (v48) -- (v53);
  \draw[edge] (v49) -- (v58);
  \draw[edge] (v50) -- (v51);
  \draw[edge] (v50) -- (v53);
  \draw[edge] (v51) -- (v52);
  \draw[edge] (v52) -- (v54);
  \draw[edge] (v53) -- (v54);
  \draw[edge] (v54) -- (v59);
  \draw[edge] (v55) -- (v56);
  \draw[edge] (v55) -- (v58);
  \draw[edge] (v56) -- (v57);
  \draw[edge] (v57) -- (v59);
  \draw[edge] (v58) -- (v59);

  % --- nodes LAST (cover crossings) ---
  \foreach \i in {0,...,59} \node[vertex] at (v\i) {};
\end{tikzpicture}}
    \\
    \small $GP(15,2)$ &
    \small Double-layer triangular &
    \small Truncated icosahedron C$_{60}$
  \end{tabular}
  \caption{Three structured graph benchmarks.}
  \label{fig:three-structured-graphs}
\end{figure*}
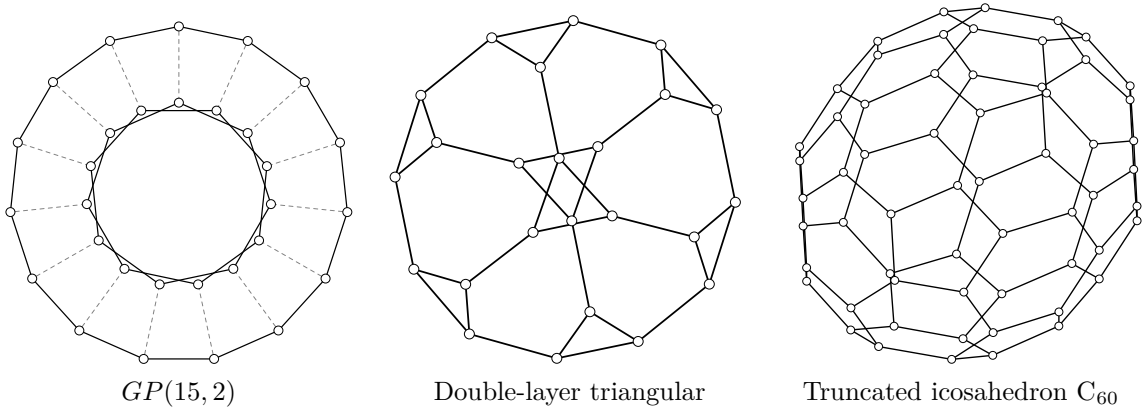

\subsection{Evaluation Protocol and Classical Baselines}

\paragraph{Metric and evaluator.}
We report the normalized expected cut fraction $\langle C\rangle/|E|$. Expectations at fixed round $p$ are computed exactly using a $p$-local evaluator: for each edge, we assemble the contribution from its radius-$p$ neighborhood via tree-decomposition DP and sum over edges.
Our complexity bounds in Sec.~\ref{QAOADP} rely on the near-linear-time $2$-approximation algorithm for treewidth~\cite{Korhonen2021TwoApproxTW}. 
In the experiments of Sec.~\ref{sec:experiment}, however, we instead use a simple branch-and-bound treewidth routine to construct tree decompositions, since our benchmark graphs are small. 
This choice only affects preprocessing time and does not change any of the reported expectations.

\paragraph{Classical locality-matched baselines.}
To match the \(p\)-local lightcone of QAOA, we fix the locality budget at \(k=p\) and compare two representative families. 
For each instance and each \(p\), we tune the family-specific parameters and report the larger expected cut fraction as the ``best classical local'' reference. 
Both families are label-invariant and access only radius-\(k\) neighborhoods.

\emph{(i) Threshold-flip \(\mathrm{Th}\text{-}s(\tau)\)}~\cite{hirvonen2014large}.
Starting from i.i.d.\ unbiased Rademacher spins \(\sigma^{(0)}\in\{\pm1\}^V\),
\[
\Pr[\sigma^{(0)}_v=+1]=\Pr[\sigma^{(0)}_v=-1]=\tfrac12 \quad \text{for all } v\in V,
\]
we perform \(s\) synchronous rounds. 
In the abstract \(\mathrm{Th}\text{-}s(\tau)\) family one may parameterize round-\(t\) thresholds by fractions \(\tau_t\in[0,1]\) and set
\[
\sigma^{(t+1)}_v \;=\;
\begin{cases}
-\sigma^{(t)}_v, & \text{if } \#\{u\in N(v): \sigma^{(t)}_u=\sigma^{(t)}_v\}\;\ge\;\theta_v^{(t)},\\[2pt]
\ \ \sigma^{(t)}_v, & \text{otherwise,}
\end{cases}
\qquad
\theta_v^{(t)}=\bigl\lceil \tau_t\, d_v\bigr\rceil,
\]
so that a flip is triggered once a large enough fraction \(\tau_t\) of the neighbors agree with \(\sigma^{(t)}_v\).
After \(s\) rounds, the final assignment \(\sigma^{(s)}\) induces the cut.
This rule is \(s\)-local (information cannot propagate beyond \(s\) hops), so we instantiate it with \(s\le k=p\).

In our experiments we use a discrete variant tailored to the small, bounded-degree benchmark graphs.
Instead of optimizing over all \(\tau_t\in[0,1]\), we work directly with integer thresholds shared by all vertices in a round:
for each round \(t\) we pick an integer \(t_t\in\{1,2,3\}\) and use
\[
\theta_v^{(t)} = t_t \quad \text{for all } v\in V.
\]
Equivalently, this corresponds to restricting \(\tau_t\) to a coarse grid that depends on the degree, \(\tau_t \approx t_t/d_v\).
For each graph and each round \(p\), we search over \(s\in\{1,\dots,p\}\) and all sequences \((t_t)_{t=0}^{s-1}\) with entries in \(\{1,2,3\}\).
For every candidate choice we approximate the expected cut fraction by Monte Carlo, averaging the cut ratio over \(T=5000\) independent draws of the unbiased initial spins.

\emph{(ii) Barak--Marwaha--type \( \mathrm{BM}\text{-}k \) Gaussian-sum rule}~\cite{barak2021classical}.
Sample i.i.d.\ seeds \(g_u\sim\mathcal{N}(0,1)\) for all \(u\in V\).
For each vertex \(v\), compute a radius-\(k\) linear score with distance-dependent weights \(\alpha=(\alpha_0,\ldots,\alpha_k)\):
\[
S_v(\alpha,g)\;=\;\sum_{d=0}^{k}\alpha_d\!\!\!\sum_{\substack{u\in V:\\ \mathrm{dist}(u,v)=d}}\!\!\! g_u,
\qquad 
\sigma_v \;=\; \mathrm{sign}\bigl(S_v(\alpha,g)\bigr),
\]
with the convention \(\mathrm{sign}(0)=+1\).
This is a one-shot radius-\(k\) local rule; we take \(k=p\) and optimize \(\alpha\) up to an overall scale (which does not affect the signs). For convenience we fix \(\alpha_0=1\) and optimize the remaining coordinates by a small black-box search (random restarts followed by simple coordinate-wise refinement).
Because the scores \(\{S_v(\alpha,g)\}\) are jointly Gaussian, the expected cut fraction can be evaluated exactly using the identity
\(\Pr[\mathrm{sign}(X)\neq \mathrm{sign}(Y)]=\arccos(\rho)/\pi\) for jointly Gaussian \(X,Y\) with correlation \(\rho\), rather than by sampling the seeds \(g\).

\emph{Protocol.}
For each graph and round \(p\), we (a) fix \(k=p\); (b) independently tune \((s,(t_t)_{t=0}^{s-1})\) for \(\mathrm{Th}\text{-}s(\tau)\) and \(\alpha\) for \(\mathrm{BM}\text{-}k\) within the ranges described above; (c) evaluate each family’s expected cut fraction (exactly for \(\mathrm{BM}\text{-}k\) and by Monte Carlo for \(\mathrm{Th}\text{-}s\)); and (d) take their pointwise maximum as the ``best classical local'' baseline.
The optimized baseline parameters are summarized in Tabs.~\ref{tab:baseline-params} and~\ref{tab:thr-params}, and the resulting cut fractions appear in Tabs.~\ref{tab:prodK2}, \ref{tab:tri2lift}, and~\ref{tab:mobius}.

\begin{table}[h!]
\centering
\caption{BM-k parameter settings for Sec.~5.}
\label{tab:baseline-params}
\begin{tabular}{r l l l}
\hline
$p$ & Instance & BM-$k$ params $\boldsymbol{\alpha}$ \\
\hline
1 & $GP(15,2)$                    & $\{1,-0.5770\}$        \\
1 & Double-layer triangular 2-lift & $\{1,-0.4800\}$     \\
1 & Truncated icosahedron $C_{60}$ & $\{1,-0.5779\}$ \\
2 & $GP(15,2)$                    & $\{1,-0.7021,0.2391\}$   \\
2 & Double-layer triangular 2-lift & $\{1,-0.5932,0.2842\}$          \\
2 & Truncated icosahedron $C_{60}$ & $\{1,-0.7255,0.2897\}$         \\
3 & $GP(15,2)$                    & $\{1,-0.7209,0.2793,-0.0973\}$    \\
3 & Double-layer triangular 2-lift & $\{1,-0.6428,0.4050,-0.1638\}$  \\
3 & Truncated icosahedron $C_{60}$ & $\{1,-0.7864,0.4270,-0.2092\}$  \\
\hline
\end{tabular}
\end{table}

\begin{table}[h!]
\centering
\caption{Th-s parameter settings for Sec.~5.}
\label{tab:thr-params}
\begin{tabular}{r l l}
\hline
$p$ & Instance &  Th-$s$ thresholds $(t_t)_{t=0}^{s-1}$ \\
\hline
1 & $GP(15,2)$                    & $\{3\}$ \\
1 & Double-layer triangular 2-lift  & $\{3\}$ \\
1 & Truncated icosahedron $C_{60}$& $\{3\}$ \\
2 & $GP(15,2)$                    & $\{2,3\}$ \\
2 & Double-layer triangular 2-lift     & $\{3,3\}$ \\
2 & Truncated icosahedron $C_{60}$& $\{2,3\}$ \\
3 & $GP(15,2)$                    & $\{3,2,3\}$ \\
3 & Double-layer triangular 2-lift     & $\{2,3,3\}$ \\
3 & Truncated icosahedron $C_{60}$& $\{3,2,3\}$ \\
\hline
\end{tabular}
\end{table}

\begin{table}[h!]
\centering
\caption{QAOA and classical local baselines on $GP(15,2)$($n=30,m=45$).}
\label{tab:prodK2}
\begin{tabular}{r l l r r r r}
\hline
$p$ &  $\boldsymbol{\gamma}$ &  $\boldsymbol{\beta}$ & QAOA & BM-$k$ & Th-$s$ \\
\hline
1 & $\{0.3077\}$ & $\{0.3927\}$ & 0.6624 & \textbf{0.6959} & 0.6882  \\
2 & $\{0.2437,\, 0.4431\}$ & $\{0.5159,\,0.2513\}$ & 0.7360 & \textbf{0.7477} & 0.7021 \\
3 & $\{0.2110,\,0.3990,\,0.4685\}$ & $\{0.6090,\,0.4590,\,0.2350\}$ & \textbf{0.7836} & 0.7561 & 0.7589\\
\hline
\end{tabular}
\end{table}

\begin{table}[h!]
\centering
\caption{QAOA and classical local baselines on the double-layer triangular 2-lift ($n=24,m=36$).}
\label{tab:tri2lift}
\begin{tabular}{r l l r r r r}
\hline
$p$ &  $\boldsymbol{\gamma}$ & $\boldsymbol{\beta}$ & QAOA & BM-$k$ & Th-$s$ \\
\hline
1 & $\{0.2851\}$ & $\{0.3481\}$ & \textbf{0.6599} & 0.6587& 0.6448\\
2 & $\{0.2917,\,0.5623\}$ & $\{0.4090,\,0.2408\}$ & \textbf{0.7208} & 0.6964 & 0.6578 \\
3 & $\{0.2014,\,0.4855,\,0.5916\}$ & $\{0.5410,\,0.3187,\,0.1851\}$ & \textbf{0.7456} & 0.7114 & 0.7145 \\
\hline
\end{tabular}
\end{table}

\begin{table}[h!]
\centering
\caption{QAOA and classical local baselines on the truncated icosahedron graph $C_{60}$ ($n=60,m=90$).}
\label{tab:mobius}
\begin{tabular}{r l l r r r r}
\hline
$p$  &  $\boldsymbol{\gamma}$ &  $\boldsymbol{\beta}$ & QAOA & BM-$k$ & Th-$s$  \\
\hline
1 & $\{0.3078\}$ & $\{0.3927\}$ & 0.6925 & \textbf{0.6959} & 0.6870 \\
2 & $\{0.2490,\,0.4451\}$ & $\{0.5252,\,0.2469\}$ & 0.7514 & \textbf{0.7584} & 0.7297 \\
3 & $\{0.2110,\,0.3990,\,0.4685\}$ & $\{0.6090,\,0.4590,\,0.2350\}$ & \textbf{0.7893} & 0.7868 & 0.7801  \\
\hline
\end{tabular}
\end{table}

\begin{figure}[t]
\centering
% Colorblind-safe palette (Okabe–Ito)
\definecolor{OIblue}{RGB}{0,114,178}
\definecolor{OIorange}{RGB}{230,159,0}
\definecolor{OIgreen}{RGB}{0,158,115}
\definecolor{OIverm}{RGB}{213,94,0}
\definecolor{OIpurple}{RGB}{204,121,167}
\definecolor{OIblack}{RGB}{0,0,0}

\begin{tikzpicture}
\begin{axis}[
  width=0.8\linewidth, height=0.5\linewidth,
  xlabel={Round $p$}, ylabel={Cut fraction},
  xmin=1, xmax=3, ymin=0.65, ymax=0.8,
  grid=both, legend columns=2,
  legend style={at={(0.5,1.02)}, anchor=south},
  legend cell align=left
]

% GP(15,2)
\addplot+[very thick, color=OIblue, mark=o, mark options={solid}]
  coordinates {(1,0.6624) (2,0.7360) (3,0.7836)};
\addlegendentry{Generalized Petersen $GP(15,2)$ — QAOA};

\addplot+[very thick, color=OIorange, mark=square*, dashed]
  coordinates {(1,0.6959) (2,0.7477) (3,0.7589)};
\addlegendentry{Generalized Petersen $GP(15,2)$ — Classical};

% 2-lift triangular
\addplot+[very thick, color=OIgreen, mark=o, mark options={solid}]
  coordinates {(1,0.6599) (2,0.7208) (3,0.7456)};
\addlegendentry{Double-layer triangular $2$-lift — QAOA};

\addplot+[very thick, color=OIverm, mark=square*, dashed]
  coordinates {(1,0.6587) (2,0.6964) (3,0.7145)};
\addlegendentry{Double-layer triangular $2$-lift — Classical};

% Möbius ladder
\addplot+[very thick, color=OIpurple, mark=o, mark options={solid}]
  coordinates {(1,0.6925) (2,0.7514) (3,0.7893)};
\addlegendentry{Truncated icosahedron C$_{60}$ — QAOA};

\addplot+[very thick, color=OIblack, mark=square*, dashed]
  coordinates {(1,0.6959) (2,0.7584) (3,0.7868)};
\addlegendentry{Truncated icosahedron C$_{60}$ — Classical};

\end{axis}
\end{tikzpicture}
\caption{Cut fraction versus QAOA round $p$ on three graph families (solid: QAOA; dashed: best classical local).}
\label{fig:qaoa-classical-depth-pgf}
\end{figure}
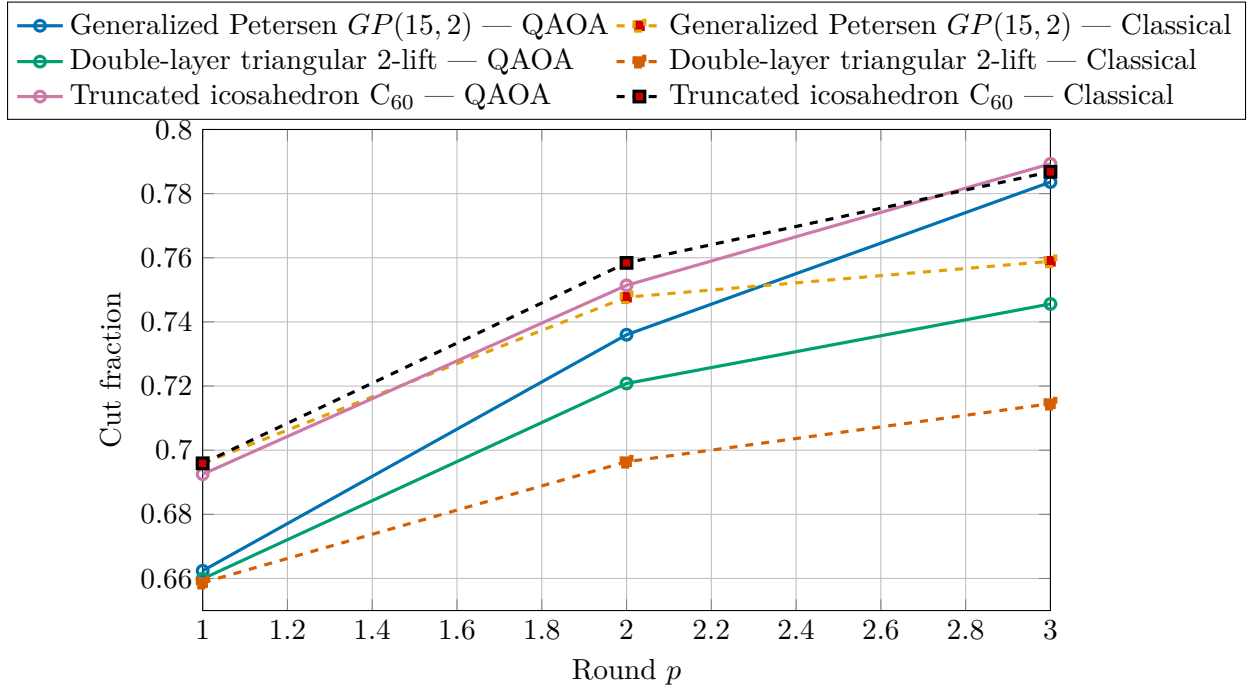

\subsection{Synthesis}

Tabs~\ref{tab:prodK2}, \ref{tab:tri2lift}, and \ref{tab:mobius}, together with Fig.~\ref{fig:qaoa-classical-depth-pgf}, point to three nontrivial features across the generalized Petersen graph \(GP(15,2)\), the double-layer triangular \(2\)-lift,  and the truncated icosahedron graph $C_{60}$.
Note that for \(p=3\) we did not perform exhaustive per-instance angle optimization; instead we used heuristic schedules for 3-regular graphs proposed in~\cite{WLFixedangle} as our initial choices and applied only mild tuning.
Even with this deliberately non-optimized choice, QAOA already exceeds the locality-matched classical baselines (Tabs~\ref{tab:prodK2}, \ref{tab:tri2lift}, \ref{tab:mobius}; Fig.~\ref{fig:qaoa-classical-depth-pgf}).

\medskip\noindent
\textbf{Matched-locality gap depends on the family.}
When QAOA at round \(p\) is compared with classical \(k\)-local baselines at matched lightcone radius (\(k\approx p\)), both the sign and magnitude of the quantum–classical gap vary markedly by graph family. On the double-layer triangular $2$-lift, QAOA shows a consistent advantage already at \(p\le 3\) (cf. Tab~\ref{tab:tri2lift}).
%, suggesting that highly regular, low-variance \(p\)-hop neighborhoods are conducive to a positive matched-locality gap.

\medskip\noindent
\textbf{Crossover depth on \(GP(15,2)\).}
On \(GP(15,2)\), classical \(p\)-local rules are competitive at very small \(p\), whereas
QAOA overtakes at a moderate depth (see Tab.~\ref{tab:prodK2} and
Fig.~\ref{fig:qaoa-classical-depth-pgf}). At depth \(p=1\), both
\(\mathrm{Th}\text{-}s(\tau)\) and \(\mathrm{BM}\text{-}k\) achieve slightly higher cut
fractions than QAOA under the matched locality budget \(k=p\). By \(p=2\), QAOA already
surpasses the threshold-flip baseline, although \(\mathrm{BM}\text{-}k\) still attains the
best cut fraction. At \(p=3\), even with only mildly tuned angles, QAOA yields the
highest cut fraction among all tested methods on this instance. This behaviour exhibits a
genuine crossover in the constant-round regime: very local classical rules can match or
slightly outperform depth-one QAOA, but shallow QAOA gains a clear advantage once a few
layers are allowed while still remaining in the constant-round regime.

% On \(GP(15,2)\), classical \(p\)-local rules are competitive at very small \(p\), whereas QAOA overtakes at a moderate depth (see Tab~\ref{tab:prodK2} and Fig.~\ref{fig:qaoa-classical-depth-pgf}). 
% This may indicates that coordinating phases across the alternating outer/inner/spoke structure requires a minimal depth that strictly \(p\)-local updates do not capture.

\medskip\noindent
\textbf{Emergent advantage on $C_{60}$.}
On the truncated icosahedron $C_{60}$ (3-regular, girth~5), QAOA surpasses the best matched-local classical baselines once the depth reaches a modest threshold: by $p=3$, QAOA yields the highest cut fraction among all $k\!\approx\!p$ methods under our exact, lightcone-matched evaluation protocol (cf.\ Tab~\ref{tab:mobius}).
The absolute margins are small, but the direction is consistent with a depth requirement for coordinating phases across the pentagon–hexagon tiling and resolving local frustration that strictly $k$-local updates cannot capture at shallower radii.
In short, on $C_{60}$ we observe a numerical crossover at moderate depth: classical locality remains competitive in value, yet the best performance is achieved by QAOA once the problem-dependent depth threshold is cleared, aligning with the crossover pattern observed on $GP(15,2)$.

\section{Conclusions}\label{sec:conclusion}

We provided a complexity-algorithm account of exact expectation evaluation for fixed–round QAOA. On the complexity side, for Max-Cut and round $p\ge 2$, we proved NP–hardness of exact evaluation and, further, NP–hardness under $2^{-O(n)}$ additive error; on the algorithmic side, we designed a bag–level dynamic program over $p$–lightcone neighborhoods whose running time is polynomial in the input size and singly–exponential in a local width parameter (the $p$–local treewidth), enabling exact evaluation on families with (hyper)local width growing at most logarithmically. Using a $p$-local evaluator, we benchmarked against locality-matched classical baselines on $GP(15,2)$, a double-layer triangular $2$-lift, and $C_{60}$. We observe modest gaps at shallow $p$ on the triangular $2$-lift, a $GP(15,2)$ crossover with QAOA ahead at $p{=}3$, and a small $p{=}3$ advantage on $C_{60}$ (3-regular, girth 5). These effects are instance-dependent and limited to the tested rounds.

In the future, we will further conduct in-depth research on the following issues: (i) \emph{Rounds under bounded degree.} For bounded-degree graphs at rounds $p=\Theta(\log n)$, either design a polynomial time procedure that evaluates the QAOA expected cut, or show that no such polynomial-time evaluation exists.
% Establish inapproximability thresholds for graphs of bounded maximum degree while allowing $p=\Omega(\log n)$, ideally via degree–reduction gadgets or typical–case lower bounds that preserve locality. 
(ii) \emph{Beyond the standard mixer.} Quantify how alternative mixers (e.g., $XY$, constraint–preserving) and generalized alternating–operator families affect additive/relative–error hardness, and integrate exact $p$–local evaluation with structure–aware parameter search to provide reproducible tuning baselines.

\section*{Acknowledgements}
This work was supported in part by the Quantum Science and Technology-National Science and Technology Major Project under Grant No. 2024ZD0300500 and the National Natural Science Foundation of China Grants No. 62325210, 62272441, 12501450.

\bibliographystyle{quantum}
\bibliography{refs}

@article{Bodlaender1996LinearTW,
  author  = {Bodlaender, H. L.},
  title   = {A linear-time algorithm for finding tree-decompositions of small treewidth},
  journal = {SIAM Journal on Computing},
  volume  = {25},
  pages   = {1305--1317},
  year    = {1996},
  doi     = {10.1137/S0097539793251219},
}

@article{Eppstein2000LocalTW,
  author  = {Eppstein, David},
  title   = {Diameter and treewidth in minor-closed graph families},
  journal = {Algorithmica},
  volume  = {27},
  number  = {3--4},
  pages   = {275--291},
  year    = {2000},
  doi     = {10.1007/s004530010020},
  url     = {https://doi.org/10.1007/s004530010020},
}

@inproceedings{DemaineHajiaghayi2004LocalTW,
  author    = {Demaine, Erik D. and Hajiaghayi, Mohammad Taghi},
  title     = {Equivalence of local treewidth and linear local treewidth and its algorithmic applications},
  booktitle = {Proceedings of the ACM-SIAM Symposium on Discrete Algorithms (SODA)},
  pages     = {840--849},
  year      = {2004},
  publisher = {SIAM},
  url       = {https://dl.acm.org/doi/10.5555/982792.982919},
}

@article{Korhonen2021TwoApproxTW,
  author  = {Korhonen, Tuukka},
  title   = {A single-exponential time 2-approximation algorithm for treewidth},
  journal = {SIAM Journal on Computing},
  year    = {2023},
  doi     = {10.1137/22M147551X},
}

@article{Farhi2014QAOA,
  author        = {Farhi, Edward and Goldstone, Jeffrey and Gutmann, Sam},
  title         = {A quantum approximate optimization algorithm},
  journal       = {arXiv preprint arXiv:1411.4028},
  year          = {2014},
  archivePrefix = {arXiv},
  eprint        = {1411.4028},
  primaryClass  = {quant-ph},
}

@article{Weidenfeller2022Scaling,
  author  = {Weidenfeller, Jakob and Valor, Luis C. and Gacon, Julien and Tornow, Christian and Bello, Lucas and Egger, Daniel J. and Woerner, Stefan},
  title   = {Scaling of the quantum approximate optimization algorithm on superconducting qubit\mbox{-}based hardware},
  journal = {Quantum},
  volume  = {6},
  pages   = {870},
  year    = {2022},
  doi     = {10.22331/q-2022-12-07-870},
}

@article{Marwaha2021LocalVsQAOA,
  author  = {Marwaha, Kushagra},
  title   = {Local classical {Max-Cut} algorithm outperforms $p{=}2$ {QAOA} on high-girth regular graphs},
  journal = {Quantum},
  volume  = {5},
  pages   = {437},
  year    = {2021},
  doi     = {10.22331/q-2021-04-20-437},
}

@article{Hastings2019BoundedDepth,
  author  = {Hastings, Matthew B.},
  title   = {Classical and quantum bounded depth approximation algorithms},
  journal = {Quantum Information and Computation},
  volume  = {19},
  number  = {13--14},
  pages   = {1116--1140},
  year    = {2019},
  doi     = {10.26421/qic19.13-14-3},
}

@article{OumSeymour2006RankWidth,
  author  = {Oum, Sang-il and Seymour, Paul D.},
  title   = {Approximating clique-width and branch-width},
  journal = {Journal of Combinatorial Theory, Series B},
  volume  = {96},
  number  = {4},
  pages   = {514--528},
  year    = {2006},
  doi     = {10.1016/j.jctb.2005.10.006},
}

@article{Hadfield2019QAOAOperator,
  author  = {Hadfield, Stuart and Wang, Zhihui and O'Gorman, Bryan and others},
  title   = {From the quantum approximate optimization algorithm to a quantum alternating operator ansatz},
  journal = {Algorithms},
  volume  = {12},
  number  = {2},
  pages   = {34},
  year    = {2019},
  doi     = {10.3390/a12020034},
}

@article{FarhiHarrow2016QAOASupremacy,
  author        = {Farhi, Edward and Harrow, Aram W.},
  title         = {Quantum supremacy through the quantum approximate optimization algorithm},
  journal       = {arXiv preprint arXiv:1602.07674},
  year          = {2016},
  archivePrefix = {arXiv},
  eprint        = {1602.07674},
  primaryClass  = {quant-ph},
  doi           = {10.48550/arXiv.1602.07674},
}

@inproceedings{Basso2022TQC,
  author    = {Basso, Jacopo and Farhi, Edward and Marwaha, Kushagra and Villalonga, Benjamin and Zhou, Leo},
  title     = {The {QAOA} at high depth for {MaxCut} on large-girth regular graphs and the {SK} model},
  booktitle = {Proceedings of TQC 2022},
  series    = {LIPIcs},
  volume    = {232},
  pages     = {7:1--7:21},
  year      = {2022},
  doi       = {10.4230/LIPIcs.TQC.2022.7},
}

@article{Zhu2022ADAPTQAOA,
  author  = {Zhu, Luning and Tang, Ho Lun and Barron, George S. and Calderon-Vargas, Fernando A. and Mayhall, Nicholas J. and Barnes, Edwin and Economou, Sophia E.},
  title   = {Adaptive quantum approximate optimization algorithm for solving combinatorial problems on a quantum computer},
  journal = {Physical Review Research},
  volume  = {4},
  pages   = {033029},
  year    = {2022},
  doi     = {10.1103/PhysRevResearch.4.033029},
}

@article{Chandarana2022DCQAOA,
  author  = {Chandarana, Pranav and Hegade, Nayana N. and Paul, Kankana and Albarr{\'a}n-Arriagada, Fredy and Solano, Enrique and del Campo, Adolfo and Chen, Xi},
  title   = {Digitized-counterdiabatic quantum approximate optimization algorithm},
  journal = {Physical Review Research},
  volume  = {4},
  pages   = {013141},
  year    = {2022},
  doi     = {10.1103/PhysRevResearch.4.013141},
}

@article{Egger2021WarmStart,
  author  = {Egger, Daniel J. and Mare{\v c}ek, Jakub and Woerner, Stefan},
  title   = {Warm-starting quantum optimization},
  journal = {Quantum},
  volume  = {5},
  pages   = {479},
  year    = {2021},
  doi     = {10.22331/q-2021-06-17-479},
}

@article{Sack2021QAInit,
  author  = {Sack, Stefan H. and Serbyn, Maksym},
  title   = {Quantum annealing initialization of the quantum approximate optimization algorithm},
  journal = {Quantum},
  volume  = {5},
  pages   = {491},
  year    = {2021},
  doi     = {10.22331/q-2021-07-01-491},
}

@article{Krovi2022,
  author        = {Krovi, Hari},
  title         = {Average-case hardness of estimating probabilities of random quantum circuits with a linear scaling in the error exponent},
  journal       = {arXiv preprint arXiv:2206.05642},
  year          = {2022},
  archivePrefix = {arXiv},
  eprint        = {2206.05642},
  primaryClass  = {quant-ph},
  doi           = {10.48550/arXiv.2206.05642},
}

@article{Tate2023Warmest,
  author  = {Tate, Ryan and Moondra, Jay and Gard, Bryan and Mohler, Greg and Gupta, Swati},
  title   = {Warm-started {QAOA} with custom mixers provably converges and computationally beats {Goemans--Williamson}'s {Max-Cut} at low circuit depths},
  journal = {Quantum},
  volume  = {7},
  pages   = {1121},
  year    = {2023},
  doi     = {10.22331/q-2023-09-26-1121},
}

@article{Sureshbabu2024ParamWeighted,
  author  = {Sureshbabu, Sai Hari and Herman, Daniel and Shaydulin, Ruslan and Basso, Jacopo and Chakrabarti, Shouvanik and Sun, Yao and Pistoia, Marco},
  title   = {Parameter setting in quantum approximate optimization of weighted problems},
  journal = {Quantum},
  volume  = {8},
  pages   = {1231},
  year    = {2024},
  doi     = {10.22331/q-2024-01-18-1231},
}

@article{LiSuYangZhang2024,
  author  = {Li, Tongyang and Su, Yiming and Yang, Zhen and Zhang, Shengyu},
  title   = {Quantum approximate optimization algorithms for maximum cut on low-girth graphs},
  journal = {Physical Review Research},
  volume  = {7},
  pages   = {033014},
  year    = {2025},
  doi     = {10.1103/PhysRevResearch.7.033014},
}

@inproceedings{Basso2022FOCS,
  author    = {Basso, Joao and Gamarnik, David and Mei, Song and Zhou, Leo},
  title     = {Performance and limitations of the {QAOA} at constant levels
               on large sparse hypergraphs and spin glass models},
  booktitle = {2022 IEEE 63rd Annual Symposium on Foundations of Computer Science (FOCS)},
  pages     = {335--343},
  year      = {2022},
  doi       = {10.1109/FOCS54457.2022.00039},
  eprint    = {2204.10306},
  archivePrefix = {arXiv},
  primaryClass  = {quant-ph},
}

@article{WurtzLove2021PRA,
  author  = {Wurtz, Jonathan and Love, Peter J.},
  title   = {{MaxCut} quantum approximate optimization algorithm performance guarantees for $p>1$},
  journal = {Physical Review A},
  volume  = {103},
  pages   = {042612},
  year    = {2021},
  doi     = {10.1103/PhysRevA.103.042612},
}

@article{WLFixedangle,
  author  = {Wurtz, Jonathan and Lykov, Denis},
  title   = {Fixed-angle conjectures for the quantum approximate optimization
             algorithm on regular {MaxCut} graphs},
  journal = {Physical Review A},
  volume  = {104},
  number  = {5},
  pages   = {052419},
  year    = {2021},
  doi     = {10.1103/PhysRevA.104.052419},
}

@article{Zhou2020QAOAPRX,
  author  = {Zhou, Leo and Wang, Sheng-Tao and Choi, Soonwon and Pichler, Hannes and Lukin, Mikhail D.},
  title   = {Quantum Approximate Optimization Algorithm: Performance, Mechanism, and Implementation on Near-Term Devices},
  journal = {Physical Review X},
  volume  = {10},
  number  = {2},
  pages   = {021067},
  year    = {2020},
  doi     = {10.1103/PhysRevX.10.021067},
}

@article{Blekos2024QAOAReview,
  author  = {Blekos, Kostas and Brand, Dean and Ceschini, Andrea and Chou, Chiao-Hui and Li, Rui-Hao and Pandya, Komal and Summer, Alessandro},
  title   = {A Review on Quantum Approximate Optimization Algorithm and Its Variants},
  journal = {Physics Reports},
  volume  = {1068},
  pages   = {1--66},
  year    = {2024},
  doi     = {10.1016/j.physrep.2024.03.002},
}

@article{DemaineFominHajiaghayiThilikos2005HMinorFree,
  author  = {Demaine, Erik D. and Fomin, Fedor V. and Hajiaghayi, Mohammad Taghi and Thilikos, Dimitrios M.},
  title   = {Subexponential Parameterized Algorithms on Graphs of Bounded Genus and $H$-Minor-Free Graphs},
  journal = {Journal of the ACM},
  volume  = {52},
  number  = {6},
  pages   = {866--893},
  year    = {2005},
  doi     = {10.1145/1101821.1101823},
}

@article{hirvonen2014large,
  title         = {Large cuts with local algorithms on triangle-free graphs},
  author        = {Hirvonen, Juho and Rybicki, Joel and Schmid, Stefan and Suomela, Jukka},
  journal       = {arXiv preprint arXiv:1402.2543},
  year          = {2014},
  archivePrefix = {arXiv},
  eprint        = {1402.2543},
  primaryClass  = {cs.DC},
}

@article{barak2021classical,
  title         = {Classical algorithms and quantum limitations for maximum cut on high-girth graphs},
  author        = {Barak, Boaz and Marwaha, Kunal},
  journal       = {arXiv preprint arXiv:2106.05900},
  year          = {2021},
  archivePrefix = {arXiv},
  eprint        = {2106.05900},
  primaryClass  = {cs.DS},
}

@article{FruchtGraverWatkins1971GP,
  author  = {Frucht, Roberto and Graver, Jack E. and Watkins, Mark E.},
  title   = {The groups of the generalized Petersen graphs},
  journal = {Mathematical Proceedings of the Cambridge Philosophical Society},
  volume  = {70},
  number  = {2},
  pages   = {211--218},
  year    = {1971},
  doi     = {10.1017/S0305004100049811},
}

@article{BiluLinial2006Lifts,
  author  = {Bilu, Yonatan and Linial, Nathan},
  title   = {Lifts, discrepancy and nearly optimal spectral gap},
  journal = {Combinatorica},
  volume  = {26},
  number  = {5},
  pages   = {495--519},
  year    = {2006},
  doi     = {10.1007/s00493-006-0029-7},
}

@article{Kostant1994TruncIco,
  author  = {Kostant, Bertram},
  title   = {Structure of the truncated icosahedron (such as fullerene or viral coatings) and a 60-element conjugacy class in $\mathrm{PSL}(2,11)$},
  journal = {Proceedings of the National Academy of Sciences of the USA},
  volume  = {91},
  number  = {24},
  pages   = {11714--11717},
  year    = {1994},
  doi     = {10.1073/pnas.91.24.11714},
}

% \onecolumn
% \appendix

% \section{First section of the appendix}
% Quantum allows the usage of appendices.
% If you want your appendices to appear in \texttt{onecolumn} mode but the rest of the
% document in \texttt{twocolumn} mode, you can insert the command
% \texttt{\textbackslash{}onecolumn\textbackslash{}newpage} or just
% \texttt{\textbackslash{}onecolumn} before
% \texttt{\textbackslash{}appendix}.

% \subsection{Subsection}
% Ideally, the command \texttt{\textbackslash{}appendix} should be put before the appendices to get appropriate section numbering.
% The appendices are then numbered alphabetically, with numeric (sub)subsection numbering.
% Equations continue to be numbered sequentially.
% \begin{equation}
%   A \neq B
% \end{equation}
% You are free to change this in case it is more appropriate for your article, but a consistent and unambiguous numbering of sections and equations must be ensured.

% \section{Problems and Bugs}
% In case you encounter problems using the quantumarticle class please analyze the error message carefully and look for help online; \href{http://tex.stackexchange.com/}{http://tex.stackexchange.com/} is an excellent resource.
% If you cannot resolve a problem, please open a bug report in our bug-tracker under \href{https://github.com/quantum-journal/quantum-journal/issues}{https://github.com/quantum-journal/quantum-journal/issues}.
% You can also contact us via email under \href{mailto:latex@quantum-journal.org}{latex@quantum-journal.org}, but it may take significantly longer to get a response.
% In any case, we need the full source of a document that produces the problem and the log file showing the error to help you.

\onecolumn
\appendix

\makeatletter
\renewcommand{\thetheorem}{\Alph{section}.\arabic{theorem}}\@addtoreset{theorem}{section}
\makeatother

\makeatletter
\renewcommand{\theequation}{\Alph{section}.\arabic{equation}}
\@addtoreset{equation}{section}
\makeatother

\section{Recovering \texorpdfstring{$\mathrm{MAXCUT}(G)$}{MAXCUT(G)} from \texorpdfstring{$\mathrm{MAXCUT}(G')$}{MAXCUT(G')}}
\label{app:C}

\renewcommand{\thelemma}{\Alph{section}.\arabic{lemma}}

This appendix proves that, for the graph
\(
G' = (V',E')
\)
obtained by the four-step construction in Sec.~\ref{QAOAhard}, one can recover in polynomial time the maximum cut of the original graph
\(
G = (V,E)
\) from any maximum cut of $G'$. Recall that the construction consists of the following four components:
\begin{enumerate}
  \item \textbf{Vertex gadgets.} For every $u \in V(G)$, replace $u$ by a complete bipartite graph
  \(
    B(u) \cong K_{n_0+1,\,10n_0}\) with \((L_u, R_u)
  \)
  , where $|L_u| = n_0 + 1$ and $|R_u| = 10n_0$.
  \item \textbf{Synchronous edges.} For every original edge $(u,v) \in E(G)$ and each pair \(u_{1,i} \in L_u\), \(v_{1,i} \in L_v\) with $\forall i \in \{0,1,\dots,n_0\}$, add an edge
  \(
    \{u_{1,i}, v_{1,i}\}.
  \)
  \item \textbf{Global bipartite frame.} Add the bipartite graph \(K_{100n_0^2,\,100n_0^2} = (X,Y)\) with \(|X| = |Y| = 100n_0^2 :=M\)
  and connect the special vertices $x_0 \in X$ and $y_0 \in Y$ to all gadget vertices.
  \item \textbf{Anchor and control edges.} Add a vertex $w$ and connect it to $x_0$, $y_0$, and every $u_{1,0}$; in total this gives $n_0+2$ edges.
\end{enumerate}

% For convenience, denote $(U^*,V^*)= \argmax MAXCUT(G')$. 
To show that $\mathrm{MAXCUT}(G)$ can be recovered from $\mathrm{MAXCUT}(G')$ in polynomial time, we first prove three lemmas.

% \subsection*{Lemma C.1 (Gadget bipartition is separated)}

\begin{lemma}(The global frame must be cut as a whole)
In any maximum cut
\[
(S', V'\setminus S')
\]
of $G'$, the two sides of the global bipartite frame must fall into opposite parts of the cut:
\[
X \subseteq S',\ Y \subseteq V'\setminus S'
\quad\text{or}\quad
Y \subseteq S',\ X \subseteq V'\setminus S'.
\]
\end{lemma}

\textit{Proof.}
Let \(X' := X \setminus \{x_0\}, Y' := Y \setminus \{y_0\}\), so that
\[
|X'| = |Y'| = M-1 = 100 n_0^2 - 1.
\]
By construction, except for $x_0,y_0$, vertices in $X',Y'$ have no edges going outside the frame; all outside edges are concentrated on $x_0,y_0$.

\paragraph*{Case (1): $x_0$ and $y_0$ are on opposite sides.}
If we put all of $X$ on the side of $x_0$ and all of $Y$ on the side of $y_0$, then all $M^2$ edges of $K_{M,M}$ are cut. In addition, $x_0,y_0$ cut their external edges in total $n_0(11n_0+1) + 1$ edges. Hence, in this case, we obtain
\begin{equation}
M^2 + n_0(11n_0+1) + 1 .
\label{C.1}
\end{equation}

\paragraph*{Case (2): $x_0$ and $y_0$ are on the same side.}
Assume they are both in $V'\setminus S'$. Write
\[
a := |X' \cap S'|,\qquad b := |X'\setminus S'| = (M-1) - a,
\]
\[
c := |Y' \cap S'|,\qquad d := |Y'\setminus S'| = (M-1) - c.
\]
The cut edges inside the frame consist of:
\begin{enumerate}
  \item edges in $X'\times Y'$: they contribute $ad + bc$;
  \item edges from $x_0$ to $Y' \cap S'$: they contribute $c$;
  \item edges from $y_0$ to $X' \cap S'$: they contribute $a$.
\end{enumerate}
Together with at most $2n_0(11n_0+1) + 2$ external cut edges, we get
\[
f(a,c) \le ad + bc + a + c + 2n_0(11n_0+1) + 2.
\]
Substituting $b = (M-1) - a$ and $d = (M-1) - c$ gives
\begin{align}
f(a,c)
&\le a((M-1)-c) + ((M-1)-a)c + a + c + 2n_0(11n_0+1) + 2 \notag \\
&= (M-1)(a+c) - 2ac + a + c + 2n_0(11n_0+1) + 2 \notag \\
&= M (a+c) - 2ac + 2n_0(11n_0+1) + 2.
\label{c.2}
\end{align}
Because of the negative term $-2ac$, the maximum is obtained on the boundary; taking $a=0$ we have
\begin{equation}
f(0,c) \le M c + 2n_0(11n_0+1) + 2 \le (M-1) M +2n_0(11n_0+1) + 2.
\label{C.3}
\end{equation}
Comparing \eqref{C.1} and \eqref{C.3}, and using $M = 100 n_0^2$, we obtain
\[
\bigl(M^2 + n_0(11n_0+1) + 1\bigr) - \bigl((M-1) M + 2n_0(11n_0+1) + 2\bigr)
= M - n_0(11n_0+1) -1 = 89n_0^2 -n_0- 1 > 0.
\]
Thus every arrangement with $x_0,y_0$ on the same side yields a strictly smaller cut than putting the whole frame on opposite sides. By maximality only Case~(1) can occur. \qed

\begin{lemma}(Gadget bipartition is separated)
For every $u \in V(G)$, let \(B(u) \cong K_{n_0+1,\,10n_0}\) be its gadget with bipartition \((L_u,R_u)\), where \(|L_u|=n_0+1\) and \(|R_u|=10n_0\).
In any maximum cut
\[
(S',V'\setminus S')
\]
of $G'$, the gadget \(B(u)\) is separated along this bipartition, i.e.
\[
L_u \subseteq S',\ R_u \subseteq V'\setminus S'
\quad\text{or}\quad
R_u \subseteq S',\ L_u \subseteq V'\setminus S'.
\]
\end{lemma}

\textit{Proof.}
Fix any original vertex $u \in V(G)$ and its gadget $B(u)$.
By Lemma~A.1, in any maximum cut of $G'$ the special vertices $x_0$ and $y_0$ lie on opposite sides of the cut.
Moreover, by construction every vertex of every gadget $B(u)$ is adjacent to both $x_0$ and $y_0$.
Therefore, for each vertex $v \in B(u)$, exactly one of the edges $\{v,x_0\}$ and $\{v,y_0\}$ is cut, regardless of which side of the cut $v$ lies on.
In particular, if we compare two cuts that differ only in the placement of vertices inside $B(u)$, then the total contribution of all edges between $B(u)$ and $\{x_0,y_0\}$ to the cut size is the same.
Hence, when reasoning about local modifications inside $B(u)$, we may ignore edges from $B(u)$ to $\{x_0,y_0\}$ and only track edges inside $B(u)$ and the remaining external edges of $B(u)$.

By the construction of synchronous and anchor edges, each vertex in $L_u$ is adjacent to all $10n_0$ vertices of $R_u$ and to at most $n_0+2$ vertices outside $B(u)$.
Vertices in $R_u$ are adjacent only to $L_u$ (besides $x_0,y_0$, which we ignore as argued above).

\paragraph*{Step 1: in any maximum cut all of $L_u$ lie on the same side.}
Suppose for contradiction that there exists a maximum cut $(S_0,V'\setminus S_0)$ of $G'$ such that $L_u$ is split between the two sides.
Write
\[
a := |L_u \cap S_0|,\qquad b := |L_u \cap (V'\setminus S_0)|,
\]
and assume without loss of generality that $a \ge b \ge 1$.

First, we move all vertices of $R_u$ to the side $V'\setminus S_0$ without decreasing the cut.
Indeed, consider any $r \in R_u \cap S_0$.
The edges from $r$ to $L_u$ contribute $b$ cut edges before the move (those to $L_u \cap (V'\setminus S_0)$), and $a$ cut edges after we move $r$ to $V'\setminus S_0$ (those to $L_u \cap S_0$).
Since $a \ge b$ and $r$ has no neighbours outside $B(u)$ other than $x_0,y_0$, the cut size does not decrease.
Performing this operation for all $r \in R_u \cap S_0$ yields a maximum cut $(S_1,V'\setminus S_1)$ with
\[
R_u \subseteq V'\setminus S_1
\]
and still $L_u$ split between the two sides.

Now consider any vertex $\ell \in L_u \cap (V'\setminus S_1)$.
Since $B(u)$ is complete bipartite between $L_u$ and $R_u$ and $R_u \subseteq V'\setminus S_1$, $\ell$ has all $10n_0$ neighbours in $R_u$ on the same side as itself.
If we move $\ell$ from $V'\setminus S_1$ to $S_1$, then:
\begin{itemize}
  \item all $10n_0$ edges from $\ell$ to $R_u$ become cut edges (they were not cut before);
  \item at most $n_0+2$ external edges from $\ell$ to vertices outside $B(u)$ may change from cut to non-cut in the worst case.
\end{itemize}
Thus the net change in the cut size is at least
\[
10n_0 - (n_0+2) = 9n_0 - 2 > 0
\]
for every integer $n_0 \ge 1$, contradicting the maximality of $(S_0,V'\setminus S_0)$.
Therefore, in any maximum cut, all vertices in $L_u$ must lie on the same side.

\paragraph*{Step 2: in any maximum cut all of $R_u$ lie on the opposite side.}
Let $(S_2,V'\setminus S_2)$ be a maximum cut.
By Step~1, all vertices in $L_u$ lie on the same side; without loss of generality we assume
\[
L_u \subseteq S_2.
\]
Suppose, for contradiction, that $R_u$ is not contained in a single side.
Then there exists some $q \in R_u \cap S_2$.
Recall that, besides $x_0,y_0$, the vertex $q$ has no neighbours outside $B(u)$; and as argued at the beginning of the proof, edges to $x_0,y_0$ can be ignored when comparing cuts that differ only on $B(u)$.

Before moving $q$, all its neighbours in $L_u$ lie in $S_2$, so none of the \(|L_u| = n_0+1\) edges from $q$ to $L_u$ are cut.
If we move $q$ from $S_2$ to $V'\setminus S_2$, then all these \(n_0+1\) edges become cut edges, and we do not lose any cut edges incident to $q$ outside $B(u)$.
Therefore the cut size strictly increases by at least $n_0+1 > 0$, contradicting the maximality of $(S_2,V'\setminus S_2)$.
Hence $R_u$ cannot be split: in every maximum cut all vertices in $R_u$ lie on the side opposite to $L_u$.

Combining Steps~1 and~2, we conclude that in any maximum cut of $G'$ the gadget $B(u)$ is separated along its bipartition \((L_u,R_u)\), as claimed. \qed

\begin{lemma}(Projection of a maximum cut is a maximum cut of $G$)

Let $(S', V'\setminus S')$ be any maximum cut of $G'$. Define \(S := \{\, u \in V(G) \mid L_u \subseteq S' \,\}.\)
Then
\(
(S,\, V(G)\setminus S)
\)
is a maximum cut of the original graph $G$.
\end{lemma}

\textit{Proof.}
By Lemma~A.2, in every maximum cut of $G'$ each gadget $B(u)$ has its two parts entirely on opposite sides, so the above definition is well-defined: for each $u$ exactly one of $L_u,R_u$ is contained in $S'$.

For every original edge $(u,v) \in E(G)$, the construction adds $n_0+1$ synchronous edges \(\{u_{1,0}, v_{1,0}\}, \dots, \{u_{1,n_0}, v_{1,n_0}\}\), with endpoints in $L_u$ and $L_v$, respectively. If by the above definition $u \in S$ and $v \notin S$, then all these $n_0+1$ edges are cut; if $u,v$ fall on the same side, none of them is cut. Hence the contribution of synchronous edges to the cut value of $G'$ is exactly
\(
(n_0+1)\cdot \mathrm{cut}_G(S).
\)

On the other hand, by Lemma~A.1 and Lemma~A.2 the internal cut value of the frame $K_{100n_0^2,100n_0^2}$ and of every gadget $B(u)$ is already fixed at its maximum. The frame contributes
\(
M^2 + n_0 (11n_0+1) + 1,
\)
and each gadget, when separated, contributes
\(
(n_0+1)\cdot 10 n_0.
\)
There are $n_0$ gadgets, so altogether they contribute
\(
10 n_0^2 (n_0+1).
\)
These parts are independent of $S$ and can be grouped into an explicit polynomial in $n_0$:
\[
b(n_0)
:= \bigl(10000 n_0^4 + n_0 (11n_0+1) + 1\bigr) + 10 n_0^2 (n_0+1).
\]
There remain at most $n_0$ “variable” edges incident to the anchor $w$; collect their total contribution into a remainder $r(S,n_0)$, where $0 \le r(S,n_0) \le n_0$. Thus for the maximum cut $(S',V'\setminus S')$ we have
\begin{equation}
\mathrm{cut}_{G'}(S',V'\setminus S')
= b(n_0) + (n_0+1)\cdot \mathrm{cut}_G(S) + r(S,n_0),
\qquad 0 \le r(S,n_0) \le n_0.
\label{C.5}
\end{equation}

Let $S^\ast \subseteq V(G)$ be an actual maximum cut of $G$, so that $\mathrm{cut}_G(S^\ast) = \mathrm{MAXCUT}(G)$. If we carry out the above construction according to $S^\ast$, we obtain a cut of $G'$ of value
\begin{equation}
b(n_0) + (n_0+1)\cdot \mathrm{MAXCUT}(G) + r(S^\ast,n_0).
\label{C.6}
\end{equation}
Since $(S',V'\setminus S')$ is a \emph{maximum} cut of $G'$, comparing \eqref{C.5} and \eqref{C.6} yields
\[
b(n_0) + (n_0+1)\cdot \mathrm{cut}_G(S) + r(S,n_0)
\;\ge\;
b(n_0) + (n_0+1)\cdot \mathrm{MAXCUT}(G) + r(S^\ast,n_0).
\]
Because $0 \le r(S,n_0), r(S^\ast,n_0) \le n_0$, the above forces
\(
\mathrm{cut}_G(S) \ge \mathrm{MAXCUT}(G),
\)
and hence equality must hold:
\(
\mathrm{cut}_G(S) = \mathrm{MAXCUT}(G).
\)
\qed

\subsection*{Polynomial-time recovery}

From \eqref{C.5} we have
\begin{equation}
\mathrm{MAXCUT}(G')
= b(n_0) + (n_0+1)\cdot \mathrm{MAXCUT}(G) + r(S,n_0),
\qquad 0 \le r(S,n_0) \le n_0.
\end{equation}
Since \(b(n_0) = \bigl(10000 n_0^4 + n_0 (11n_0+1) + 1\bigr) + 10 n_0^2 (n_0+1)\) is explicit and depends only on the construction, we can recover $\mathrm{MAXCUT}(G)$ from $\mathrm{MAXCUT}(G')$ by
\begin{equation}
\mathrm{MAXCUT}(G)
= \left\lfloor \frac{\mathrm{MAXCUT}(G') - b(n_0)}{\,n_0+1\,} \right\rfloor.
\end{equation}
This involves only a constant number of arithmetic operations and is clearly polynomial time. Hence, from any maximum cut of $G'$ we can recover $\mathrm{MAXCUT}(G)$ in polynomial time.

% \subsection*{A.5 Magnitude of the Largest-Term Coefficient (Lemma 3.7)}
% The modulus of the coefficient at the largest exponent is $\Omega(2^{-n})$.

% \noindent\textbf{Proof.}
% In configurations achieving \eqref{3.7}, each non-end layer contributes a factor
% \begin{equation}\label{eq:A5-1}
% \cos^{2n(p-2)}\psi \ge ((1-2^{-n^2}))^{2n(p-2)}\ge 1-2^{-1}=2^{-1}.
% \end{equation}
% End-layer transitions $p\leftrightarrow p-1$ and $-p\leftrightarrow -p+1$ contribute $2^{-n/2}$ each, while $p\leftrightarrow 0$ and $0\leftrightarrow -p$ together contribute $2^{-n}$. By Lemma~3.2 and by construction, layer $0$ can be chosen arbitrarily, giving $2^n$ choices. Writing $cc=2^{-1}$, the magnitude of the largest coefficient scales as
% \begin{equation}\label{eq:A5-2}
% 2^{-n/2}\cdot 2^{-n/2}\cdot 2^{-n}\cdot 2^{n}\cdot cc \;=\; \,2^{-n-1},
% \end{equation}
% i.e.\ $\Theta(2^{-n})$. Since the largest exponent coefficient is an integer multiple of this, its modulus is $\Omega(2^{-n})$.
% \qed

% \bigskip
% \noindent\textbf{Notes.} References to equation numbers such as \emph{(3.7)} and \emph{(3.10)} follow the numbering used in the main text.

\section{ From Complete-Basis Insertions to the Master Formula}\label{app:B}

Inspired by~\cite{Basso2022TQC}, we derive a compact ``master formula'' for two-point observables under fixed-round QAOA that separates on-site mixer amplitudes from two-site cost phases and makes $p$-locality explicit.

\paragraph{Setup.}
Let
\begin{equation}\label{eq:B-psip}
\ket{\psi_p(\bm\gamma,\bm\beta)}
  = \Bigl(\prod_{\ell=1}^{p} e^{-i\beta_\ell B}\,e^{-i\gamma_\ell C}\Bigr)\ket{\bm s}, 
  \qquad \ket{\bm s}=\ket{+}^{\otimes n}.
\end{equation}
For an edge $e=(u,v)$ we consider
\begin{equation}\label{eq:B-target}
\big\langle \bm\gamma,\bm\beta \big| Z_u Z_v \big| \bm\gamma,\bm\beta\big\rangle
= \bra{s}\Bigl(\prod_{\ell=1}^{p} e^{i\gamma_\ell C} e^{i\beta_\ell B}\Bigr)\,Z_u Z_v\,
\Bigl(\prod_{\ell=1}^{p} e^{-i\beta_\ell B} e^{-i\gamma_\ell C}\Bigr)\ket{\bm s}.
\end{equation}

\paragraph{Complete-basis insertions in the $Z$ basis.}
% Insert $(2p{+}1)$ resolutions of the identity between all neighboring operators and around the central observable $Z_uZ_v$, with layer labels $j\in\{-p,\ldots,-1,0,1,\ldots,p\}$. Denote the layer-$j$ spin configuration by $z^{[j]}\in\{\pm1\}^V$, and write $z^{[j]}_x$ for the spin of vertex $x$ at layer $j$. Using that $C$ is diagonal in the $Z$ basis, that $e^{\pm i\beta B}=\bigotimes_{x\in V} e^{\pm i\beta X_x}$ factorizes across sites, and that $\langle s|z\rangle=\langle z|s\rangle=2^{-n/2}$, we obtain
% \begin{align}
% \langle\gamma,\beta|Z_uZ_v|\gamma,\beta\rangle
% &= \frac{1}{2^n}\!
%   \sum_{\{z^{[j]}\}}
%   \exp\!\Bigl(i\!\sum_{\ell=1}^{p}\gamma_\ell\,C\bigl(z^{[\ell]}\bigr)
%               - i\!\sum_{\ell=1}^{p}\gamma_\ell\,C\bigl(z^{[-\ell]}\bigr)\Bigr)\,
%   z_u^{[0]} z_v^{[0]}\nonumber\\
% &\quad\times\prod_{x\in V}
% \Bigl(\langle z_x^{[1]}|e^{i\beta_1 X}|z_x^{[2]}\rangle\cdots
%       \langle z_x^{[p]}|e^{i\beta_p X}|z_x^{[0]}\rangle\,
%       \langle z_x^{[0]}|e^{-i\beta_p X}|z_x^{[-p]}\rangle\cdots
%       \langle z_x^{[-2]}|e^{-i\beta_1 X}|z_x^{[-1]}\rangle\Bigr).
% \label{eq:B-expanded}
% \end{align}

Insert \((2p{+}1)\) resolutions of identity in the \(Z\)-basis between neighboring operators and around \(Z_uZ_v\).
Using that \(C\) is diagonal in the \(Z\)-basis, \(e^{\pm i\bm\beta B}=\bigotimes_v e^{\pm i\bm\beta X_v}\), and \(\bra{\bm s}\bm z\rangle=\bra{\bm z}\bm s\rangle=2^{-n/2}\), we have
% \begin{adjustwidth*}
\begin{equation}\label{eq:ZLZR-master}
\begin{aligned}
&\langle\bm\gamma, \bm\beta| Z_uZ_v|\bm\gamma, \bm\beta\rangle 
\\&= \langle \bm s|e^{i\gamma_1 C} e^{i\beta_1 B}...e^{i\gamma_p C}e^{i\beta_p B}Z_uZ_ve^{i\gamma_{-p} C}e^{i\beta_{-p} B}... e^{-i\beta_1 B} e^{-i\gamma_1 C}|\bm s\rangle
\\&=\sum_{\{\bm z^{[i]}\}}\langle \bm s|\bm z^{[1]}\rangle e^{i\gamma_1 C(\bm z^{[1]})}\langle \bm z^{[1]}|e^{i\beta_1 B}...|\bm z^{[p]}\rangle e^{i\gamma_p C(\bm z^{[p]})}\langle \bm z^{[p]}|e^{i\beta_p B}|\bm z^{[0]}\rangle z_u^{[0]}z_v^{[0]} 
\\ &\quad\times\langle \bm z^{[0]}|...e^{-i\beta_2 B}|\bm z^{[-2]}\rangle e^{-i\gamma_2 C(z^{[-2]})}\langle \bm z^{[-2]}|e^{-i\beta_1 B}|\bm z^{[-1]}\rangle e^{-i\gamma_1 C(\bm z^{[-1]})}\langle \bm z^{[-1]}|\bm s\rangle 
\\ &=\frac{1}{2^n}\sum_{\{\bm z^{[i]}\}}\exp\left[i\gamma_1 C(\bm z^{[1]}) +...+i\gamma_p C(\bm z^{[p]})-i\gamma_{p} C(\bm z^{[-p]})-... - i\gamma_1 C(\bm z^{[-1]})\right] z_u^{[0]}
\\ &\quad\times z_v^{[0]} \prod_{v = 1}^n\langle z_v^{[1]}|e^{i\beta_1 X}|z_v^{[2]}\rangle\langle z_v^{[2]}|e^{i\beta_2 X}...|z_v^{[0]}\rangle\langle z_v^{[0]}|...e^{-i\beta_2 X}|z_v^{[-2]}\rangle\langle z_v^{[-2]}|e^{-i\beta_1 X}|z_v^{[-1]}\rangle.  
\end{aligned}
\end{equation}

\paragraph{Angle vector and one-site kernel.}
With 
\[
  f(\bm z_v)
  = \frac{1}{2}\,
    \langle z_v^{[1]}|e^{i\beta_1 X}|z_v^{[2]}\rangle\cdots
    \langle z_v^{[p]}|e^{i\beta_p X}|z_v^{[0]}\rangle\,
    \langle z_v^{[0]}|e^{-i\beta_p X}|z_v^{[-p]}\rangle\cdots
    \langle z_v^{[-2]}|e^{-i\beta_1 X}|z_v^{[-1]}\rangle .
\]
and 
\[
\bm\Gamma=\{\gamma_1,\dots,\gamma_p,0,-\gamma_p,\dots,-\gamma_1\}
\]
defined above, \eqref{eq:B-target} becomes the vertex-factorized form
\begin{equation}\label{eq:B-master-global}
% \boxed{~~
\langle\bm\gamma,\bm\beta|Z_uZ_v|\bm\gamma,\bm\beta\rangle
= \frac{1}{2^n}\!\sum_{\{\bm z^{[j]}\}}
\exp\!\Bigl(i\sum_{j=-p}^{p}\Gamma_j\,C\bigl(\bm z^{[j]}\bigr)\Bigr)\,
\bigl(z_u^{[0]} z_v^{[0]}\bigr)\cdot\prod_{x\in V} f\!\bigl(\bm z_x\bigr).
% ~~}
\end{equation}

\paragraph{$p$-locality and restriction to the lightcone.}
Let $H = G_p((u,v))$ be the $p$-local graph of the edge $(u,v)$ as defined in Sec.~\ref{sec:2.1},
with vertex set $V(H) = N_p((u,v))$ and edge set given by~\eqref{eq:2.2}. Since the Heisenberg-evolved observable for the edge $e=(u,v)$ is supported on $V(H)$, the phase term in \eqref{eq:B-master-global} depends only on edges of $H$, and the mixer kernels factorize over $V\setminus V(H)$. Tracing out spins outside $H$ therefore leaves
\begin{equation}\label{eq:B-master-local}
% \boxed{~~
\langle\bm\gamma,\bm\beta|Z_uZ_v|\bm\gamma,\bm\beta\rangle
= \sum_{\{\bm z_x:x\in V(H)\}}
\bigl(z_u^{[0]} z_v^{[0]}\bigr)\,
\exp\!\Bigl(i\sum_{j=-p}^{p}\Gamma_j\,C_H\bigl(\bm z^{[j]}\bigr)\Bigr)\,
\prod_{x\in V(H)} f\!\bigl(\bm z_x\bigr),
% ~~}
\end{equation}
where $C_H$ is the unnormalized cut on $H$ and $\bm z^{[j]}$ is restricted to $V(H)$. \eqref{eq:B-master-local} is the localized master formula used by our dynamic programs.

\section{Lemmas and Proof Details}\label{app:A}

For completeness, we collect the proof sketches of several technical lemmas used in Sec.~\ref{QAOAhard} (the presentation and notation are lightly normalized but the mathematical content is unchanged). Throughout, $Z_{x_0}Z_{y_0}$ denotes the two-qubit observable acting on endpoints $(x_0,y_0)$ and $\mathrm{MC}=\mathrm{MAXCUT}(G')$.

\subsection*{C.1\quad Endpoint Consistency (Lemma 3.2)}

\noindent\textbf{Statement.}
For any vertex $q$, the dependence of the contribution to $\langle Z_uZ_v\rangle$ on $z_q^{[0]}$ falls into the following cases:
\begin{enumerate}
  \item If $q\notin\{u,v\}$ and $z_q^{[p]}\neq z_q^{[-p]}$, then replacing $z_q^{[0]}$ by $-z_q^{[0]}$ flips the sign of the contribution. The two configurations related by this flip therefore cancel in the sum.
  \item If $q\in\{u,v\}$ and $z_q^{[p]}=z_q^{[-p]}$, then flipping $z_q^{[0]}$ again flips the sign of the contribution, so the corresponding pair of configurations cancels.
  \item In all remaining cases (that is, when $q\notin\{u,v\}$ with $z_q^{[p]}=z_q^{[-p]}$, or $q\in\{u,v\}$ with $z_q^{[p]}\neq z_q^{[-p]}$), the contribution is independent of $z_q^{[0]}$.
\end{enumerate}

\medskip
\noindent\textbf{Proof.}
Denote by $W$ the multiplicative factor from all spins and layers other than the two end-layer mixers acting on $q$ (and, when $q\in\{u,v\}$, excluding the explicit observable factor $z_u^{[0]}z_v^{[0]}$). With $\beta_{p}=\frac{\pi}{4}$, the one-qubit mixer matrix elements in the $Z$-basis are
\[
\langle a|e^{i\beta_{p}X}|b\rangle=\begin{cases}
\cos(\frac{\pi}{4})=\tfrac{1}{\sqrt2}, & a=b,\\[2pt]
i\sin(\frac{\pi}{4})=\tfrac{i}{\sqrt2}, & a=-b,
\end{cases}\qquad a,b\in\{\pm1\}.
\]

\smallskip
\emph{Case 1: $q\notin\{u,v\}$ and $z_q^{[p]}\neq z_q^{[-p]}$.}
Fix all other spins and compare the two configurations that differ only by $z_q^{[0]}=+1$ vs.\ $-1$.
The $z_q^{[0]}$-dependent factor is the product of the two end-layer mixers on $q$.
Writing the remaining (fixed) part as $W$, the two contributions are
\begin{equation}\label{eq:A1-1}
W\,\cos\!\Bigl(\tfrac{\pi}{4}\Bigr)\!\left(\mp\, i\,\sin\!\Bigl(\tfrac{\pi}{4}\Bigr)\right),
\end{equation}
which are negatives of each other and sum to zero. This proves Item~1.

\smallskip
\emph{Case 2: $q\in\{u,v\}$ and $z_q^{[p]}=z_q^{[-p]}$.}
Here the observable contributes an extra factor $z_u^{[0]}z_v^{[0]}$, which flips sign when $z_q^{[0]}$ is flipped (exactly one endpoint bit changes). With the other-spin factor denoted by $W$, the two contributions for $z_q^{[0]}=\pm1$ are
\begin{equation}\label{eq:A1-2}
% \pm\, W\,\cos\!\Bigl(\tfrac{\pi}{4}\Bigr)\!\left(i\,\sin\!\Bigl(\tfrac{\pi}{4}\Bigr)\right),
W\cos^2(\tfrac{\pi}{4})\;\text{and}\;-W\sin^2(\tfrac{\pi}{4}),
\end{equation}
when $z_q^{[p]}=z_q^{[-p]}=1$ and
\begin{equation}\label{eq:A1-2}
W-\cos^2(\tfrac{\pi}{4})\;\text{and}\;W\sin^2(\tfrac{\pi}{4}),
\end{equation}
when $z_q^{[p]}=z_q^{[-p]}=-1$. Again negatives of each other and hence cancel. This proves Item~2.

\smallskip
\emph{Case 3: Remaining two situations.}
\begin{enumerate}
\item[(3a)] $q\notin\{u,v\}$ with $z_q^{[p]}=z_q^{[-p]}$.  
For any valid configuration with contribution $W$, flipping $z_q^{[0]}$ changes the mixer product on $q$ from $\cos^2(\frac{\pi}{4})$ to $-(i\sin(\frac{\pi}{4}))^2=\sin^2(\frac{\pi}{4})$ or vice versa. Thus the new contribution equals
\begin{equation}\label{eq:A1-3}
W\cdot\frac{-\,i^{2}\,\sin^{2}(\frac{\pi}{4})}{\cos^{2}(\frac{\pi}{4})}=W,
\end{equation}
or, symmetrically,
\begin{equation}\label{eq:A1-4}
W\cdot\frac{\cos^{2}(\frac{\pi}{4})}{-\,i^{2}\,\sin^{2}(\frac{\pi}{4})}=W,
\end{equation}
so the value is independent of $z_q^{[0]}$.

\item[(3b)] $q\in\{u,v\}$ with $z_q^{[p]}\neq z_q^{[-p]}$.  
Let the original contribution be $W$.
When $(z_q^{[p]},z_q^{[0]},z_q^{[-p]})=(1,1,-1)$ or $(1,-1,-1)$, flipping $z_q^{[0]}$ gives
\begin{equation}\label{eq:A1-5}
W\cdot\frac{-\,i\,\sin(\frac{\pi}{4})\,\cos(\frac{\pi}{4})}{-\,i\,\sin(\frac{\pi}{4})\,\cos(\frac{\pi}{4})}=W,
\end{equation}
and when $(z_q^{[p]},z_q^{[0]},z_q^{[-p]})=(-1,-1,1)$ or $(-1,1,1)$, flipping $z_q^{[0]}$ gives
\begin{equation}\label{eq:A1-6}
W\cdot\frac{i\,\sin(\frac{\pi}{4})\,\cos(\frac{\pi}{4})}{i\,\sin(\frac{\pi}{4})\,\cos(\frac{\pi}{4})}=W.
\end{equation}
In both subcases the contribution is unchanged, hence independent of $z_q^{[0]}$.
\end{enumerate}

Combining the three cases establishes the claim. \hfill\qed

\subsection*{C.2\quad Asynchronous Flips (Lemma 3.3)}
\noindent\textbf{Statement.}
For the observable \(Z_{x_0}Z_{y_0}\), if \(z_{x_0}^{[p]}\neq z_{y_0}^{[p]}\), then there is no contribution to coefficients with exponents
\[
k\ge \Bigl((p-1)\,\mathrm{MC}+200n_0^2+2n_0(11n_0+1)\Bigr).
\]

\medskip

\noindent\textbf{Proof.}
Fix a configuration \(\{ \bm z \} = (z_v^{[j]})_{v\in V(G'),\, j=-p,\dots,p}\) with
\(z_{x_0}^{[p]}\neq z_{y_0}^{[p]}\).
As in Eq.~(3.3), the exponent contributed by \(\{\bm z\}\) is
\[
D(\bm z)
=\sum_{j=1}^p\bigl(C_j(\bm z)-C_{-j}(\bm z)\bigr),
\]
where \(C_j(\bm z)\) denotes the cut value of \(G'\) on layer \(j\). Since every \(C_j(\bm z)\) is a cut of \(G'\), we have \(0\le C_j(\bm z)\le \mathrm{MC}\) for all \(j\), where \(\mathrm{MC}=\mathrm{MAXCUT}(G')\).
In particular, for the inner layers \(j=1,\dots,p-1\) we obtain
\[
C_j(\bm z)-C_{-j}(\bm z)\le \mathrm{MC},
\]
and hence
\[
D(\bm z)
=\sum_{j=1}^{p-1}\bigl(C_j(\bm z)-C_{-j}(\bm z)\bigr)
 + \bigl(C_p(\bm z)-C_{-p}(\bm z)\bigr)
\le (p-1)\,\mathrm{MC}+\bigl(C_p(\bm z)-C_{-p}(\bm z)\bigr).
\]
Therefore, to show that no configuration with \(z_{x_0}^{[p]}\neq z_{y_0}^{[p]}\) can contribute to a coefficient with exponent
\[
k\ge (p-1)\,\mathrm{MC}+200n_0^2+2n_0(11n_0+1),
\]
it suffices to prove that, under this hypothesis,
\begin{equation}\label{eq:async-Cp-gap}
C_p(\bm z)-C_{-p}(\bm z)
<200n_0^2+2n_0(11n_0+1).
\end{equation}

We now analyze \(C_p(\bm z)-C_{-p}(\bm z)\) by exploiting endpoint consistency (Lemma~3.2).
Recall that for any vertex \(q\notin\{x_0,y_0\}\), if \(z_q^{[p]}\neq z_q^{[-p]}\), then flipping \(z_q^{[0]}\) reverses the sign of the contribution of \(\{\bm z\}\) without changing the exponent.
Thus, in the coefficient of any fixed Laurent exponent, all such configurations cancel pairwise.
Consequently, when bounding \eqref{eq:async-Cp-gap} for nonvanishing coefficients, we may restrict attention to configurations obeying
\begin{equation}\label{eq:nonend-eq}
z_q^{[p]}=z_q^{[-p]}\quad\text{for all }q\notin\{x_0,y_0\}.
\end{equation}

Applying Lemma~3.2 with \(q\in\{x_0,y_0\}\) shows that, in order to avoid cancellation, we must have
\(z_{x_0}^{[p]}\neq z_{x_0}^{[-p]}\) and \(z_{y_0}^{[p]}\neq z_{y_0}^{[-p]}\); otherwise flipping the corresponding layer-\(0\) spin would again flip the sign of the contribution and force cancellation.
Since we are in the case \(z_{x_0}^{[p]}\neq z_{y_0}^{[p]}\), we have
\(z_{x_0}^{[p]}z_{y_0}^{[p]}=-1\).
Multiplying the two inequalities
\(z_{x_0}^{[p]}\neq z_{x_0}^{[-p]}\) and \(z_{y_0}^{[p]}\neq z_{y_0}^{[-p]}\) yields
\[
z_{x_0}^{[p]}z_{x_0}^{[-p]}=-1,\qquad
z_{y_0}^{[p]}z_{y_0}^{[-p]}=-1,
\]
so
\[
\bigl(z_{x_0}^{[p]}z_{y_0}^{[p]}\bigr)\bigl(z_{x_0}^{[-p]}z_{y_0}^{[-p]}\bigr)
=1,
\]
and hence also \(z_{x_0}^{[-p]}\neq z_{y_0}^{[-p]}\).
In particular, the spins at \(x_0\) and \(y_0\) are opposite on both boundary layers \(\pm p\).

Recall that
\[
C_j(\bm z)
=\sum_{(u,v)\in E(G')}
\frac{1}{2}\bigl(1-z_u^{[j]}z_v^{[j]}\bigr),
\]
so the difference \(C_p(\bm z)-C_{-p}(\bm z)\) receives a nonzero contribution from an edge \((u,v)\) only if its cut/non-cut status differs between layers \(p\) and \(-p\).
We now consider edges by their relation to \(\{x_0,y_0\}\).

\smallskip\noindent
\emph{Edges not incident to \(x_0\) or \(y_0\).}
For such an edge \((u,v)\) both endpoints satisfy \eqref{eq:nonend-eq}, so
\(z_u^{[p]}=z_u^{[-p]}\) and \(z_v^{[p]}=z_v^{[-p]}\).
Therefore the product \(z_u^{[j]}z_v^{[j]}\) is identical for \(j=p\) and \(j=-p\), and the edge contributes equally to \(C_p(\bm z)\) and \(C_{-p}(\bm z)\).
Hence these edges contribute zero to \(C_p(\bm z)-C_{-p}(\bm z)\).

\smallskip\noindent
\emph{Edges between \(\{x_0,y_0\}\) and \(V(G')\setminus(X\cup Y\cup\{x_0,y_0\})\).}
By construction of \(G'\), every vertex outside the global bipartite frame \(X\cup Y\) that is adjacent to \(x_0\) or \(y_0\) is in fact adjacent to \emph{both} \(x_0\) and \(y_0\) (these are the vertices in the blow-up gadgets \(B(u)\) for \(u\in V(G)\), together with the controller \(w\)).
Let \(u\) be such a vertex.
Then \eqref{eq:nonend-eq} gives \(z_u^{[p]}=z_u^{[-p]}\), while the previous paragraph shows that
\(z_{x_0}^{[-p]}=-\,z_{x_0}^{[p]}\) and \(z_{y_0}^{[-p]}=-\,z_{y_0}^{[p]}\).
Writing \(s=z_u^{[p]}\) and choosing signs so that \(z_{x_0}^{[p]}=+1\), \(z_{y_0}^{[p]}=-1\), we have
\[
\begin{aligned}
\frac{1}{2}\bigl(1-z_u^{[p]}z_{x_0}^{[p]}\bigr)
+\frac{1}{2}\bigl(1-z_u^{[p]}z_{y_0}^{[p]}\bigr)
&=\frac{1}{2}\bigl(1-s\bigr)+\frac{1}{2}\bigl(1+s\bigr)=1,\\[2mm]
\frac{1}{2}\bigl(1-z_u^{[-p]}z_{x_0}^{[-p]}\bigr)
+\frac{1}{2}\bigl(1-z_u^{[-p]}z_{y_0}^{[-p]}\bigr)
&=\frac{1}{2}\bigl(1+s\bigr)+\frac{1}{2}\bigl(1-s\bigr)=1.
\end{aligned}
\]
Thus the \emph{pair} of edges \((u,x_0)\) and \((u,y_0)\) contributes the same total amount to \(C_p(\bm z)\) and to \(C_{-p}(\bm z)\), and hence contributes zero to their difference.
Summing over all such vertices \(u\), we conclude that edges between \(\{x_0,y_0\}\) and \(V(G')\setminus(X\cup Y\cup\{x_0,y_0\})\) do not affect \(C_p(\bm z)-C_{-p}(\bm z)\).

\smallskip\noindent
\emph{Edges within the global frame \(X\cup Y\).}
We are left with edges whose endpoints lie in the bipartite frame.
Among these, the special edge \((x_0,y_0)\) has endpoints that are opposite on both layers \(p\) and \(-p\), so it is cut on both layers and contributes equally to \(C_p(\bm z)\) and \(C_{-p}(\bm z)\); hence its contribution to the difference again vanishes.

All remaining edges incident to \(x_0\) or \(y_0\) have the form \((x_0,y)\) with \(y\in Y\setminus\{y_0\}\) or \((x,y_0)\) with \(x\in X\setminus\{x_0\}\).
There are exactly
\[
(|Y|-1)+(|X|-1)=(100n_0^2-1)+(100n_0^2-1)=200n_0^2-2
\]
such edges.
For each of these edges, the neighbor in \(X\cup Y\) has the same spin on layers \(p\) and \(-p\), while \(x_0\) or \(y_0\) flips between the two layers.
Thus, in the most favorable situation for maximizing \(C_p(\bm z)-C_{-p}(\bm z)\), we may assume that every such edge is cut at layer \(p\) and uncut at layer \(-p\), contributing \(1\) to the difference.
Therefore,
\[
C_p(\bm z)-C_{-p}(\bm z)\le 200n_0^2-2.
\]

This immediately implies
\[
C_p(\bm z)-C_{-p}(\bm z)
\le 200n_0^2-2
<200n_0^2+2n_0(11n_0+1),
\]
which is exactly \eqref{eq:async-Cp-gap}.
Combining this with the earlier bound
\(D(\bm z)\le(p-1)\,\mathrm{MC}+\bigl(C_p(\bm z)-C_{-p}(\bm z)\bigr)\), we obtain
\[
D(\bm z) < (p-1)\,\mathrm{MC}+200n_0^2+2n_0(11n_0+1)
\]
for every configuration with \(z_{x_0}^{[p]}\neq z_{y_0}^{[p]}\) that survives the cancellations enforced by Lemma~3.2.
Hence no such configuration can contribute to coefficients with exponent
\(k\ge(p-1)\,\mathrm{MC}+200n_0^2+2n_0(11n_0+1)\),
as claimed.
\hfill\(\Box\)

\medskip

\subsection*{C.3 Attainability of the Largest Nonzero Exponent (Lemma 3.4)}

\noindent\textbf{Statement.}
For the observable $Z_{x_0}Z_{y_0}$, the largest nonzero exponent in $h_{(x_0,y_0)}$ attainable by the induced polynomial is
\begin{equation}
\Bigl((p-1)\,\mathrm{MC}+200n_0^2+2n_0(11n_0+1)\Bigr).
\end{equation}

We prove two parts: (i) \emph{Maximality}: the upper bound of the exponent in the display above is attainable; (ii) \emph{Non-vanishing}: the coefficient of the term at that exponent is nonzero.

\medskip
\noindent\textbf{Notation.}
For any configuration $\{\bm z\}$, denote its monomial contribution to the exponent term by
\[
\mathrm{Contr}(\bm z)\;:=\; f(\bm z)\,z_{x_0}^{[0]}z_{y_0}^{[0]},
\]
and the corresponding exponent by
\[
D(\bm z)\;:=\; \sum_{j=1}^{p}\!\bigl(C_j(\bm z)-C_{-j}(\bm z)\bigr),
\]
where $C_j(\bm z)$ and $C_{-j}(\bm z)$ are the cut values on layers $j$ and $-j$, respectively. Let $\mathrm{MC}$ be a uniform upper bound for the inner-layer difference.

\paragraph{Maximality}
\emph{Goal and idea.} If each difference $C_j(\bm z)-C_{-j}(\bm z)$ can simultaneously attain its individual upper bound, then their sum $\sum_{j=1}^{p}\!\bigl(C_j(\bm z)-C_{-j}(\bm z)\bigr)$ is naturally maximized.

\begin{enumerate}
\item \emph{Inner layers ($1\le j\le p-1$) fixed at extremum.}
For $j\in[1,p-1]$, take
\(
C_j(\bm z)-C_{-j}(\bm z)=\mathrm{MC},
\)
so that the inner-layer contribution is exactly $(p-1)\,\mathrm{MC}$.

\item \emph{Necessary constraints and non-vanishing on the end layers ($j=\pm p$).}
To avoid cancellation on the end layers and to ensure a nonzero coefficient, Lemmas~3.2 and~3.3 allow us to restrict feasible configurations to those satisfying the following necessary conditions (“endpoint consistency / asynchronous flips yield no large exponent”):
\[
\begin{aligned}
&\text{(i)}\ \forall\,u\notin\{x_0,y_0\}:\ z_u^{[p]}=z_u^{[-p]};\\
&\text{(ii)}\ \forall\,u\in\{x_0,y_0\}:\ z_u^{[p]}\neq z_u^{[-p]};\\
&\text{(iii)}\ z_{x_0}^{[p]}=z_{y_0}^{[p]}.
\end{aligned}
\]
Moreover, to maximize $C_p(\bm z)-C_{-p}(\bm z)$ we impose:
\[
\text{(iv)}\ \text{All non-endpoint spins on layer }p\text{ are identical and opposite to }z_{x_0}^{[p]}=z_{y_0}^{[p]}.
\]

\item \emph{Achievable upper bound for the end-layer maximum difference.}
Under (i)–(iv), the entire effective contribution from the end layers is generated by flipping $x,y$ from $p$ to $-p$; via explicit construction and an upper-bound estimate we have
\[
\begin{aligned}
\max\bigl(C_p(\bm z)-C_{-p}(\bm z)\bigr)&=max_{(u,v)\in E}(d(u)+d(v)-2)
\\&=d(x_0)+d(y_0)-2
\\&=200n_0^2+2n_0(11n_0+1)
\end{aligned}
\]
Here, “attainability’’ follows from the explicit arrangement of endpoints and non-endpoints given in the main text; the “upper bound’’ follows from the maximum gain accrued by the number of edges adjacent to $x_0,y_0$ when they flip between the two end layers, subtracting the amplifier edge between $(x_0,y_0)$ to obtain the net value.
\end{enumerate}

In summary, the inner layers contribute $(p-1)\,\mathrm{MC}$ and the end layers contribute $200n_0^2+2n_0(11n_0+1)$. Therefore the attainable upper bound of the largest exponent is
\[
(p-1)\,\mathrm{MC}+200n_0^2+2n_0(11n_0+1),
\]
namely as in (\ref{eq:3.8}).
\paragraph{Non-vanishing}

We prove that the coefficient of the term at the exponent
\[
D_{\max}=\Bigl((p-1)\,\mathrm{MC}+200n_0^2+2n_0(11n_0+1)\Bigr)
\]
is \emph{nonzero}.

\subparagraph{Configurations achieving the largest exponent: necessary and sufficient conditions}

We first characterize the family of configurations that can reach the largest exponent $D_{max}$. The following conditions are \emph{necessary and sufficient}:
\begin{itemize}
  \item[(1)] For \(i\in[1,p-1]\), take \(C_i(\bm z)=\mathrm{MC}\); for \(i\in[-p,-1]\), take \(C_i(\bm z)=0\).
  \item[(2)] On the end layers, \(z_{x_0}^{[p]}=z_{y_0}^{[p]}\neq z_{x_0}^{[-p]}=z_{y_0}^{[-p]}\).
  \item[(3)] For all \(u\in V\setminus\{x_0,y_0\}\), \(z_u^{[p]}=z_u^{[-p]}\).
  \item[(4)] All non-endpoint spins on layer \(p\) are identical and \emph{opposite} to \(z_{x_0}^{[p]}=z_{y_0}^{[p]}\).
\end{itemize}
\noindent\textit{Explanation of necessity and sufficiency.} If any one of the above fails, the differences \(C_j(\bm z)-C_{-j}(\bm z)\) for all \(j\in[1,p]\) cannot jointly attain their upper bounds; conversely, if (1)–(4) hold, then each layer difference can be maximized, thereby achieving the largest exponent \(D_{max}\).

\subparagraph{Ignoring configurations containing a \(\sin\psi\) factor}
If there exist \(u\in V\) and some \(i\in[1,p-1)\cup[-p+1,-1)\) such that \(z_u^{[i]}\neq z_u^{[i+1]}\), then the amplitude of that configuration contains at least one factor \(\sin\psi=2^{O(-n^2)}\), and its contribution is dominated by configurations without \(\sin\psi\). Hence such configurations can be ignored (this does not affect the final non-vanishing conclusion, Because even if there are $2^{O(n)}$ such terms, the sum of their moduli is smaller than the modulus of a single term without the $\sin\psi$ factor; see also below that their sum is nonzero).

\subparagraph{Fixing the value of \(w\): equal amplitude and end-layer structure}
For convenience, first \emph{fix} the configuration of the controller \(w\). Under (1)–(4), only these configurations can reach the largest exponent. We show that the monomial contributions of these configurations have \emph{equal value}.
\begin{itemize}
  \item \emph{Uniform value of inner-layer transitions:} \\
  (i) For \(u\in V\) and \(i\in[1,p-1)\cup[-p+1,-1)\),
\(
  \langle z_u^{[i]}|e^{i\beta_i X}|z_u^{[i+1]}\rangle=\cos\psi.
\)
  \\
  (ii) For \(i\in V\) and \(i=-p\),
\(
  \langle z_u^{[i]}|e^{i\beta_i X}|z_u^{[i+1]}\rangle=\cos\frac{\pi}{4}.
\)
  \\
  (iii) For the forward boundary layer \(i=p-1\), we use the
  Max-Cut structure. Fix the spin of \(w\), and consider only
  configurations that attain the largest exponent.
  By construction of the reduction, these configurations correspond
  to maximum cuts \((S,V\setminus S)\) of \(G'\) satisfying:
  \begin{itemize}
    \item \(x_0\) and \(y_0\) lie on opposite sides of the cut;
    \item we may assume \(w\in V\setminus S\);
    \item each such maximum cut of \(G'\) induces a maximum cut of the
      original instance \(G\), and we refine our choice so that the
      induced cut on \(G\) maximizes \(|S\cap V(G)|\).
  \end{itemize}
  With this tie-breaking, once the spin of \(w\) is fixed, the number
  of vertices \(u\in V(G)\) that lie in \(S\) is the same for every
  configuration achieving the largest exponent.
  Equivalently, the number of vertices \(u\) for which
  \(z_u^{[p-1]}\neq z_u^{[p]}\) is \emph{independent} of the particular
  maximizing configuration.
  At the boundary choice \(\beta_{p-1}=\frac{\pi}{4}\), each one-qubit
  factor \(\langle z_u^{[p-1]}|e^{i\beta_{p-1}X}|z_u^{[p]}\rangle\)
  has value \(\cos\frac{\pi}{4}\) when \(z_u^{[p-1]}=z_u^{[p]}\) and
  value \(i\sin\frac{\pi}{4}\) when \(z_u^{[p-1]}\neq z_u^{[p]}\),
  and the number of ``flip'' transitions is fixed across all maximizers.
  Therefore the product
  \[
    \prod_{u\in V'}
      \bigl\langle z_u^{[p-1]}\big|e^{i\beta_{p-1}X}\big|z_u^{[p]}\bigr\rangle
  \]
  has the same value for all
  configurations that achieve the largest exponent.
  Hence, within the family of configurations that can reach the largest exponent, the inner-layer contribution is identical.
  \item \emph{Uniform value across the three end layers \((p,0,-p)\):}\\
  (i) For any \(u\in V\setminus\{x_0,y_0\}\), the product
\(
  \langle z_u^{[p]}|e^{i\beta_pX}|z_u^{[0]}\rangle\;\langle z_u^{[0]}|e^{-i\beta_pX}|z_u^{[-p]}\rangle
\)
  takes the \emph{same} value across all configurations achieving the largest exponent;\\
  (ii) The endpoint factor
\(
  z_{x_0}^{[0]}z_{y_0}^{[0]}\!\!\prod_{u\in\{x_0,y_0\}}\!\langle z_u^{[p]}|e^{i\beta_pX}|z_u^{[0]}\rangle\;\langle z_u^{[0]}|e^{-i\beta_pX}|z_u^{[-p]}\rangle
\)
  is also \emph{identical}. 
\end{itemize}
It follows that, with \(w\) fixed, the configurations achieving the largest exponent have \emph{equal contributions}.
Moreover, the sum of contributions corresponding to the two choices of \(w\) is also the same; this will be used together with the “places where \(w\) may vary’’ in the next step.

\subparagraph{The only two places that change \(w\) and the four types of contributions}
Consider when the value of \(w\) \emph{changes}. Clearly, only two places can change the spin of \(w\):
\[
\bigl(z_w^{[p-1]},\,z_w^{[p]}\bigr)\quad\text{and}\quad\bigl(z_w^{[-p]},\,z_w^{[-p+1]}\bigr).
\]
Under the properties of configurations that can reach the largest exponent:
\begin{itemize}
  \item The first change swaps the “bits that should flip’’ with those that should not, thereby flipping \(n-r\) bits and introducing a factor \(i^{\,n-r}\) (where \(r\) is the number of bits that should flip).
  \item The second change flips \emph{all} vertices, multiplying the contribution by \(({-}i)^{\,n}\).
\end{itemize}
Accordingly, classifying by whether \(w\) is changed at these two places, the four representative contributions are
\[
\begin{aligned}
&\text{\(\bigl(\)unchanged, unchanged\(\bigr)\)}:\quad i^{\,r}\,T,\qquad
&&\text{\(\bigl(\)unchanged, changed\(\bigr)\)}:\quad -\,i^{\,r+n}\,T,\\
&\text{\(\bigl(\)changed, unchanged\(\bigr)\)}:\quad i^{\,n-r}\,T,\qquad
&&\text{\(\bigl(\)changed, changed\(\bigr)\)}:\quad -\,i^{\,2n-r}\,T,
\end{aligned}
\]
where \(T>0\) is a common real magnitude, and \(r\) is the maximum, over all maximum-cut assignments, of the number of vertices whose configuration differs from that of \(w\). To compare their \emph{phase sum}, multiply all four terms by \(T^{-1}i^{\,3r}\) to obtain
\[
1,\quad -\,i^{\,n},\quad i^{\,n+2r},\quad -\,i^{\,2n+2r}.
\]
We then distinguish cases by the parity of \(r\) and of \(n\) (in our construction, \(n=211n_0^2+n_0+1\) is odd):
\begin{itemize}
  \item \emph{\(r\) odd:} The four terms are \(1,\ -i^{\,n},\ i^{\,n+2r},\ -1\). Using \(i^{\,n+2r}=-\,i^{\,n}\),
\[
  1-1-i^{\,n}-i^{\,n}=-\,2\,i^{\,n}\neq 0,
\]
  hence the phase sum is nonzero.
  \item \emph{\(r\) even:} The four terms are \(1,\ -i^{\,n},\ i^{\,n},\ -i^{\,2n}\). Since \(i^{\,2n}=(-1)^{n}\),
\[
  1-(-1)^{n}+i^{\,n}-i^{\,n}=1-(-1)^{n}=2\neq 0,
\]
  so the phase sum is again nonzero.
\end{itemize}

\subparagraph{Conclusion}
Within the family of configurations achieving the largest exponent \(D_{max}\), each single contribution has \emph{equal magnitude}; the \emph{sum of phases} over the four representative configurations above is \emph{nonzero}. Therefore the coefficient of the Laurent term at exponent \(D\) is \emph{nonzero}, proving “non-vanishing.’’

\medskip
Combining “maximality’’ and “non-vanishing’’ shows that the exponent in (\ref{eq:3.8}) is both attainable and has a nonzero coefficient, so
\(
D_{\max}=\Bigl((p-1)\,\mathrm{MC}+200n_0^2+2n_0(11n_0+1)\Bigr)
\)
is the \textbf{largest nonzero exponent} for $Z_{x_0}Z_{y_0}$, completing the proof of Lemma~3.4.
\qed

\subsection*{C.4 Upper Bound for Other Edges (Lemma 3.5)}
\noindent\textbf{Statement.}
Let the observable be $Z_iZ_j$ with $(i,j)\notin\{(x_0,y_0),(y_0,x_0)\}$. Then all coefficients of the terms in $h_{(i,j)}$ with exponents \(\ge \Bigl((p-1)\,\mathrm{MC}+200n_0^2+2n_0(11n_0+1)\Bigr)\) vanish.

\noindent\textbf{Proof.}
First suppose exactly one endpoint is in $\{x_0,y_0\}$, say the observable is $Z_{x_0}Z_o$ for some $o\in V(G')\setminus\{y_0\}$.
Because 
\(
\sum_{i=1}^{p-1}(C_i(\bm z)-C_{-i}(\bm z))\le (p-1)\mathrm{MC}
\)
 and Lemma~\ref{lem:5.2}, we have
\[
\begin{aligned}
C_p(\bm z)-C_{-p}(\bm z)&\le d(x_0)+d(o)-2
\\&\le n_0(11n_0+1)+1+100n_0^2+100n_0^2-2
\\&< d(x_0)+d(y_0)-2
\end{aligned}
\]
which is at most the degrees of the two endpoints. So
\[
\sum_{i=1}^{p}(C_i(\bm z)-C_{-i}(\bm z))<\Bigl((p-1)\,\mathrm{MC}+200n_0^2+2n_0(11n_0+1)\Bigr).
\]
which means that the term of the highest non-zero degree in the polynomial \( h_{(x_0,o)}(x) \), obtained by taking edges with exactly one of \( (x_0, y_0) \) as an endpoint as observables, is less than $D_{max}$.

If neither endpoint is in $\{x_0,y_0\}$, then
\begin{equation}\label{eq:A4-3}
\sum_{i=1}^{p-1}(C_i(\bm z)-C_{-i}(\bm z))\le (p-1)\mathrm{MC}
\end{equation}
and
\begin{equation}\label{eq:A4-4}
C_p(\bm z)-C_{-p}(\bm z)< d(i)+d(j)\le 200n_0^2.
\end{equation}
Thus the largest exponent attainable is at most
\(
\,\bigl((p-1)\mathrm{MC}+200n_0^2\bigr)
\)
which is less than $D_{max}$ proving the claim.
\qed

\subsection*{C.5 Magnitude of the Extreme-Term Coefficients (Lemma~\ref{lem:3.7})}
\noindent\textbf{Statement.}
Let $D_{\max}$ be the maximal exponent that appears with nonzero
coefficients in the Laurent polynomial $h_{(x_0,y_0)}(x)$ induced by the observable $Z_{x_0} Z_{y_0}$ (cf.\ Sec.~\ref{QAOAhard}).
Then
\begin{equation*}
  |w(D_{\max})| \;\ge\; 2^{-2n-1}.
\end{equation*}

\noindent\textbf{Proof.}
Fix the canonical family of configurations from Lemma~3.4 that attain $D_{\max}$ in \eqref{eq:3.8}:
layers $j=1,\dots,p-1$ take a common maximum cut; layers $j=-1,\dots,-(p-1)$ take a common
minimum cut; and the boundary layers $\{\pm p,0\}$ are set as in that construction. For any such
configuration, the absolute magnitude of its contribution factorizes into (i) the two overlaps with
$|s\rangle=|+\rangle^{\otimes n}$ and (ii) on-site one-qubit transition amplitudes coming from the
mixers:

\smallskip
\emph{(a) Overlaps with $|s\rangle$.} From the complete-basis insertions, the two overlaps contribute
a uniform factor $2^{-n}$.

\smallskip
\emph{(b) $\frac{\pi}{4}$ mixers at the outer boundary $\boldsymbol{\pm p}$.}
With $\beta_{ p}=\frac{\pi}{4}$, for each qubit and for each of the two boundary mixers the one-qubit
transition amplitude has magnitude in $\{\cos(\frac{\pi}{4}),\sin(\frac{\pi}{4})\}=2^{-1/2}$. Hence the total
outer-boundary mixer magnitude is $2^{-n}$.

\smallskip
\emph{(b$'$) $\frac{\pi}{4}$ mixers at the near-boundary transitions $\boldsymbol{p-1\to p}$ and
$\boldsymbol{-p\to -p+1}$.}
Under the parameter choice $\beta_{p-1}=\frac{\pi}{4}$, these two additional near-boundary
mixer blocks contribute another factor $2^{-n}$ in magnitude across $n$ qubits.

\smallskip
\emph{(c) Inner mixers.} Excluding the four $\frac{\pi}{4}$ boundary/near-boundary mixer blocks
$\{\pm p,\,p-1,\, -p+1\}$, each qubit undergoes exactly $2(p-2)$ inner mixer transitions, each
contributing a factor $\cos\psi$, where the parameter choice in Sec.~\ref{QAOAhard} ensures $\cos\psi \ge 1-2^{-n^2}$.
Therefore, across all qubits,
\[
  (\cos\psi)^{\,2n(p-2)} \;\ge\; \bigl(1-2^{-n^2}\bigr)^{2n(p-2)}
  \;\ge\; 1-2n(p-2)\cdot 2^{-n^2} \;\ge\; \tfrac34
\]
for all sufficiently large $n$ (Bernoulli’s inequality and $p\ge 2$).

\smallskip
Multiplying (a), (b), (b$'$), and (c), each fixed configuration contributes at least
$2^{-n}\cdot 2^{-n}\cdot 2^{-n}\cdot \tfrac34 = 3*2^{-3n-2}$ in absolute value.
By Endpoint Consistency (Lemma~3.2) together with the $\frac{\pi}{4}$ phase-pairing argument, the
integrand is insensitive to the entire layer-$0$ string; hence summing over this layer multiplies
the contribution by $2^{n}$. Note that the terms with the $\sin\psi$ factor are equivalent to a small perturbation $\delta\le 2^{O(n)}*2^{O(-n^2)}=2^{O(-n^2)}$ here. Thus
\[
  |w(D_{\max})| \;\ge\; 2^{n}\cdot \frac34 \cdot 2^{-3n}-\delta \;\ge\; 2^{-2n-1}.
\]
\qed

\end{document}